\documentclass{article}

\title{Polynomials Modulo Composite Numbers: Ax-Katz type theorems for the structure of their solution sets}

\ifodd 1
\usepackage[usenames,dvipsnames]{color}

\usepackage[bookmarks=yes,bookmarksopen,bookmarksdepth=2]{hyperref}
\hypersetup{
	colorlinks,
	bookmarksnumbered,
	citecolor=PineGreen,
	linkcolor=Bittersweet,
	urlcolor=RubineRed
	}
\makeatletter
\hypersetup{pdftitle=\@title,pdfauthor={RSKWR}}
\makeatother
\fi

\usepackage{float}
\restylefloat{figure}

\usepackage[margin=2.8cm]{geometry}

\usepackage[numbers]{natbib}
\makeatletter
\renewcommand\bibsection%
{
  \section*{\refname
    \@mkboth{\MakeUppercase{\refname}}{\MakeUppercase{\refname}}}
}
\makeatother

\usepackage{geometry}
\usepackage{listings}

\usepackage{moreverb}
\usepackage{amsmath}
\usepackage{amssymb}

\usepackage{amsthm}
\usepackage{amstext}
\usepackage{float} 
\usepackage{bbm}
\usepackage{enumerate}
\usepackage{enumitem}
\usepackage{mathtools}
\usepackage{subfig}
\usepackage{array}
\usepackage{tikz}
\usetikzlibrary{trees}
\DeclarePairedDelimiter{\ceil}{\lceil}{\rceil}
\DeclarePairedDelimiter{\floor}{\lfloor}{\rfloor}
\DeclarePairedDelimiter{\p}{ \{ }{ \} }
\DeclarePairedDelimiter{\bb}{ ( }{ ) }
\renewcommand{\l}[2]{\multicolumn{2}{c|}{r:#1  n:#2}}
\newcommand{\m}[2]{\multicolumn{2}{c|}{r:#1  n:#2}}
\renewcommand{\r}[2]{\multicolumn{2}{c}{r:#1  n:#2}}
\newcommand{\mc}[1]{\multicolumn{1}{c}{#1}}

\usepackage{braket}
\usepackage{caption}

\makeatletter
\newtheorem*{rep@theorem}{\rep@title}
\newcommand{\newreptheorem}[2]{%
\newenvironment{rep#1}[1]{%
 \def\rep@title{#2 \ref{##1}}%
 \begin{rep@theorem}}%
 {\end{rep@theorem}}}
\makeatother

\newtheorem{Col}{Corollary}

\newtheorem{theorem}{Theorem}

\newreptheorem{theorem}{Theorem}
\newtheorem{lemma}{Lemma}
\newreptheorem{lemma}{Lemma}
\newreptheorem{Col}{Corollary}

\theoremstyle{plain}

\theoremstyle{plain}

\theoremstyle{definition}

\newcommand{\eg}{e.g.\ }
\newcommand{\ie}{i.e.\ }
\newcommand{\s}{\%}

\newcommand{\ZZ}{\mathbb{Z}}
\newcommand{\FF}{\mathbb{F}}
\newcommand{\PP}{\textbf{P}}
\newcommand{\xx}{\textbf{x}}

\newcommand{\TT}{\mathcal{T}}

\newcommand{\wu}{\mathfrak{w}}

\renewcommand{\mu}{\mathfrak{m}}
\newcommand{\md}{\mathbbm{m}}

\allowdisplaybreaks

\definecolor{javared}{rgb}{0.6,0,0} 
\definecolor{javagreen}{rgb}{0.25,0.5,0.35} 
\definecolor{javapurple}{rgb}{0.5,0,0.35} 
\definecolor{javadocblue}{rgb}{0.25,0.35,0.75} 

\usepackage{authblk}

\author{Robert Sur\'owka}
\author{Kenneth W.\ Regan}

\affil{Department of CSE, University at Buffalo,
Amherst, NY 14260 USA
\{robertlu,regan\}@buffalo.edu}

\begin{document}
\sloppy
\maketitle
\begin{abstract}

Marshall and Ramage extended a theorem of Ax from finite fields to finite principal rings, including
the rings $\ZZ_m$ with $m$ composite. We extend their result further by showing additional symmetric structure of the solution spaces. Additionally, for the restricted case of $\ZZ_{2^r}$ and polynomials of degree up
to $2$, we demonstrate even more complex symmetries. Finally, we present experimental results showing solution spaces of polynomials for chosen rings and degrees, to facilitate further hypothesis formulation in this area. 
Polynomials modulo composites are the focus of some computational complexity lower bound frontiers, while those modulo $2^r$ arise
in the simulation of quantum circuits. We give some prospective applications of this line of inquiry.
\end{abstract}

\section{Introduction}
Let $P$ be an $n$-variable polynomial $P: \FF_{p^r}^n \to \FF_{p^r}$ over a finite field $\FF_{p^r}$. The Chevalley-Warning theorem
\cite{Chev35, War36} states that if $n > deg P \geq 1$ then $p$ divides $\#0_P$ (where $\#0_P$ denotes the number of zeros of $P$ in
$\FF_{p^r}$). Ax \cite{Ax64}, using an idea of Dwork \cite{Dwo60}, greatly improved this result, to state that
\[ 
\#0_P
\text{~~is a multiple of~~} p^{r\left(\ceil*{\frac{n}{d}}-1\right)}, 
\] 
where $d$ is the degree of $P$. This result was extended to systems
$\PP$ of $q$-many polynomials $P_i: \FF_{p^r}^n \to \FF_{p^r}$ with respective degrees $d_i$.  Letting $\#0_\PP$ be the number of their
common zeroes, Katz \cite{Kat71} proved that 
\[ 
\#0_\PP \text{~~is a multiple of~~} p^{r\ceil*{\frac{n-\sum_{i=1}^{q} d_i}{max_{i}\p*{d_i} }}}.
\] 
For
a single polynomial this is equivalent to the initial result by Ax. Additionally, the Ax-Katz theorem is known to
be optimal, in regard to the gcd of the cardinalities of the solution sets.

The result we directly build on was obtained by Marshall and Ramage \cite{mar75}. For a polynomial
$P$ over a ring $\ZZ_m$ (or even for any finite, principal ring), where $m =
{p_1^{r_1}p_2^{r_2}\ldots p_{k}^{r_{k}}}$ and all $p_1, p_2, \ldots, p_k$ are different primes, they proved that
\[ 
\#0_P  \text{~~is a multiple of~~}\prod_{i: r_i=1}  p_i^{\ceil*{\frac{n}{d}} -1}\prod_{i:r_i>1} p_i^{\ceil*{\frac{r_{i}n}{2} } -1} .
\]
The above was extended by Daniel Katz \cite{Dkat09} to find the gcd of the numbers of solutions of sets of $q$-many polynomials in $n$ variables. There had been a lot of
additional work done in the area of properties of polynomial solution spaces, especially focused and building from the Ax-Katz theorem, which is by far the most well known.

One of the popular routes was simplifying that theorem's proof. First was an `elementary' proof by Wan \cite{Wan89}, later made especially simple for prime fields \cite{Wan95}. Hou
showed how to obtain the Ax-Katz theorem by direct deduction from Ax's original theorem \cite{Hou05}. A proof requiring probably the least number-theoretic
background was presented by Wilson \cite{Wil06}.
Another results include improvements when the degrees of all variables in monomials are powers of
the characteristic of the field ($p$-weighted degree) \cite{mor95, mor04, cas12}, specializations for so-called general diagonal equations \cite{cao06, cao07}, partial results when
variables with high degrees are ignored or there are isolated variables \cite{cao11, cao12, cas12}, 
divisibility for exponential sums \cite{ado87, mor04}, situations when solutions are specific subspaces of the domain
\cite{hea11}, and many more instances.
Apart from the interest in the divisibility of the numbers of solutions, there is also a research in establishing how large such numbers
need to be, when they are non-zero. First such a result is Warning's Second Theorem \cite{War36}, followed by results of Schanuel
\cite{sch74}, Brink \cite{bri11} and, very recently, Clark, Forrow and Schmitt \cite{cla14}.  This last reference improves the bound and also explores the situation when
variables of polynomials are bounded to subsets of the domain, notably the Boolean cube $\{0,1\}^n$.

\section{Statement of the results}

Our first result applies the proof technique of Marshall and Ramage \cite{mar75} to show an
additional symmetry in the solution space. Taking $\#k_Q$ to be a number of solutions of $Q-k$ we
prove that:

\begin{theorem} \label{generalb}
For any polynomial $Q$ of $n \geq 2 $ variables $\xx$ of degree $d$ over $\ZZ_{m}$, where $m =
{p_1^{r_1}p_2^{r_2}\ldots p_{v}^{r_{v}}}$ and all $p_1, p_2, \ldots, p_v$ are different primes, and any integers $k, w_1, w_2, \ldots, w_v, q_1, q_2, \ldots, q_v$, where $q_i \leq r_i$ it
holds that:
\[
\sum\limits_{i_1=0}^{p_1^{q_1}-1} \sum\limits_{i_2=0}^{p_2^{q_2}-1} \ldots \sum\limits_{i_v=0}^{p_v^{q_v}-1} 
\#\bb*{k{+} \sum\limits_{j=1}^{v}  w_j\frac{m}{p_{j}^{q_j}}i_j}_{Q} \text{~~is a multiple of~~}  
\prod_{i: r_i=1}p_i^{\ceil*{\frac{n}{d}}+q_i-1}\prod_{i:r_i>1} p_i^{\ceil*{\frac{r_{i}n}{2}}+q_i-1}.
\]
\end{theorem}

Using another approach, we obtain a result demonstrating even more symmetries, but restricting both the degree of the polynomial and the ring characteristic to $2$. 

\begin{theorem} \label{main}
(Main Theorem) For any polynomial $Q$ of $n\geq 3$ variables ($\xx, z$) over $\ZZ_{2^r}$ of degree up to $2$, any
integers $q, v\leq r$ and $k, w, g, u$ and any linear polynomial $T(\xx)$, it holds that:
\[
\sum\limits_{i=0}^{2^q-1} \sum\limits_{j=0}^{2^v-1}\#(k{+}w2^{r-q}i{+}g2^{r-v}j)_{Q, z=T(\xx)+u2^{r-v}j}   \text{ is a multiple of }
2^{\ceil*{
\frac{r(n-1)+\min(2v,r)}{2} } +q-1}.
\]
\end{theorem}
  
The properties below easily follow, as we will show in section \ref{s:disMain}.

\begin{Col} \label{maincol}
For any polynomial $Q$ of $n\geq3$ variables ($\xx, z$) over $\ZZ_{2^r}$ of degree up to $2$, and any
integers $q, v \leq r$ and $k, w, g, l$, it holds that:
\begin{enumerate}[label=(\alph*)]
  \item 
  $\sum\limits_{j=0}^{2^v-1} \#k_{Q, z=l+g2^{r-v}j} \text{~~is a multiple of~~}2^{\ceil*{ \frac{r(n-1)+\min(2v,r)}{2} }  -1}$,  $ n \geq 3$
  \item 
  $\sum\limits_{i=0}^{2^q-1} \sum\limits_{j=0}^{2^v-1} \#(k{+}w2^{r-q}i)_{Q, z=l+g2^{r-v}j} \text{~~is a multiple of~~} 2^{\ceil*{
  \frac{r(n-1)+\min(2v,r)}{2} } +q-1}$,  $ n \geq 3$
  \item 
  $ \sum\limits_{j=0}^{2^v-1} \#(k{+}w2^{r-q}j)_{Q, z=l+g2^{r-v}j}   \text{~~is a multiple of~~}2^{\ceil*{ \frac{r(n-1)+\min(2v,r)}{2} } 
  -1}$,  $ n \geq 3$.
\end{enumerate}
\end{Col}

Part (a) says that when one variable is limited to a coset of an ideal it only
moderately decreases the divisibility---while if the coset is at least half of the ring, the divisibility does not deteriorate at all.
Parts (b) and (c), and Theorem \ref{main} overall, show that properties from Theorem \ref{generalb} and point (c) add the same degree of
divisibility even when both properties are present.  This works even in somewhat more general settings.

We should note here, that the proof technique of Marshall and Ramage \cite{mar75} easily allows to prove point (a) of the above corollary
for polynomial of unbounded degree and for $\ZZ_m$, but only when $v=r-1$.


\section{Proof of Theorem \ref{generalb}} \label{s:generalbproof}
\begin{reptheorem}{generalb}
For any polynomial $Q$ of $n \geq 2 $ variables $\xx$ over $\ZZ_{m}$, where $m = {p_1^{r_1}p_2^{r_2}\ldots p_{v}^{r_{v}}}$ and all $p_1,
p_2, \ldots, p_v$ are different primes, and any integers $k, w_1, w_2, \ldots, w_v, q_1, q_2, \ldots, q_v$, where $q_i \leq r_i$ there
is an integer $\TT$ such that:
\[
\sum\limits_{i_1=0}^{p_1^{q_1}-1} \sum\limits_{i_2=0}^{p_2^{q_2}-1} \ldots \sum\limits_{i_v=0}^{p_v^{q_v}-1} 
\#\bb*{k{+} \sum\limits_{j=1}^{v}  w_j\frac{m}{p_{j}^{q_j}}i_j}_{Q} = \TT 
\prod_{i: r_i=1}p_i^{\ceil*{\frac{n}{d}}+q_i-1}\prod_{i:r_i>1} p_i^{\ceil*{\frac{r_{i}n}{2}}+q_i-1}.
\]
\end{reptheorem}
\begin{proof}
We rely on the proof technique of \cite{mar75}. Let us start by proving the hypothesis for a ring
$\ZZ_{p^r}$, where $p$ is prime. If $r=1$ we wish to prove that 
\[
\sum\limits_{i=0}^{p^{q}-1}  \#\bb*{k{+}  wp^{r-q}i}_{Q} = \TT p^{\ceil*{\frac{n}{d}}+q-1},
\]
and it trivially follows from the Ax's theorem. We take now $r\geq2$, for which we intend to prove
\[
\sum\limits_{i=0}^{p^{q}-1}  \#\bb*{k{+}  wp^{r-q}i}_{Q} = \TT p^{\ceil*{\frac{rn}{2}}+q-1}.
\]
Consider
\[
C = \sum\limits_{i=0}^{p^q-1} \#(k{+}wp^{r-q}i)_{Q} = \sum\limits_{i=0}^{p^q-1} \#(wp^{r-q}i)_U
\]
where $U(\xx) = Q(\xx) -k$, 
\[
= p^{min(\wu, q)}\sum\limits_{i=0}^{p^{max(q-\wu, 0)}-1} \#(p^{r-q+\wu}i)_U
\]
where $\wu$ is the order of $w$ (\ie biggest power of $p$ dividing $w$, but $\wu \leq r$). Let
$e= max(q-\wu, 0)$ and let
\[
C' = \sum\limits_{i=0}^{p^e-1} \#(p^{r-e}i)_{U}.
\]
We need now to prove that $C'$ is a multiple of $p^{\ceil*{\frac{rn}{2} } +e-1}$. Note that if $e=r$ the result is trivial,
therefore we assume $e<r$. Additionally, if $e=0$, the result instantly reduces to the theorem by Marshall and Ramage \cite{mar75}. This
allows us to take $0<e<r$. Let 
\[
H(\xx, y)  = U(\xx) + y.  
\]
Then
\[
C' =  \sum\limits_{i=0}^{p^e-1} \#0_{H, y = p^{r-e}i},
\]
because for any assignment to $\xx$ that makes $U(\xx)$ have order at least
$r-e$, there is exactly one assignment to $y$ that evaluates $H$ to $0$.
Let $(x'_1, x'_2,\ldots, x'_{n}, y')$ be a solution of $H$ over $\ZZ_{p^r}$. Let us consider assignments to $H$ of the pattern
\[
H(x'_1+px_1, x'_2+px_2, \ldots, x'_{n-1}+px_{n-1}, y'+ p^{r-e}y),
\]
which then has a form
\[
pG_1 + p^2G_2+\ldots+p^dG_d + p^{r-e}y
\]
where $G_i $ are homogeneous functions of degree $i$ in variables $\xx$. Thus we wish to
count the number of zeroes of
\[
G = G_1 + pG_2+\ldots+p^{d-1}G_d + p^{r-e-1}y
\]
over $\ZZ_{p^{r-1}}$, where additionally $y \in \ZZ_{p^{e}}$ (\ie $0\leq y < p^e$). First let us consider $r=2$. Then $e=1$, and from direct
use of the Ax's theorem on $G$, we obtain divisibility of solutions number by $p^{n}$ (note that $G$
in this case is linear). Let us take now $r \geq 3$. If any of the variables $x_1,\ldots, x_n$ is multiplied by a unit in $G_1$, let $t$ be one of these variables. If not,
then if $r-e-1=0$ let $t=y$.
If $t$ was picked to be one of the variables, let us notice that for any assignment to all other variables in $G$, there is precisely one
assignment to $t$ that solves $G$ (via the main Lemma of \cite{mar75} if $t\neq y$).  Hence $G$ has $p^{(r-1)(n-1)+e}$ solutions.  Let
us compare this exponent with our hypothesis (we can omit the ceiling function, since the left-hand side is integer):
\[
  (r-1)(n-1)+e  \geq \frac{rn}{2} +e-1
\]
\[
  2rn-2r-2n+2e+2  \geq  rn +2e-2
\]
\[
  rn-2r-2n+4  \geq  0  \Leftrightarrow (r-2)(n-2) \geq 0,
\]
which is always true under the theorem's assumptions. Let us assume now that it was impossible to pick $t$, \ie all coefficients in $G_1$ are
divisible by $p$ and $r-e-1 \geq 1$. We take $G_1 = pG'_1 $ and write that
\[
G' = G'_1 + G_2+\ldots+p^{d-2}G_d + p^{r-e-2}y.
\]
The number of zeroes of $G$ over $\ZZ_{p^{r-1}}$, with
constraint on $y$ as earlier, equals the number of zeroes of $G'$ over $\ZZ_{p^{r-2}}$ multiplied by $p^{n}$, 
with unchanged constraint on $y$.
By induction, or by the Ax's result if $r=3$, we obtain that the number of zeroes of $G'$, under the
aforementioned settings, is divisible by $p^{\ceil*{\frac{(r-2)n}{2} } +e-1}$, which multiplied by $p^{n}$ gives the desired divisibility. 

This analysis extends to all the rings $\ZZ_m$ via the decomposition of the ring $\ZZ_m$ into its local rings, in the same
way as applied by Marshall and Ramage \cite{mar75}. Equivalently, an argument using a simple application of Chinese remaindering can be
employed.
\end{proof}

\section{Discussion of proof of Theorem \ref{main}} \label{s:disMain}

We prove this theorem by induction on the number of variables. We state the base case and the induction step of Theorem \ref{main}
separately. Curiously the base case $n=3$ has by far the longer proof, yet while working on it we additionally prove several lemmas of independent interest.
We present both proofs in section \ref{mainproof}.

\begin{reptheorem}{main}
\emph{\textbf{(Base case)}}

For any polynomial $Q$ of $3$ variables ($x, y, z$) over $\ZZ_{2^r}$ of degree up to $2$, any
integers $q, v\leq r$ and $k, w, g, u$ and any linear polynomial $T(x, y)$, there is an integer $\TT$ such that:
\[ 
\sum\limits_{i=0}^{2^q-1} \sum\limits_{j=0}^{2^v-1}\#(k{+}w2^{r-q}i{+}g2^{r-v}j)_{Q, z=T(x, y)+u2^{r-v}j}   = \TT2^{r+ \ceil*{
\frac{\min(2v,r)}{2}} +q-1}
\]
\end{reptheorem}

\begin{reptheorem}{main}
\emph{\textbf{(General induction step)}}

Let any quadratic polynomial $Q$ of $n\geq 4$ variables ($\xx, z$) over $\ZZ_{2^r}$, and any
integers $q, v\leq r$ and $k, w, g, u$, and any linear polynomial $T(\xx)$ be given. Suppose that for any $Q'$ of $n-1$ variables ($\xx',
z'$), any $q', v'\leq r$ and $k', w', g', u'$, and any linear polynomial $T'(\xx')$ it holds that:
\[ 
\sum\limits_{i=0}^{2^{q'}-1} \sum\limits_{j=0}^{2^{v'}-1}\#(k'{+}w'2^{r-q'}i{+}g'2^{r-v'}j)_{Q', z'=T'(\xx')+u'2^{r-v'}j}   = \TT'2^{\ceil*{
\frac{r(n-2)+\min(2v',r)}{2} } +q'-1}
\]
for certain integer $\TT'$. Then there is an integer $\TT$ such that 
\[ 
\sum\limits_{i=0}^{2^q-1} \sum\limits_{j=0}^{2^v-1}\#(k{+}w2^{r-q}i{+}g2^{r-v}j)_{Q, z=T(\xx)+u2^{r-v}j}   = \TT2^{\ceil*{
\frac{r(n-1)+\min(2v,r)}{2} } +q-1}.
\]
\end{reptheorem}

We apply the induction hypothesis for $n-1$ where $z'$ is a different variable from the $z$ in the goal statement for $n$. In
particular $z'$ becomes the variable in $\xx$ whose coefficient in a certain linear functional inside of $Q$ has the least order. 

In the general
induction step we show that if $Q$ has no term of type $z^2$ then for it the theorem holds basing on induction hypothesis for $n-1$
variables. Then we prove that adding the square term for $z$ does not change the divisibility lower bound. 

The form with two summation signs is very general, but we found that by our approach of inducting on $n$ even obtaining a simple statement
about the divisibility of $\#0_Q$ requires them, else our induction does not close. The following corollary indicates statements one can
obtain by substituting for the more general quantities. 

\begin{repCol}{maincol}
 For any polynomial $Q$ of $n$ variables ($\xx, z$) over $\ZZ_{2^r}$ of degree up to $2$, and any
integers $q, v \leq r$ and $k, w, g, l$, it holds that:
\begin{enumerate}[label=\alph*)]
\item 
$\sum\limits_{j=0}^{2^v-1} \#k_{Q, z=l+g2^{r-v}j} = \TT_c2^{\ceil*{ \frac{r(n-1)+\min(2v,r)}{2} }  -1}$,  $ n \geq 3$
\item 
$\sum\limits_{i=0}^{2^q-1} \sum\limits_{j=0}^{2^v-1} \#(k{+}w2^{r-q}i)_{Q, z=l+g2^{r-v}j} = \TT_d2^{\ceil*{ \frac{r(n-1)+\min(2v,r)}{2} } 
+q-1}$,  $ n \geq 3$
\item 
$ \sum\limits_{j=0}^{2^v-1} \#(k{+}w2^{r-q}j)_{Q, z=l+g2^{r-v}j}   = \TT_e2^{\ceil*{ \frac{r(n-1)+\min(2v,r)}{2} }  -1}$,  $ n \geq 3$
\end{enumerate}
for certain integers $\TT_a, \TT_b, \TT_c$.
\end{repCol}
\begin{proof}
All above are special cases of Theorem \ref{main}:
\begin{enumerate}[label=\alph*)]
\item 
Take $q=0$, $g=0$ and $T(\xx) = l$.
\item 
Take $g=0$ and $T(\xx) = l$.
\item 
Take $q=0$ and $T(\xx) = l$.
\end{enumerate}
\end{proof}

\section{Experiments} \label{ch:exp}

Before we were able to get a feel of the behaviour of the polynomial solutions over rings, in order to formulate our theorem's statements, we
had to see first some examples. Then we were able to extrapolate from them the general properties. In this chapter we present sample computer
programs we wrote and some results we obtained, which allowed us to probe this area. Those programs are extremely simple, yet they were
only a mundane means for the general formulations of the properties. We present their actual Java code instead of a
pseudocode, so that anyone interested can directly copy and run them himself, possibly with different parameters than the ones presented
here.
The most basic version of the programs we run is presented in the form of the code \ref{code1}. Variables \emph{ring} and \emph{vars\_num}
are set to the size of the ring and the number of variables of the polynomial respectively. The program outputs, for each possible number of
solutions, how many polynomials of the given number of variables over the given ring (and with degree up to $2$), have that many solutions.
For example ``$0$:$~~80$'', means there are $80$ polynomials that are unsolvable, whereas ``$3$:$~~702$'' means there are $702$
polynomials having $3$ solutions.
The code iterates over all possible polynomials, by iterating over all possible coefficients of the terms. For each such polynomial, the
code iterates over all possible assignments to the polynomial's variables, and records what the polynomial evaluates to. All those results
are combined to produce the final output. 

 \begin{minipage}{\textwidth}

\vspace{-1cm}
{\small
\begin{lstlisting}[caption={Degree 2, iterates over all polynomials.}, label=code1]

public class Solutions {

	public static boolean next(int[] arr, int ring) {
		for (int index = arr.length - 1; index >= 0; --index) {
			if (arr[index] < ring - 1) {
				++arr[index];
				return true;
			} else {
				arr[index] = 0;
			}
		}
		return false;
	}

	public static void main(String[] args) {
		int ring = 3;
		int vars_num = 3;	
		int coef_num = vars_num + (vars_num * (vars_num + 1)) / 2;		
		int[] coefs = new int[coef_num];
		int[] vars = new int[vars_num];
		long[] out = new long[(int) Math.pow(ring, vars_num) + 1];
		int[] ringout = new int[ring];
		coefs[coefs.length - 1] = -1;

		while (next(coefs, ring)) {
			for (int i = 0; i < ring; ++i) {
				ringout[i] = 0;
			}
			for (int i = 0; i < vars_num; ++i) {
				vars[i] = 0;
			}
			vars[vars.length - 1] = -1;		
			while (next(vars, ring)) {
				int result = 0;
				for (int i = 0; i < vars_num; ++i) {
					result += coefs[i] * vars[i];
				}
				int off = vars_num;
				for (int i = 0; i < vars_num; ++i) {
					for (int j = i; j < vars_num; ++j) {
						result += coefs[off++] * vars[i] * vars[j];
					}
				}
				++ringout[result % ring];
			}
			for (int sol_num : ringout) {
				++out[sol_num];
			}
		}
		for (int i = 0; i < out.length; ++i) {
			System.out.println(i + ":  " + out[i]);
		}
	}
}
\end{lstlisting}
}
\end{minipage}

\vfill

\restoregeometry
\newgeometry{top=2cm} 
\newgeometry{left=1.8cm, top=2cm}
\begin{table}
\thispagestyle{empty}
\begin{center}
\begin{tabular}{ll|ll|ll|ll|ll|ll}
\l{2}{1}        &  \m{3}{1}        &  \m{4}{1}        &  \m{5}{1}         &  \m{6}{1}       &  \r{6}{4}                             \\
0 &  2          &  0 &  8          &  0 &  26         &  0 &  44          &  0 &  102       &  0 &  237508496                                       \\
1 &  4          &  1 &  12         &  1 &  16         &  1 &  40          &  1 &  48        &  36 &  15724800                                      \\
2 &  2          &  2 &  6          &  2 &  20         &  2 &  40          &  2 &  48        &  54 &  50319360                                       \\                      \cline{1-2} 
\l{2}{2}        &  3 &  1          &  4 &  2          &  5 &  1           &  3 &  4         &  72 &  1292803200                           \\                             \cline{3-4}   \cline{5-6}  \cline{7-8}  
0 &  4          &  \m{3}{2}        &  \m{4}{2}        &  \m{5}{2}         &  4 &  12        &  84 &  3056901120                                  \\
1 &  16         &  0 &  26         &  0 &  484        &  0 &  244         &  6 &  2         &  90 &  50319360                                  \\                            \cline{9-10} 
2 &  24         &  1 &  54         &  2 &  768        &  1 &  1000        &  \m{6}{2}       &  96 &  7642252800                                        \\
\end{tabular} 
 \caption{Degree 2, rings 1 to 6.}\label{tab:d2r1-6}                                              
\end{center}

\end{table}

\newgeometry{top=2cm, left=2.5cm}
 
\begin{table}
\thispagestyle{empty} 
\begin{center}                   
  \begin{tabular}{ll|ll|ll|ll|ll}
\l{7}{1}        &      \m{8}{3}          &  \m{10}{2}         &  \m{11}{1}          &  \r{12}{1}                         \\
0 &  132        &      0 &  16979592     &  0 &  77140        &  0 &  560           &  0 &  1006                       \\
1 &  84         &      16 &  3670016     &  1 &  16000        &  1 &  220           &  1 &  192                          \\
2 &  126        &      32 &  217333760   &  2 &  24000        &  2 &  550           &  2 &  336                                \\
7 &  1          &      48 &  121110528   &  3 &  16000        &  11 &  1            &  3 &  16                             \\       \cline{1-2}      \cline{7-8}
\l{7}{2}        &      64 &  400357888   &  4 &  100000       &  \m{11}{2}          &  4 &  144                               \\
0 &  1014       &      80 &  69730304    &  5 &  42240        &  0 &  6610          &  6 &  20                             \\          
1 &  6174       &      96 &  204603392   &  6 &  64000        &  1 &  66550         &  8 &  12                            \\
\end{tabular}                                                            
                                                                                         
\end{center}                                                                                                                                             
\caption{Degree 2, rings 7 to 12.}\label{tab:d2r7-12}                                              
\end{table}                                                              
                                                      
\newgeometry{top=1.5cm, left=2.5cm} 
\begin{table}                                                                                               
\thispagestyle{empty}               
\begin{center}                      
  \begin{tabular}{ll|ll|ll|ll|ll}   
 \l{13}{1}              &  \m{14}{1}            &  \m{15}{1}     &  \m{15}{3}          &  \r{16}{3}                          \\
 0 &  948               &  0 &  1478            &   0 &  1836    &  0 &  854607436     &  0 &  18502424200                  \\
 1 &  312               &  1 &  336             &   1 &  480     &  15 &  21762000     &  64 &  2113929216                  \\
 2 &  936               &  2 &  672             &   2 &  720     &  30 &  478764000    &  128 &  222744281088                   \\
 13 &  1                &  4 &  252             &   3 &  40      &  45 &  846300000    &  192 &  124721823744                   \\      \cline{1-2}
 \l{13}{2}              &  7 &  4               &   4 &  240     &  60 &  2742012000   &  256 &  409111248896                    \\
 0 &  13116             &  14 &  2              &   5 &  12      &  75 &  2272823280   &  288 &  30064771072                    \\                       \cline{3-4}
 1 &  158184            &  \m{14}{3}            &   6 &  40      &  90 &  2265666000   &  320 &  12448694272                               \\
\end{tabular}                                                    
                                                                 
\end{center}                                                     
\caption{Degree 2, rings 13 to 17.}\label{tab:d2r13-17}          
\end{table}                                                                                           
                                                                                           
\newgeometry{top=1.5cm}               
\begin{table}                       
\thispagestyle{empty}               
\begin{center}                      
  \begin{tabular}{ll|ll|ll|ll|ll}   
\l{2}{1}	      &   \m{3}{3}             &   \m{5}{1}		  &  \m{6}{1}		 &  \r{8}{2}       	    \\           
0 &  4            &   0 &  52704           &   0 &  204       &  0 &  612        &  0 &  110516800             \\     
1 &  8            &   1 &  892296          &   1 &  260       &  1 &  288        &  4 &  304971776             \\      
2 &  4            &   2 &  5117580         &   2 &  120       &  2 &  288        &  8 &  332768256              \\     \cline{1-2} 
\l{2}{2}          &   3 &  22561578        &   3 &  40        &  3 &  24         &  12 &  197853184           \\         
0 &  64           &   4 &  69599088        &   5 &  1         &  4 &  72         &  16 &  75890944           \\          \cline{5-6}
1 &  256          &   5 &  149433336       &   \m{5}{2}       &  6 &  12         &  20 &  36962304           \\                         \cline{7-8}
2 &  384          &   6 &  283362300       &   0 &  37204     &  \m{6}{2}        &  24 &  10127360                \\
\end{tabular}                                                                                             
\end{center}                                                                                                
\caption{Degree 3, rings 2 to 9.}\label{tab:d3r2-9}                                                       
\end{table}                                                                                                                                                                                                     
\restoregeometry

\newgeometry{top=1.5cm, left=2.5cm}            
\begin{table}                    
\thispagestyle{empty}            
\begin{center}                   
  \begin{tabular}{ll|ll|ll|ll|ll}   
\l{10}{1}			&    \l{11}{1}         & \m{12}{1}        &  \m{13}{1}         &  \r{14}{2}                 \\    
0 &  4948           &  0 &  4960           & 0 &  11844       &  0 &  9684         &  0 &  18305826496                       \\ 
1 &  2080           &  1 &  6380           & 1 &  2880        &  1 &  12636        &  1 &  161742336                   \\  
2 &  2000           &  2 &  1650           & 2 &  3456        &  2 &  2808         &  2 &  1178295552                    \\   
3 &  320            &  3 &  1650           & 3 &  816         &  3 &  3432         &  3 &  4199505408                   \\     
4 &  480            &  11 &  1             & 4 &  1152        &  13 &  1           &  4 &  7968444288                     \\        \cline{3-4}    \cline{7-8}
5 &  8              &  \l{11}{2}           & 6 &  456         &  \m{13}{2}         &  5 &  8775180288                    \\     
6 &  160            &  0 &  5383410        & 8 &  72          &  0 &  17852028     &  6 &  17761099776                      \\
\end{tabular}                                                                   
                                                                                 
\end{center}                                                                     
\caption{Degree 3, rings 10 to 14.}\label{tab:d3r10-14}                          
\end{table}                                                     
\restoregeometry                                                

\newgeometry{top=4.5cm, left=3.0cm}
\begin{table}                                    
\begin{center}                     
  \begin{tabular}{ll|ll|ll|ll}  
\l{15}{1}             &    \m{15}{2 (cont.)} &  \m{16}{1}          &  \r{17}{1}            \\            
0 &  26628            &    33 &  598752000   &  0 &  29668         &  0 &  28304                     \\      
1 &  9360             &    35 &  8944966080  &  1 &  22528         &  1 &  37808                    \\       
2 &  9000             &    36 &  6080904000  &  2 &  7488          &  2 &  6528                     \\       
3 &  2220             &    39 &  133056000   &  4 &  3152          &  3 &  10880                 \\          
4 &  2160             &    40 &  5097556800  &  5 &  2048          &  17 &  1                   \\                   \cline{7-8}
5 &  36               &    42 &  2678400000  &  6 &  128           &  \r{17}{2}               \\            
6 &  1080             &    44 &  408240000   &  8 &  440           &  0 &  121251088                \\
\end{tabular}                                                                   
                                                                                 
\end{center}                                                                     
\caption{Degree 3, rings 15 to 17.}\label{tab:d3r15-17}                          
\end{table}                                                     
\restoregeometry 

In Tables \ref{tab:d2r1-6} to \ref{tab:d2r13-17} we present some outputs of such a program, whereas in Tables \ref{tab:d3r2-9} to
\ref{tab:d3r15-17} we present outputs of an analogous program, yet with allowing the polynomials to have degree up to $3$. Each section of
those tables has a value $r$ denoting size of the ring, and a value $n$ which denotes number of variables. For example $3$rd line in section
for $r:3~~n:2$ of table \ref{tab:d2r1-6} says that there are $216$ polynomials of degree up to $2$, over the ring $\ZZ_3$ with $2$
variables, that have precisely $2$ solutions. In the tables we omitted those numbers of solutions for which there are no polynomials that have that
many solutions. Additionally, to obtain those results we used a little more advanced programs, that iterated only once over many
isomorphic polynomials, and, most importantly, were multi-threaded. We don't present their code here, as it is long and doesn't add
much to this discussion.
There are many interesting things that we can notice in the tables below. Because the number of solutions of any such polynomial has to
be divisible by a certain constant, that number can occupy only one of the allowed ``slots''. Yet, as we see, for a lot of those slots
there are no polynomials that have that number of solutions. It also often happens, that several of the slots directly following
the $0$-slot are not taken. For example, in the section of Table \ref{tab:d2r1-6} for $r:6~~n:4$, the divisibility is $6$, yet the
smallest possible non-zero number of solutions is $36$, which shows that this initial gap can be significantly larger than the divisibility.
The earlier mentioned research by Clark, Forrow and Schmitt \cite{cla14} is focused on counting the size of this first gap. It is important
to mention that the gaps in the second half of the spectrum tend to be even larger. In the example we just looked at, half of the range is $1296/2 =
648$, and there are only $5$ slots ``taken'' after that half, whereas there are $39$ taken slots up to that half. Additionally, the last
gap, that is the difference between the two largest possible numbers of solutions, seems to be consistently the largest in all examples.
This may be connected to the fact that also $k$-CNF has worse granularity on the number of satisfying assignments when the set of satisfying
assignments is large. For example, a $3$-CNF formula over $n$ variables cannot have more than $7\times2^{n-3}$ satisfying assignments,
unless it is a tautology. We directly present the sizes for the first and last gaps that we obtained in our experiments in Tables
\ref{tab:d2 front gap} to \ref{tab:d3 front gap} that are close to the end of this chapter (columns go by the number of variables, and rows
by the size of the ring).
Another thing to notice is that, for example, each of the sections $r:2~~n:4$, $r:4~~n:4$, $r:6~~n:4$ from Table \ref{tab:d2r1-6} has
certain numbers of solutions for which there are exactly $7168$ and $2240$ polynomials having that many solutions. There are also
noticeable cases when number of polynomials having particular number of solutions is very small, when compared to neighbouring numbers. We
can see this for example in Table \ref{tab:d3r10-14} for $r:10~~n:4$ for $25$ solutions, and for $r:14~~n:4$ in line for $49$ solutions. We
are sure that there are many other properties waiting to be noticed, and we contribute the above tables to facilitate future research and
heuristic formulation in this area.

For phenomena we would like to especially focus our attention on, we provide Tables \ref{tab:d2 con div} to \ref{tab:d2 per slots used}
for polynomials of degree up to $2$, and Tables \ref{tab:d3 con div} to \ref{tab:d3 per slots used} for polynomials of degree up to $3$. In
all those tables the columns refer to number of variables of the polynomial, while the rows
represent the size of the ring. For example, from Table \ref{tab:d2 con div} we can read that a polynomial of degree up to $2$ over $\ZZ_{12}$ with $4$
variables must have its number of solutions be a multiple of $24$; this follows by
the theorem of \cite{mar75}. Table \ref{tab:d2 per min div} says for what percent of all polynomials their number of solutions is divisible by the minimum divisibility, and not by any higher
power of the size of the ring. For example over $\ZZ_2$, half of the polynomials of $2$ variables of degree up to $2$ have the minimum
divisibility of their solutions numbers (those are the polynomials that have $1$ or $3$ solutions in this case). Table \ref{tab:d2 slots
used} tells how many different numbers of solutions a polynomial over a given ring and number of variables may have. For example,
polynomials over $\ZZ_{11}$ with $2$ variables (and still with degree up to $2$), may only have one of $8$ different solution numbers - in
this case $0, 1, 10, 11, 12, 21, 22,$ and $121$. Finally, the Table \ref{tab:d2 per slots used} gives the division, of the number of
possible solutions numbers from Table \ref{tab:d2 slots used} by the number of solutions numbers allowed by the minimum divisibility. For
example over $\ZZ_5$ with $4$ variables there is $12$ possible solution numbers, whereas due to divisibility by $5$, $125+1=126$ slots are
allowed.
Then, $12/126$ gives the $9.5$\s$~$that we find in the table. Tables \ref{tab:d3 con div} to \ref{tab:d3 per slots used} are
respectively analogous, but for polynomials of degree up to $3$.

Having introduced the tables given below, let us discuss what we can learn from them. Tables \ref{tab:d2 con div} and
\ref{tab:d3 con div} illustrate how much quicker the divisibility of solution numbers rises when
we work over rings that include at least a second prime power, compared to when we work over fields.
It is especially visible for the ring of size $16$, in marked contrast to the ring $\ZZ_{15}$ and the field $\ZZ_{17}$. We see that the extension of the Ax's theorem to non-field rings gives
much greater divisibilities than the original Ax theorem does. This divisibility gap grows even
larger as we increase any of the following parameters: ring size, the number of variables, or the degree of the polynomial. It is also interesting to have a look at \eg $\ZZ_{12}$, where we see that the divisibility
of its solution numbers is a multiplication of divisibilities for $\ZZ_3$ and $\ZZ_4$ for the same $n$. 

Tables \ref{tab:d2 per min div} and \ref{tab:d3 per min div} show us that a very large part of all the polynomials have the minimum
divisibility allowed. This means, that if we would pick several polynomials at random, there is a
high chance that at least one of them would have the minimum divisibility of its number of solutions. This is a very important observation;
we based on it another program that we used, which will be shortly presented.

From Tables \ref{tab:d2 slots used} and \ref{tab:d3 slots used}, we learn that the numbers of numbers of solutions the polynomials may have
are surprisingly small. It would not be surprising at all, if those numbers of the ``used slots'' were bounded by a polynomial in the size
of the ring $r$ and the number of variables $n$. For degree $2$ it may even be a polynomial like $rn^k$, where $k$ is the number of prime
divisors of $r$. Let us notice that for degree $2$, for fields other than $\ZZ_2$, the numbers of used slots seem to be exactly the same,
and they grow by $4$ at every second increase in number of variables. Additionally, it seems that if for certain $\ZZ_x$ and $\ZZ_y$ the 
number of slots used don't change between some numbers of variables $a$ and $b$, then also $\ZZ_{xy}$, has the same number of slots used
for $n = a$ and $n = b$. For example for $\ZZ_2$, $\ZZ_5$, we consider $\ZZ_{10}$ with $n=2$ and $3$. This observation also strongly
suggests that the number of slots used for $\ZZ_6$ and $n=5$ is $44$, even though we didn't run an experiment which would confirm that.

The last pair of tables, \ref{tab:d2 per slots used} and \ref{tab:d3 per slots used}, show that the fraction of slots allowed by minimum
divisibility that are actually used, nearly always decreases with increasing number of variables. Even though the number
of allowed slots increases exponentially, this strong decrease in how many of them are used gives hope that the ultimate
number of used slots is actually only polynomially large. This has especially high probability of being true for fields and for rings over prime
powers. 

Other properties in this data seem to be worthy of further investigation.
Many of them may be just mathematical curiosities, yet some may play a key role in understanding the shape of polynomial solutions spaces,
and be fundamental in future results in low circuit complexity (especially $\mathsf{ACC^0}[m]$), and also classical simulation of quantum
circuits via polynomials. 

The type of program that we used in the end to test our hypotheses is presented as the code
\ref{code2}. The user chooses \emph{ring} size, \emph{vars\_num} as the number of variables, \emph{div} as the hypothetical divisibility of
the number of solutions, and \emph{tries} as the number of ``tries'' to check the hypothesis. The program counts numbers of solutions of
\emph{tries}-many randomly chosen polynomials of degree up to $3$ over the given \emph{ring} with the given number of variables. When the number of solutions
of a checked polynomial is a multiple of \emph{div} it passes silently, otherwise a remainder of the division of number of solutions by
\emph{div} is printed. The code is presented with \emph{ring}=$8$, \emph{vars\_num}=$6$, and \emph{div}=$512$. The given divisibility is too
large, as it should be only $256$, therefore upon running this code we should see information of multiple remainders of
$256$, and the run won't pass. We are highly likely to see polynomials that don't have number of their solutions divisible by $512$, even
with the very small number of $50$ \emph{tries}. It is, because the fraction of polynomials that have minimum divisibility is very significant,
as also tables \ref{tab:d2 per min div} and \ref{tab:d3 per min div} show. It may be true that when ring, degree and number of variables
increase those fractions significantly decrease, yet even then most probably they still are not minuscule, and in our real experiments we
were using large numbers of tries.

\clearpage
\newgeometry{top=3cm}
\pagestyle{plain}

{\Large
\begin{table}
\parbox{.40\linewidth}{
\centering
\begin{tabular}{l|rrrrrr}
r\textbackslash n			   &    1     &      2     &      3       &       4      &       5       &       6            \\        
\hline                                                                                                       \\   
2                              &    1     &      1     &      2       &       2      &       4       &         4                \\                
3                              &    1     &      1     &      3       &       3      &       9       &         9                   \\               
4                              &    1     &      2     &      4       &       8      &       16      &         32          \\                    
5                              &    1     &      1     &      5       &       5      &       25      &         25                 \\                 
6                              &    1     &      1     &      6       &       6      &       36      &         36                \\              
7                              &    1     &      1     &      7       &       7      &       49      &         49       \\
\end{tabular}
\caption{Degree 2, minimum divisibility of number of solutions.} \label{tab:d2 con div}
}
\hspace{1cm}
\parbox{.55\linewidth}{
\centering

\begin{tabular}{l|rrrrrr}                                                                                                                       
r\textbackslash n			   &   \mc{1} &    \mc{2}  &    \mc{3}    & \mc{4}      &    \mc{5}    &  \mc{6}              \\                 
\hline                                                                                                       \\                                 
2                              &   50.0\s &   50.0\s   &   43.8\s     &    43.8\s    &     42.4\s    &     42.4\s               \\                 
3                              &   66.7\s &   66.7\s   &   53.5\s     &    64.2\s    &     53.3\s    &                         \\               
4                              &   25.0\s &   25.0\s   &   21.9\s     &    21.9\s    &               &                  \\                      
5                              &   64.0\s &   80.0\s   &   67.5\s     &    79.4\s    &               &                         \\               
6                              &   51.9\s &   64.6\s   &   62.4\s     &    62.9\s    &               &                        \\                
7                              &   61.2\s &   85.7\s   &   75.3\s     &              &               &               \\
\end{tabular}                                                                                                                                   
\caption{Degree 2, percent of polynomials with number of solutions with minimum divisibility.} \label{tab:d2 per min div}                          
} 
                                                                                                                                              
\end{table}

\begin{table}
\parbox{.40\linewidth}{
\centering
\begin{tabular}{l|rrrrrr}
r\textbackslash n			   &    1     &      2     &      3       &       4      &       5       &       6               \\        
\hline                                                                                                       \\   
2                              &    3     &      5     &      5       &       7      &       7       &       9               \\                
3                              &    4     &      8     &      8       &       12     &       12      &                         \\               
4                              &    4     &      7     &      9       &       16     &               &                  \\                    
5                              &    4     &      8     &      8       &       12     &               &                         \\                 
6                              &    6     &      18    &      18      &       44     &               &                        \\              
7                              &    4     &      8     &      8       &              &               &               \\
\end{tabular}
\caption{Degree 2, number of possible numbers of solutions.\vspace{2cm}} \label{tab:d2 slots used}
}
\hspace{1cm}
\parbox{.55\linewidth}{
\centering

\begin{tabular}{l|rrrrrr}                                                                                                                       
r\textbackslash n			   &   \mc{1} &    \mc{2}  &    \mc{3}    & \mc{4}      &    \mc{5}    &  \mc{6}                \\                 
\hline                                                                                                                         \\                                 
2                              &   100.0\s  &   100.0\s  &   100.0\s    &    77.8\s    &     77.8\s    &     52.9\s               \\                 
3                              &   100.0\s  &   80.0\s   &   80.0\s     &    42.9\s    &     42.9\s    &                         \\               
4                              &   80.0\s   &   77.8\s   &   52.9\s     &    48.5\s    &               &                  \\                      
5                              &   66.7\s   &   30.8\s   &   30.8\s     &    9.5\s     &               &                         \\               
6                              &   85.7\s   &   48.6\s   &   48.6\s     &    20.3\s    &               &                        \\                
7                              &   50.0\s   &   16.0\s   &   16.0\s     &              &               &               \\
\end{tabular}                                                                                                                                   
\caption{Degree 2, percent of solution-number slots allowed by minimum divisibility, that are used.
\vspace{2cm}}  \label{tab:d2 per slots used}                         
} 
                                                                                                                                              
\end{table}
}
\vfill
\clearpage

\begin{table}
\parbox{.45\linewidth}{
\centering
\begin{tabular}{l|rrrrrr}
r\textbackslash n			   &    1     &      2     &      3       &       4      &       5       &       6               \\        
\hline                                                                                                                  \\   
2                              &    1     &      1     &      1       &       2      &      2       &         2                \\                
3                              &    1     &      1     &      1       &       3      &      3       &         3                   \\               
4                              &    1     &      2     &      4       &       8      &      16      &         32          \\                    
5                              &    1     &      1     &      1       &       5      &      5       &         5                 \\                 
6                              &    1     &      1     &      1       &       6      &      6       &         6                \\              
7                              &    1     &      1     &      1       &       7      &      7       &         7       \\
\end{tabular}
\caption{Degree 3, minimum divisibility of number of solutions.} \label{tab:d3 con div}
}
\hspace{1cm}
\parbox{.55\linewidth}{
\centering

\begin{tabular}{l|rrrr}                                                                                                                       
r\textbackslash n			   &  \mc{1}  &   \mc{2}   &   \mc{3}    &    \mc{4} \\  
\hline                                                                           \\  
2                              &   50.0\s &   50.0\s   &   50\s     &    49.6\s  \\ 
3                              &   66.7\s &   66.7\s   &   66.7\s   &            \\ 
4                              &   37.5\s &   49.2\s   &   49.5\s   &            \\ 
5                              &   67.2\s &   77.7\s   &            &            \\ 
6                              &   51.9\s &   66.0\s   &            &            \\ 
7                              &   66.5\s &   83.2\s   &            &            \\
\end{tabular}                                                                                                                                   
\caption{Degree 3, percent of polynomials with number of solutions with minimum divisibility.}  \label{tab:d3 per min div}                         
} 
                                                                                                                                              
\end{table}

\begin{table}
\parbox{.45\linewidth}{
\centering
\begin{tabular}{l|rrrr}
r\textbackslash n			   &    1     &      2     &      3       &       4    \\  
\hline                                                                             \\  
2                              &    3     &      5     &      9       &       9    \\ 
3                              &    4     &      9     &     22       &            \\ 
4                              &    5     &      8     &     15       &            \\ 
5                              &    5     &      16    &              &            \\ 
6                              &    6     &      22    &              &            \\ 
7                              &    5     &      20    &              &            \\
\end{tabular}
\caption{Degree 3, number of possible numbers of solutions.\vspace{2cm}} \label{tab:d3 slots used}
}
\hspace{1cm}
\parbox{.55\linewidth}{
\centering

\begin{tabular}{l|rrrr}                                                                                                                       
r\textbackslash n			   &  \mc{1}  &   \mc{2}   &   \mc{3}     &    \mc{4}   \\  
\hline                                                                             \\  
2                              &   100.0\s  &   100.0\s  &   100.0\s  &    100.0\s  \\ 
3                              &   100.0\s  &   90.0\s   &   78.6\s   &            \\ 
4                              &  100.0\s   &   88.9\s   &   88.2\s   &            \\ 
5                              &   83.3\s   &   61.5\s   &            &            \\ 
6                              &   85.7\s   &   59.5\s   &            &            \\ 
7                              &   62.5\s   &   40.0\s   &            &            \\
\end{tabular}                                                                                                                                   
\caption{Degree 3, percent of solution-number slots allowed by minimum divisibility, that are
used.\vspace{2cm}}
\label{tab:d3 per slots used} } 
                                                                                                                                              
\end{table}

\clearpage
\begin{table}
\parbox{.40\linewidth}{
\centering
\begin{tabular}{l|rrrrrr}
r\textbackslash n			   &    1     &      2     &      3       &       4      &       5       &       6            \\        
\hline                                                                                                       \\   
2                              &    1     &      1     &      2       &       4      &       8       &         16                \\                
3                              &    1     &      1     &      3       &       9      &       27      &                            \\               
4                              &    1     &      2     &      8       &       16     &               &                    \\                    
5                              &    1     &      1     &      5       &       25     &               &                         \\                 
6                              &    1     &      1     &      6       &       36     &               &                        \\              
7                              &    1     &      1     &      7       &              &               &               \\
\end{tabular}
\caption{Degree 2, size of the first gap between solution numbers.} \label{tab:d2 front gap}
}
\hspace{1cm}
\parbox{.55\linewidth}{
\centering

\begin{tabular}{l|>{$}r<{$} >{$}r<{$} >{$}r<{$} >{$}r<{$} >{$}r<{$} >{$}r<{$}}                                                                                                                       
r\textbackslash n			   &  \mc{1}  &   \mc{2}   &    \mc{3}    &     \mc{4}   &    \mc{5}     &      \mc{6}            \\                  
\hline                                                                                                       \\                               
2                              &    1     &      1     &      2       &       4      &       8       &         16                \\           
3                              &    1     &      3     &      9       &       27     &       81      &                            \\          
4                              &    2     &      4     &      16      &       64     &               &                    \\                  
5                              &    3     &      15    &      75      &  3\cdot5^3   &               &                         \\             
6                              &    2     &      9     &   2\cdot3^3  &  2^2\cdot3^4 &               &                        \\              
7                              &    5     &      35    &   5\cdot7^2  &              &               &               \\
\end{tabular}                                                                                                                                   
\caption{Degree 2, size of the last gap between solution numbers.} \label{tab:d2 back gap}                          
} 
                                                                                                                                              
\end{table}


\begin{table}
\parbox{.40\linewidth}{
\centering
\begin{tabular}{l|rrrr}                                                                                                                                 
r\textbackslash n			   &    1     &      2     &      3       &       4    \\  
\hline                                                                             \\  
2                              &    1     &      1     &      1       &       2    \\ 
3                              &    1     &      1     &      1       &            \\ 
4                              &    1     &      2     &      4       &            \\ 
5                              &    1     &      1     &              &            \\ 
6                              &    1     &      1     &              &            \\ 
7                              &    1     &      1     &              &            \\
\end{tabular}                                                                                                                                           
\caption{Degree 3, size of the first gap between solution numbers.\vspace{2cm}} \label{tab:d3 front gap}
}
\hspace{1cm}
\parbox{.55\linewidth}{
\centering

\begin{tabular}{l|>{$}r<{$}  >{$}r<{$}  >{$}r<{$}  >{$}r<{$}}                                                                                                                                        
r\textbackslash n			   &  \mc{1}  &   \mc{2}   &    \mc{3}    &     \mc{4} \\  
\hline                                                                             \\  
2                              &    1     &      1     &      1       &       2    \\ 
3                              &    1     &      2     &      6       &            \\ 
4                              &    1     &      4     &      8       &            \\ 
5                              &    2     &      10    &              &            \\ 
6                              &    2     &      8     &              &            \\ 
7                              &    4     &      28    &              &            \\
\end{tabular}                                                                                                                                                  
\caption{Degree 3, size of the last gap between solution numbers.\vspace{2cm}} \label{tab:d3 back gap}                        
} 
                                                                                                                                              
\end{table}

\clearpage
\restoregeometry
\begin{minipage}{\textwidth}
\thispagestyle{empty}
\vspace{-2cm}
\thispagestyle{empty}
{\footnotesize
\thispagestyle{empty}
\begin{lstlisting}[caption={Degree 3, checks divisibility hypothesis on random polynomials.}, label=code2]

public class SolutionsRandom {
	
    public static boolean next(int[] arr, int ring){
        for (int index = arr.length -1; index >= 0; --index ){
            if(arr[index] < ring-1){
                ++arr[index];
                return true;
            } else {
                arr[index] = 0;
            }
        }
        return false;
    }
   
    public static void main(String[] args) {
        int ring = 8;
        int vars_num = 6;
        int div = 512;
        int tries = 50;
        int coef_num = vars_num + (vars_num*(vars_num+1))/2 
        		+ (vars_num*(vars_num+1)*(vars_num+2))/6;
        int[] coefs = new int[coef_num];
        int[] vars = new int[vars_num];
        int[] ringout = new int[ring];
        for(int counter=0; counter < tries; counter++) {
        	for(int i =0; i < coefs.length; i++) {
        		coefs[i] = (int) (Math.random()*ring);
        	}
            for(int i = 0; i < ring; ++i){
                ringout[i] = 0;
            }
            for(int i = 0; i < vars_num; ++i){
                vars[i] = 0;
            }
            vars[vars.length -1] = -1;
            while(next(vars, ring)){
                int result = 0;
                for(int i = 0; i < vars_num; ++i){
                    result += coefs[i]*vars[i];
                }
                int off = vars_num;
                for(int i = 0; i < vars_num; ++i){
                    for(int j = i; j < vars_num; ++j){
                        result += coefs[off++]*vars[i]*vars[j];
                    }
                }
                for(int i = 0; i < vars_num; ++i){
                    for(int j = i; j < vars_num; ++j){
                    	for(int k = j; k < vars_num; ++k){
                    		result += coefs[off++]*vars[i]*vars[j]*vars[k];
                    	}
                    }
                }
                ++ringout[result % ring];
            }
            for(int sol_num : ringout){
            	int remainder = sol_num % div;
                if(remainder != 0){
                	System.out.println("remainder: " + (remainder));
                }
            }
        }
    }
}
\end{lstlisting}
}
\thispagestyle{empty}
\end{minipage}
\pagestyle{empty}
\restoregeometry

\section{Prospective applications in computer science}
\pagestyle{plain}

Polynomials modulo composite numbers represent the frontier of what is known in
computational complexity theory, and a step beyond the well worked-out theory of
polynomials over fields.  In complexity they correspond to the class
$\mathsf{ACC^0}$ of languages represented by constant-depth, polynomial-sized
circuits of Boolean and mod-$m$ gates.  That nonuniform $\mathsf{ACC}$  was only
recently separated from the nondeterministic exponential time class
$\mathsf{NEXP}$ \cite{wil11} indicates how difficult they are to study. In
mathematics there are strange behaviors even for univariate polynomials, for
instance $x$ ``factors'' as $(4x + 3)(3x + 4)$ over $\ZZ_6$.
Improving our understanding of their behavior may be of great use, when trying
to prove more-strict lower bounds on $\mathsf{ACC^0}$. It should be noted
though, that the results are not directly translatable, as in circuits the
inputs are limited only to $\p*{0, 1}$, even when mod-$m$ gates are being used. 
Moreover the bounds are unknown only when $m$ has two or more prime factors. 
Still, greater knowledge of the solution-space structures for these $m$ may help
investigate the intersection with the image of the Boolean cube.

Cai, Chen, and Lu showed that counting number of solutions for polynomials of
degree up to $2$ in a ring of a fixed size is doable in polynomial time
\cite{cai13}. When the degree becomes $3$ or higher, however, it is known to be
$\#\mathsf{P}{-}complete$ in general \cite{cai13}. The structure of solution
spaces begun here, when further developed, may help map the boundary between
feasible and hard cases in greater detail. This is especially important for the
$Z$-function, where symmetry (or its lack) of solution cardinalities impacts
the balance of the sum around the unit circle.

The application area that directly prompted this inquiry though, is the
algebraic analysis of quantum circuits.  Implicit or explicit in several
well-known papers \cite{ADH97,FR99,daw04,BvDR08} is the conversion of a quantum
circuit $C$ into a polynomial $Q(\vec{y})$ over $\ZZ_m$ (where $m$ is usually of
the form $2^r$)  such that transition amplitude from input $\textbf{a}$ to
output $\textbf{b}$ is given by
\[
\Bra{\textbf{b}}U\Ket{\textbf{a}}= \frac{1}{R}\sum_{k=0}^{m-1}
\#k_Q^*\cdot\omega^k \] where $\omega = e^{2\pi i/m}$ and $R$ is a normalizing
constant depending only on $C$. This form of the equivalent exponential sum
$\sum_{\vec{y}} \omega^{Q(\vec{y})}$ emphasizes the role played by the
solution-set cardinalities for the polynomials $Q(\vec{y}) - k$ over all $k$. 
The one hitch (as above) is that $\#k_Q^*$ restricts the count to those
arguments $\vec{y} = y_1 y_2\cdots y_h$ that belong to the Boolean cube
$\p*{0,1}^h$, taking it outside the immediate purview of the results for $\#k_Q$
which range over all of $\ZZ_m^h$.

However, in some cases there is a correspondence between $\p*{0,1}^h$ and
$\ZZ_m^h$ that enables carrying over the results.  This is the case when $C$ is
a circuit of {\em stabilizer\/} gates, which produce a polynomial $Q$ over
$\ZZ_4$ consisting entirely of terms of the form $y_i^2$ or $2 y_i y_j$
\cite{ReCh12}.  Then only the parities of $y_i$ and $y_j$ matter.  The
above-mentioned theorem of \cite{cai13} then takes effect to show that the solution
counts are polynomial-time computable, which yields yet-another-proof of the
classical polynomial-time simulation of quantum stabilizer circuits \cite{Got98} (see
also \cite{AaGo04}).

The divisibility of the numbers $\#k_Q$ by large powers of the ring size, as proved
by Marshall and Ramage \cite{mar75}, implies
limitations on the range of values that this probability can take.  In
particular, it limits the ability to reduce the failure probability $\epsilon$
of the measurement for a given size circuit---unless the circuit actually gives
$\epsilon = 0$. The size is bounded below by the number $h$ of nondeterministic
gates (which generally are all Hadamard gates), which give rise to the variables
$\vec{y} = y_1,\dots,y_h$.

\section{Proof of Theorem \ref{main}} \label{mainproof}

We present the proofs of the general induction step and the base case of Theorem \ref{main} separately, respectively in subsections
\ref{s:general} and \ref{s:base}.

\subsection{Proof of the general induction step
} \label{s:general}

For any $x$ in a ring $\ZZ_{p^r} = \ZZ/p^r\ZZ$ where $p$ is prime, let $o(x) = max \p*{m : m \leq r \wedge p^m | x}$ be the
\emph{order} of $x$.  The following lemma is a basic observation about the rings $\ZZ_{2^r}$.

\begin{lemma} \label{lem1}
Let $h \in \ZZ_{2^r}$ and $f = o(h)$. Then the following sets are equal as subsets of $\ZZ_{2^r}$:
\[
\hspace{4.5cm} \p*{ 2^fi : i \in Z_{2^r}} =\p*{ hi : i \in Z_{2^r}} = \p*{ hi : i \in Z_{2^{r-f}}}.  \hspace{4.5cm} \Box
\]

\end{lemma}

Any such set described in the lemma above, contains all elements of the ring that share the same order. As an example, for $\ZZ_{32}$,
each such set contains all elements from a single level of the tree below, or it contains just the $0$ element. 

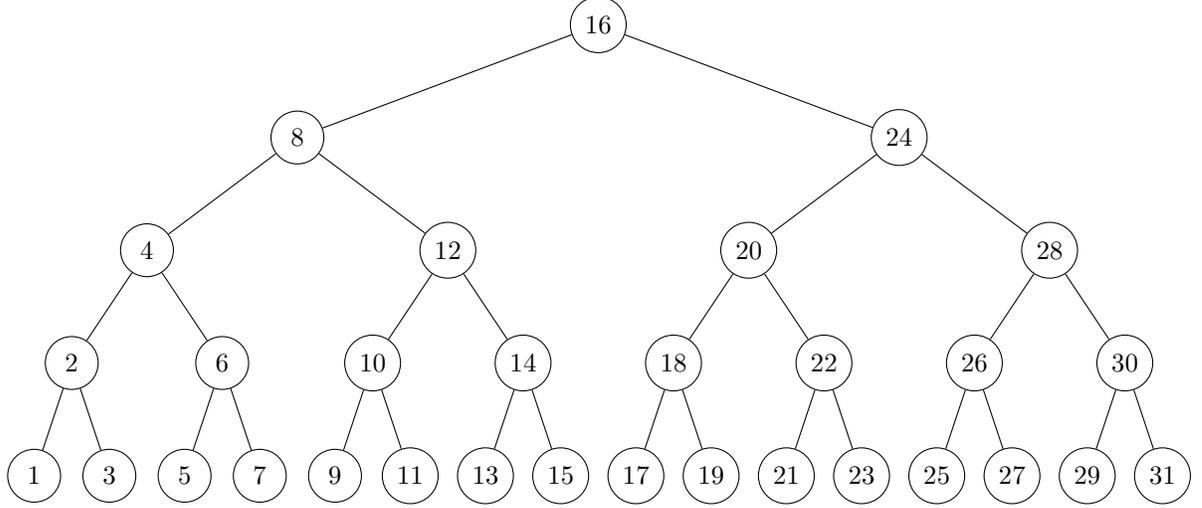
\begin{tikzpicture}%
  [
  level distance=1.5cm,
  level 1/.style={sibling distance=8cm},
  level 2/.style={sibling distance=4cm},
  level 3/.style={sibling distance=2cm},
  level 4/.style={sibling distance=1cm},
  every node/.style={draw,circle, minimum size=2em}]
  \node[circle,draw]{$16$} 
    child {node[circle,draw]{$8$}
      child {node[circle,draw] {$4$}
      	child {node[circle,draw] {$2$}
      		child {node[circle,draw] {$1$}}
      		child {node[circle,draw] {$3$}}
      	}
      	child {node[circle,draw] {$6$}
      		child {node[circle,draw] {$5$}}
      		child {node[circle,draw] {$7$}}
      	}
      }
      child {node[circle,draw] {$12$}
      	child {node[circle,draw] {$10$}
      		child {node[circle,draw] {$9$}}
      		child {node[circle,draw] {$11$}}
      	}
      	child {node[circle,draw] {$14$}
      		child {node[circle,draw] {$13$}}
      		child {node[circle,draw] {$15$}}
      	}
      }
    }
    child {node[circle,draw]{$24$}
      child {node[circle,draw] {$20$}
      	child {node[circle,draw] {$18$}
      		child {node[circle,draw] {$17$}}
      		child {node[circle,draw] {$19$}}
      	}
      	child {node[circle,draw] {$22$}
      		child {node[circle,draw] {$21$}}
      		child {node[circle,draw] {$23$}}
      	}
      }
      child {node[circle,draw] {$28$}
      	child {node[circle,draw] {$26$}
      		child {node[circle,draw] {$25$}}
      		child {node[circle,draw] {$27$}}
      	}
      	child {node[circle,draw] {$30$}
      		child {node[circle,draw] {$29$}}
      		child {node[circle,draw] {$31$}}
      	}
      }
    };
\end{tikzpicture}

\vspace{0.5cm}

We will often use this lemma to change the order of iteration. For example,
\[ 
\bigcup_{i=0}^{2^q-1} \p*{2^{r-q}i} = \bigcup_{i=0}^{2^q-1} \p*{(2k+1)2^{r-q}i}  = \bigcup_{i=0}^{2^q-1} \p*{(2^{r-q} +t2^{r-q+1})i}
\] 
where $q, t, k$ are any integers.
\vfill

\begin{reptheorem}{main}
\emph{\textbf{(General induction step)}}

Let any quadratic polynomial $Q$ of $n\geq 4$ variables ($\xx, z$) over $\ZZ_{2^r}$, and any
integers $q, v\leq r$ and $k, w, g, u$, and any linear polynomial $T(\xx)$ be given. Suppose that for any $Q'$ of $n-1$ variables ($\xx',
z'$), any $q', v'\leq r$ and $k', w', g', u'$, and any linear polynomial $T'(\xx')$ it holds that:
\[ 
\sum\limits_{i=0}^{2^{q'}-1} \sum\limits_{j=0}^{2^{v'}-1}\#(k'{+}w'2^{r-q'}i{+}g'2^{r-v'}j)_{Q', z'=T'(\xx')+u'2^{r-v'}j}   = \TT'2^{\ceil*{
\frac{r(n-2)+\min(2v',r)}{2} } +q'-1}
\]
for certain integer $\TT'$. Then there is an integer $\TT$ such that 
\[ 
\sum\limits_{i=0}^{2^q-1} \sum\limits_{j=0}^{2^v-1}\#(k{+}w2^{r-q}i{+}g2^{r-v}j)_{Q, z=T(\xx)+u2^{r-v}j}   = \TT2^{\ceil*{
\frac{r(n-1)+\min(2v,r)}{2} } +q-1}.
\]
\end{reptheorem}

\begin{proof}
Taking as the induction hypothesis that the theorem is true for all polynomials over $n-1$ variables, we would like to show that it holds
for any $Q$ such that
{\samepage
\[ 
Q(\textbf{x}, z) = M(\textbf{x}, z) + mz^2,
\]
\[ 
M(\textbf{x}, z) = P(\textbf{x}) + L(\textbf{x})z.
\]
}Here $Q$ and $M$ are over $n$ variables and have degree up to 2, $P$ is over $n-1$ variables and also has degree up to 2,  $L$
is a linear form over $n-1$ variables and $m$ is a constant. We will first prove the divisibility for $M$ using the induction hypothesis
for $n-1$, and then we will prove divisibility of $Q$, depending only on the result for $M$. \\
Let us notice that
\[
\#k_{M, z=l} = \sum_{h=0}^{2^r-1} \#(k{-}hl)_{P, L(\textbf{x})=h}. 
\]
We will frequently use decompositions of this form.

Let us move to proving the $M$ part of the theorem, 
by which we mean the conclusion of Theorem~\ref{main} with $M$ in place of $Q$.  We calculate:
\[
C = \sum\limits_{i=0}^{2^q-1} \sum\limits_{j=0}^{2^v-1} \#(k{+}w2^{r-q}i{+}g2^{r-v}j)_{M, z=T(\xx)+u2^{r-v}j}
=\sum_{i=0}^{2^q-1} \sum\limits_{j=0}^{2^v-1} \#(k{+}w2^{r-q}i{+}g2^{r-v}j)_{U, z =2^{r-v}j}\;,
\] 
where $U(\xx) =  P(\xx)+L(\xx) (T(\xx) + uz)$.  Let us write $H(\xx) = P(\xx)+L(\xx) (T(\xx) + uz) - gz$ so that
\[
C = \sum_{i=0}^{2^q-1} \sum\limits_{j=0}^{2^v-1}\#(k{+}w2^{r-q}i)_{H, z =2^{r-v}j}.
\]
Take $\wu = o(w)$, that is write $w = 2^{\wu}b$ where $b$ is odd.  By appeal to Lemma~\ref{lem1} we may ignore $b$, so we have
\[
C = \sum_{i=0}^{2^q-1} \sum\limits_{j=0}^{2^v-1}\#(k{+}2^{r-q+\wu}i)_{H, z =2^{r-v}j}.
\]
Now we can rewrite $H = P'(\xx) + L'(\xx)z$ with certain $P'$ no worse than quadratic, and importantly, certain linear $L'$. Then we
can further condition on all possible values $h$ of $L'(\xx)$, to obtain 
\[
C = \sum_{i=0}^{2^q-1} \sum\limits_{j=0}^{2^v-1} \sum_{h=0}^{2^r-1} \#(k{+}2^{r-q+\wu}i{-}2^{r-v}jh)_{H, L'(\xx)=h}.
\]
Considering all possible orders $f$ of $h$ separately then gives:
\[
C =  \sum_{i=0}^{2^q-1} \sum\limits_{j=0}^{2^{v}-1}  \sum_{f=0}^{r-1} \sum_{h=0}^{2^{r-f-1}-1}
\#(k{+}2^{r-q+\wu}i{+}2^{r-v}j2^f(2h+1))_{H, L'(\xx)=2^f(2h+1)} 
\]
\[
+  \sum_{i=0}^{2^q-1} \sum\limits_{j=0}^{2^{v}-1} \#(k{+}2^{r-q+\wu}i)_{H, L'(\xx)=0}.
\]

Let us divide the above sum into two parts $C_1$ and $C_2$ and consider them independently.  The first part is for $f\geq v$, and the second
part is for the remaining orders $f \leq v-1$.
Starting with the first part, we have:
\[
C_1 = \sum_{i=0}^{2^q-1} \sum\limits_{j=0}^{2^{v}-1}  \sum_{f=v}^{r-1} \sum_{h=0}^{2^{r-f-1}-1}
\#(k{+}2^{r-q+\wu}i{+}2^{r-v+f}j(2h+1))_{H, L'(\xx)=2^f(2h+1)} 
\]
\[
+  \sum_{i=0}^{2^q-1} \sum\limits_{j=0}^{2^{v}-1} \#(k{+}2^{r-q+\wu}i)_{H, L'(\xx)=0}
\]
Since $f \geq v$, we get $2^{r-v+f}=0$, and therefore
\[
C_1 = 2^{v}\sum_{i=0}^{2^q-1}   \sum_{h=0}^{2^{r-v}-1}
\#(k{+}2^{r-q+\wu}i)_{H, L'(\xx)=2^v h}\;, 
\]
where we also collapsed orders of $h$, by considering all of them when their order is at least $v$ at once.

For certain values of $h$, the condition $L'(\xx) = 2^v h$ may be unsolvable. Let us take a variable $y$ in $\xx$ that is multiplied by some
$\alpha2^\beta$ in $L'$, such that $\alpha$ is odd and no other variable in $\xx$ is being multiplied in $L'$ by a coefficient of a smaller
order. Let $\delta$ be the constant term in $L'$. If the order of $\delta$ is smaller than both $\beta$ and $v$, then $L'(\xx) = 2^v h$
never has a solution and our whole expression becomes $0$, which has any divisibility. When $o(\delta) \geq min(\beta, v)$ then $\delta$
only impacts which coset of set of solutions of $L'(\xx) - \delta = 2^v h $ will be the solutions of $L'(\xx) = 2^v h $. Therefore we can
assume $\delta = 0$ without the loss of generality ($h$ goes over the whole subring, while $\alpha$ can be anything).
Then $L'(\xx) = 2^v h$ is solvable only when $h = \gamma2^{max(\beta-v, 0)}$ for certain $\gamma$.
For any such $h$ we can solve the equation for $y$ obtaining $2^\beta$ solutions of the form
\[ 
y_{i} = L'_{\neg y}(\xx_{\neg y}) +\frac{\gamma2^{max(v-\beta, 0)}}{\alpha} + i2^{r-\beta}, 
\]
for $i$ between $0$ and $2^\beta-1$ and $L'_{\neg y}(\xx_{\neg y})$ being over $n-2$ variables and defined as: $L'_{\neg y}(\xx_{\neg y}) =
(- L'(\xx) +\alpha2^{\beta}y) /\alpha2^{\beta}$.
Coming back to our sum $C_1$ as given earlier, we have:
\[
C_1 = 2^{v}\sum_{i=0}^{2^q-1} \sum_{h=0}^{2^{r-v}-1}
\#(k{+}2^{r-q+\wu}i)_{H, L'(\xx)=2^v h} 
\]
\[
 = 2^{v} \sum_{i=0}^{2^q-1}   \sum_{h=0}^{2^{r-max(v, \beta)}-1}
\#(k{+}2^{r-q+\wu}i)_{H, L'(\xx)=2^{max(v , \beta)}h}\;,
\]
where we omitted $h$-s for which the constraint was always false.  Carrying on,
\[
C_1 = 2^{v} \sum_{i=0}^{2^q-1}    \sum_{h=0}^{2^{r-max(v, \beta)}-1} \sum_{t=0}^{2^{\beta}-1}
\#(k{+}2^{r-q+\wu}i)_{H,  y = L'_{\neg y}(\xx_{\neg y}) + \frac{2^{max(v-\beta, 0)}}{\alpha}h + t2^{r-\beta}} 
\]
\[
 = 2^{v} \sum_{i=0}^{2^q-1}    \sum_{s=0}^{2^{r-max(v-\beta, 0)}-1}
\#(k{+}2^{r-q+\wu}i)_{H,  y = L'_{\neg y}(\xx_{\neg y}) + 2^{max(v-\beta, 0)}s},
\]
since $\frac{2^{max(v-\beta, 0)}}{\alpha}h$ produces all values in $\ZZ_{2^{r-\beta}}$ that are divisible by $2^{max(v-\beta, 0)}$, then
$t2^{r-\beta}$ expands them to all such values in $\ZZ_{2^r}$.  Hence
\[
C_1 = 2^{v+ min(q, \wu)} \sum_{i=0}^{2^{max(q-\wu, 0)}-1} \sum_{s=0}^{2^{r-max(v-\beta, 0)}-1}
\#(k{+}2^{r-q+\wu}i)_{H,  y = L'_{\neg y}(\xx_{\neg y}) + 2^{max(v-\beta, 0)}s} 
\]
\[ 
= 2^{v + min(q,\wu)} \TT2^{\ceil*{ \frac{r(n-2)+\min(2r-2max(v-\beta, 0),r)}{2} } +max(q-\wu,0)-1}
\]
\[ 
= \TT2^{\ceil*{ \frac{r(n-1)+\min(r + 2(v - max(v-\beta, 0)),2v)}{2} } +q-1}
= \TT2^{\ceil*{ \frac{r(n-1)+\min(2v, r +min(\beta, v))}{2} } +q-1}
\]
which has possibly even more than the required divisibility. We used the induction hypothesis taking $g$ to be $0$ and $T(\xx_{\neg y}) =
L_{\neg y}(\xx_{\neg y})$.

Let us look now at the second part, \ie for $f < v$:
\[  
C_2 = \sum_{i=0}^{2^q-1} \sum\limits_{j=0}^{2^{v}-1}  \sum_{f=0}^{v-1} \sum_{h=0}^{2^{r-f-1}-1}
\#(k{+}2^{r-q+\wu}i{+}2^{r-v+f}j(2h+1))_{H, L'(\xx)=2^f(2h+1)} 
\]
\[  
= \sum_{i=0}^{2^q-1} \sum\limits_{j=0}^{2^{v}-1}  \sum_{f=0}^{v-1} \sum_{h=0}^{2^{r-f-1}-1}
\#(k{+}2^{r-q+\wu}i{+}2^{r-v+f}j)_{H, L'(\xx)=2^f + 2^{f+1}h}.
\]
Again by appeal to Lemma \ref{lem1},
\[  
C_2 = \sum_{f=0}^{v-1} 2^{min(q,\wu) + min(v, f)}\sum_{i=0}^{2^{max(q-\wu, 0)}-1} \sum\limits_{j=0}^{2^{max(v-f,0)}-1}  
\sum_{h=0}^{2^{r-f-1}-1} \#(k{+}2^{r-q+\wu}i{+}2^{r-v+f}j)_{H, L'(\xx)=2^f + 2^{f+1}h} 
\]
{\small
\[  
= \sum_{f=0}^{v-1} 2^{min(q,\wu) + min(v, f) + max(min(q-\wu,v-f), 0)} \sum_{i=0}^{2^{max(q-\wu, v-f, 0)}-1}
\sum_{h=0}^{2^{r-f-1}-1} \#(k{+}2^{r-max(q-\wu, v-f)}i)_{H, L'(\xx)=2^f + 2^{f+1}h} 
\]
}

\noindent
owing to the overlap of $2^{r-q+\wu}i$ and $2^{r-v+f}j$. Following steps are analogous to what we did in previous part: we solve
$L'(\xx)=2^f + 2^{f+1}h$ for a specific $y \in \xx$ to obtain that $C_2$ equals the sum over $f$ from $0$ to $v-1$ of
\[
2^{min(q,\wu) + min(v, f) +max(min(q-\wu,v-f), 0)} 
\]
multiplied by
\[
\sum_{i=0}^{2^{max(q-\wu, v-f, 0)}-1}
\sum_{s=0}^{2^{r-max(f+1-\beta, 0)}-1} \#(k{+}2^{r-max(q-\wu, v-f)}i)_{H, y=L_{\neg y}(\xx_{\neg
y})+2^{max(f+1-\beta, 0)}s}.
\]
This has the right form for applying the induction hypothesis, which gives us:
\[
C_2 = \sum_{f=0}^{v-1} 2^{min(q,\wu) + min(v, f) +max(min(q-\wu,v-f), 0)}  
\TT2^{\ceil*{ \frac{r(n-2)+\min(2r-2max(f+1-\beta, 0),r)}{2} } +max(q-\wu, v-f, 0)-1}
\]
\[
= \sum_{f=0}^{v-1} 2^{min(q,\wu) + min(v, f) + max(min(q,\wu), min(v, f))}  
\TT2^{\ceil*{ \frac{r(n-2)+\min(2r-2max(f+1-\beta, 0),r)}{2} } +max(q-\wu, v-f, 0)-1}
\]
\[
= \sum_{f=0}^{v-1}
\TT2^{\ceil*{ \frac{r(n-1)+\min(r-2max(f+1-\beta, 0),0)}{2} } +q+v-1}
= \sum_{f=0}^{v-1}
\TT2^{\ceil*{ \frac{r(n-1)+\min(2v, r-2max(f+1-v-\beta,-v))}{2} } +q-1}.
\]
Thus $C_2$ has the required divisibility, since $max(f+1-v-\beta,-v) \leq 0$ owing to $f < v$.  Hence so does $C = C_1 + C_2$.  This proves
the induction step for $M$.

\bigskip
Having proved the divisibility property for $M$, we may now use it in the proof for $Q$:
\[
D = \sum\limits_{i=0}^{2^q-1} \sum\limits_{j=0}^{2^v-1} \#(k{+}w2^{r-q}i{+}g2^{r-v}j)_{Q, z=T(\xx)+u2^{r-v}j}
\]
\[
 = \sum\limits_{i=0}^{2^q-1} \sum\limits_{j=0}^{2^v-1} \#(k{+}2^{r-q+\wu}i)_{G, z=2^{r-v}j},
\]
where $G(\xx, z) =  M(\xx) + T(\xx)L(\xx) + L(\xx)uz - gz+ mT(\xx)^2+ 2muT(\xx)z + mu^2z^2$ (we used analogous transformation as we did for
$M$ at the beginning of the proof). Let 
\[ 
H(\xx, z) = M(\xx) + T(\xx)L(\xx) + L(\xx)uz - gz+ mT(\xx)^2+ 2muT(\xx)z, 
\]
which means $H$ has no square term for $z$, and its degree is up to $2$. Later we will apply the divisibility we just proved for $M$, as
induction hypothesis, to $H$. The sum we are working on equals
\[
D = \sum\limits_{i=0}^{2^q-1} \sum\limits_{j=0}^{2^v-1} \#(k{+}w2^{r-q}i - mu^2(2^{r-v}j)^2)_{H, z=2^{r-v}j}
\]
\[
= \sum\limits_{i=0}^{2^q-1} \sum\limits_{j=0}^{2^v-1} \#(K(i){-} t(2^{r-v}j)^2)_{H, z=2^{r-v}j},
\]
where $K(i) =  k + w2^{r-q}i$ and $t = mu^2$. 

Now we iterate over $2^{r-v}j$, by going through all possible orders of it, as usual denoted by $f$. We have
\begin{align} 
D =\sum\limits_{f=r-v}^{r-1} \sum\limits_{i=0}^{2^q-1}  \sum\limits_{j=0}^{2^{r-f-1}-1} 
\#(K(i){-}t(2^f(2j+1))^2)_{H,  z=2^f(2j+1)}
+ \sum\limits_{i=0}^{2^q-1} \#K(i)_{H,  z=0}. \label{sqr}
\end{align}

We consider cases when $r$ is even or odd separately. Let us first take $r$ to be even.
From the sum \ref{sqr} above, let us take any component of it having a single $f \leq \frac{r}{2}-1$.  We will show that each such a
component has the required divisibility
\[
\sum\limits_{i=0}^{2^q-1}  \sum\limits_{j=0}^{2^{r-f-1}-1} 
\#(K(i){-}t(2^f(2j+1))^2)_{H,  z=2^f(2j+1)}
\]
\[
=\sum\limits_{i=0}^{2^q-1}  \sum\limits_{j=0}^{2^{\frac{r}{2}-f-1}-1} \sum\limits_{h=0}^{2^{\frac{r}{2}}-1} 
\#(K(i){-}t(2^f(2j{+}1){+}2^{\frac{r}{2}}h)^2)_{H,  z=2^f(2j+1)+2^{\frac{r}{2}}h}\;,
\]
where we changed order of iteration on $z$ to consider it in groups that belong to subrings $\ZZ_{2^{r/2}}$, 
\[
=\sum\limits_{i=0}^{2^q-1}  \sum\limits_{j=0}^{2^{\frac{r}{2}-f-1}-1} \sum\limits_{h=0}^{2^{\frac{r}{2}}-1} 
\#(K(i){-}t((2j+1)(2^f{+}2^{\frac{r}{2}}h))^2)_{H,  z=(2j+1)(2^f+2^{\frac{r}{2}}h)}\;,
\]
where we changed order of iteration for $h$ by multiplying by $2j+1$ which is odd (via Lemma \ref{lem1}),
\[
=\sum\limits_{j=0}^{2^{\frac{r}{2}-f-1}-1} \sum\limits_{i=0}^{2^q-1} \sum\limits_{h=0}^{2^{\frac{r}{2}}-1} 
\#(K(i){-}t(2j+1)^2(2^{2f}{+}2^{\frac{r}{2}+f+1}h))_{H,  z=(2j+1)(2^f+2^{\frac{r}{2}}h)}
\]
\[
=\sum\limits_{j=0}^{2^{\frac{r}{2}-f-1}-1} \sum\limits_{i=0}^{2^q-1} \sum\limits_{h=0}^{2^{\frac{r}{2}}-1} 
\#( k {-}t(2j+1)^22^{2f}  {+} w2^{r-q}i 
{-}t(2j+1)^2 2^{\frac{r}{2}+f+1}h)_{H,  z=(2j+1)(2^f+2^{\frac{r}{2}}h)}
\]
\[
=\sum\limits_{j=0}^{2^{\frac{r}{2}-f-1}-1} \TT_j2^{\ceil*{ \frac{r(n-1)+\min(2\frac{r}{2},r)}{2} }
+q-1} =  \sum\limits_{j=0}^{2^{\frac{r}{2}-f-1}-1} \TT_j2^{\ceil{ \frac{rn}{2} }
+q-1},
\]
which may even have a higher divisibility than required. We used the induction hypothesis with:
\begin{align*}
&k' = k {-}t(2j+1)^22^{2f} \\
&T'(\xx) = (2j+1)2^f \\
&g' = {-}t(2j+1)^2 2^{f+1} \\
&u' = (2j+1). \\
\end{align*}
Let us note again, that we used the induction hypothesis for a polynomial of $n$ variables, yet it is for $H$ that does not have a square
term in $z$, therefore is of the same form as $M$, for which the theorem for $n$ variables is already proved.

Now let us consider those terms involving $f$ (from the earlier mentioned sum \ref{sqr}, that $D$ became) for which $f \geq max(\frac{r}{2},
r-v)$:
\[
\sum\limits_{f=max(\frac{r}{2}, r-v)}^{r-1} \sum\limits_{i=0}^{2^q-1}  \sum\limits_{j=0}^{2^{r-f-1}-1} 
\#(K(i){-}t(2^f(2j+1))^2)_{H,  z=t2^f(2j+1)}
+ \sum\limits_{i=0}^{2^q-1} \#K(i)_{H,  z=0}
\]
\[
=  \sum\limits_{i=0}^{2^q-1} \sum\limits_{j=0}^{2^{r-max(\frac{r}{2}, r-v)}-1}
\#(K(i){-}t(2^{max(\frac{r}{2}, r-v)}j)^2)_{H,  z=2^{max(\frac{r}{2}, r-v)}j}\]
\[=  \sum\limits_{i=0}^{2^q-1} \sum\limits_{j=0}^{2^{r-max(\frac{r}{2}, r-v)}-1}
\#(k {+} w2^{r-q}i)_{H,  z=2^{max(\frac{r}{2}, r-v)}j}
\]
\[
=  \TT2^{\ceil*{ \frac{r(n-1)+\min(2(r-max(\frac{r}{2}, r-v)),r)}{2} }+q-1} =  \TT2^{\ceil*{ \frac{r(n-1)+\min(2v,r)}{2} }+q-1},
\]
with use of the induction hypothesis for $M$ of $n$ variables.

Let us move now to the second case, that is when $r$ is odd. First we take from $D$, being written in the form of the sum \ref{sqr}, any
component for a single value of $f$, such that $f \leq \frac{r-1}{2}-1$, and show that such a component has required divisibility:
\[
\sum\limits_{i=0}^{2^q-1}  \sum\limits_{j=0}^{2^{r-f-1}-1} 
\#(K(i){-}t(2^f(2j+1))^2)_{H,  z=2^f(2j+1)} 
\]
\[
=\sum\limits_{i=0}^{2^q-1}  \sum\limits_{j=0}^{2^{\frac{r-1}{2}-f-1}-1} \sum\limits_{h=0}^{2^{\frac{r+1}{2}}-1} 
\#(K(i){-}t(2^f(2j+1){+}2^{\frac{r-1}{2}}h)^2)_{H,  z=2^f(2j+1){+}2^{\frac{r-1}{2}}h} 
\]
\[
=\sum\limits_{i=0}^{2^q-1}  \sum\limits_{j=0}^{2^{\frac{r-1}{2}-f-1}-1} \sum\limits_{h=0}^{2^{\frac{r+1}{2}}-1} 
\#(K(i){-}t(2j+1)^2(2^f{+}2^{\frac{r-1}{2}}h)^2)_{H,  z=(2j+1)(2^f{+}2^{\frac{r-1}{2}}h)}
\]
Now we add and subtract $1$ inside the expression
\begin{align*}
=\sum\limits_{i=0}^{2^q-1}  \sum\limits_{j=0}^{2^{\frac{r-1}{2}-f-1}-1} \sum\limits_{h=0}^{2^{\frac{r+1}{2}}-1} 
\#\Bigg(K(i){-}t&(2j+1)^2 \\
\Big( 2^{2f}{+}2^{\frac{r+1}{2}+f}h &+
2^{r-1}h^2{+}2^{\frac{r+1}{2}}(2^{\frac{r-3}{2}}{-}2^f)h{-}2^{\frac{r+1}{2}}(2^{\frac{r-3}{2}}{-}2^f)h \Big) \Bigg) _{H,
z=(2j+1)(2^f{+}2^{\frac{r-1}{2}}h)}\\
=\sum\limits_{i=0}^{2^q-1}  \sum\limits_{j=0}^{2^{\frac{r-1}{2}-f-1}-1} \sum\limits_{h=0}^{2^{\frac{r+1}{2}}-1} 
\#\Bigg(K(i){-}t&(2j+1)^2\\
\Big(2^{2f}{+}&2^{\frac{r+1}{2}}h(2^f{+}2^{\frac{r-3}{2}}h{+}2^{\frac{r-3}{2}}{-}2^f){-}2^{\frac{r+1}{2}}(2^{\frac{r-3}{2}}{-}2^f)h\Big)\Bigg)
_{H, z=(2j+1)(2^f{+}2^{\frac{r-1}{2}}h)}\\
=\sum\limits_{i=0}^{2^q-1}  \sum\limits_{j=0}^{2^{\frac{r-1}{2}-f-1}-1} \sum\limits_{h=0}^{2^{\frac{r+1}{2}}-1} 
\#\Bigg(K(i)
{-}t&(2j+1)^2\Big(2^{2f}{+}2^{r-1}h(h{+}1){-}2^{\frac{r+1}{2}}(2^{\frac{r-3}{2}}{-}2^f)h\Big)\Bigg)
_{H, z=(2j+1)(2^f{+}2^{\frac{r-1}{2}}h)}\\
=\sum\limits_{j=0}^{2^{\frac{r-1}{2}-f-1}-1} \sum\limits_{i=0}^{2^q-1}   \sum\limits_{h=0}^{2^{\frac{r+1}{2}}-1} 
\#\Bigg(k{-}t(2j&+1)^22^{2f}{+} w2^{r-q}i 
{+}2t(2j+1)^2(2^{\frac{r-3}{2}}{-}2^f) 2^{\frac{r-1}{2}}h\Bigg)
_{H, z=(2j+1)(2^f{+}2^{\frac{r-1}{2}}h)}
\end{align*}
\[
=\sum\limits_{j=0}^{2^{\frac{r-1}{2}-f-1}-1}  \TT_j2^{\ceil*{ \frac{r(n-1)+\min(2\frac{r+1}{2},r)}{2}}+q-1}
=\sum\limits_{j=0}^{2^{\frac{r-1}{2}-f-1}-1}  \TT_j2^{\ceil*{ \frac{rn}{2}}+q-1},
\]
which may have even higher than required divisibility. The term $2^{r-1}h(h+1)$ always multiplies to $2^r$ so it cancels.  We used the
induction hypothesis of $M$ with:
\begin{align*}
&k' = k{-}t(2j+1)^22^{2f}\\
&T'(\xx) = (2j+1)2^f \\
&g' = 2t(2j+1)^2(2^{\frac{r-3}{2}}{-}2^f) \\
&u' = (2j+1). \\
\end{align*}
Let us consider now those $f$-s from $D$ written as the sum \ref{sqr} for which $f \geq max(\frac{r-1}{2}, r-v)$:
\[
D' = \sum\limits_{f=max(\frac{r-1}{2}, r-v)}^{r-1} \sum\limits_{i=0}^{2^q-1}  \sum\limits_{j=0}^{2^{r-f-1}-1} 
\#(K(i){-}t(2^f(2j+1))^2)_{H,  z=2^f(2j+1)}
+ \sum\limits_{i=0}^{2^q-1} \#K(i)_{H,  z=0}
\]
\[
=\sum\limits_{i=0}^{2^q-1}  \sum\limits_{j=0}^{2^{r-max(\frac{r-1}{2}, r-v)}-1} 
\#(K(i){-}t(2^{max(\frac{r-1}{2}, r-v)}j)^2)_{H,  z=2^{max(\frac{r-1}{2}, r-v)}j}
\]
\[
=\sum\limits_{i=0}^{2^q-1}  \sum\limits_{j=0}^{2^{r-max(\frac{r-1}{2}, r-v)}-1} 
\#(K(i){-}t2^{max(r-1, 2(r-v))}j)_{H,  z=2^{max(\frac{r-1}{2}, r-v)}j}.
\]
The last observation which allows us to use $j$ not $j^2$ is that if $j$ is even the term with $j$ will cancel, while if $j$ is odd then $j$ is multiplying either $0$ or $2^{r-1}$, so the difference between $j$ and $j^2$ is immaterial.  Finishing up:
\[
D'=\sum\limits_{i=0}^{2^q-1}  \sum\limits_{j=0}^{2^{r-max(\frac{r-1}{2}, r-v)}-1} 
\#(k{+}w2^{r-q}i{-}t2^{max(\frac{r-1}{2}, r-v)}2^{max(\frac{r-1}{2}, r-v)}j)_{H,  z=2^{max(\frac{r-1}{2}, r-v)}j}
\]
\[
=\TT2^{\ceil*{ \frac{r(n-1)+\min(2(r-max(\frac{r-1}{2}, r-v)),r)}{2} }+q-1} =\TT2^{\ceil*{ \frac{r(n-1)+\min(2v,r)}{2} }+q-1}.
\]
\end{proof}

\subsection{Proof of the base case
} \label{s:base}

We begin with statements of lemmas and corollaries that we will directly need for the theorem's proof. Then we present the proof itself. We
will end with proofs of the aforementioned lemmas and corollaries, including some additional ones that we build on. 

Let us recall that by $o(m)$ we represent the order of $m$ in a given ring.
\begin{lemma} \label{multi}
For any $m \in Z_{2^r}$,
\begin{align*}
\hspace{3.78cm} \#\p*{x,y \in Z_{2^r}: xy = m} =
\begin{cases}
(o(m)+1)2^{r-1} & \text{\;if $m \neq 0$}.\\
(r+2)2^{r-1} & \text{\;otherwise}. \hspace{3.78cm} \Box
\end{cases} \\
\end{align*}
\end{lemma}
Below we will work with {\em multisets\/}. We will use a ``multiplicative'' notation to represent them. For example,
\[ 
2 \p*{ 1} = \p*{1, 1}; \hspace{8mm} 3 \p*{1,3} \cup 2\p*{2}= \p*{1, 1, 1, 2, 2, 3, 3, 3}; \hspace{8mm} 2\bigcup_{i=0}^{1}\p*{3i} =
\p*{0,0,3,3}.
\]

\begin{lemma} \label{types}
Let us take a polynomial $P(x) = ax^2 + bx +c$ over $\ZZ_{2^r}$. Let $a = q2^w$ and $b = g2^h$ such that $q$ and $g$ are odd and $w, h$ are
orders of respectively $a$ and $b$.
Let $m = min(w, h)$. The image of $P(x)$ treated as a multiset equals
\begin{enumerate}[label=\alph*)]
  \item If $w > h$ :
  \[     2^m \bigcup_{i=0}^{2^{r-m}-1} \p*{2^m i + c}     \]
  
  \item If $w = h$ :
  
\[ 
\begin{array}{ll}
2^{m+1} \bigcup\limits_{i=0}^{2^{r-m-1}-1} \p*{2^{m+1} i + c}   & \text{if }  m < r\\
2^r\p*{c} & \text{if } m = r
\end{array}
\]

\item If $w < h$ :\\
\[
\left(\bigcup_{f=0}^{\ceil*{\frac{r-m}{2}}-1} 2^{min(f+2, r-f-1)+min(m, max(0, r-2f-3))} \bigcup_{i=0}^{max(0, 2^{r-2f-3-m}-1)}
\p*{2^{2f+3+m}i + q2^{2f+m} + t } \right) 
\]
\[
\cup\; 2^{\floor*{\frac{r+m}{2}}}\p*{t}
\]
where $ t= c-\frac{b^2}{2^{m+2}q}$. \hfill $\Box$

\end{enumerate}
\end{lemma}

Let us introduce now a concept of a \emph{slice}, which is a coset of an ideal of a ring. If we look at the multiset that
the image of $P$ in above corollary is, then in cases a) and b) it is just a single slice (possibly with each distinct element having
multiple occurrences). In the case c) the image of $P$ is built from multiple slices, one for each $f$ between $0$ and
$\ceil*{\frac{r-m'}{2}}-1$, and then a final slice $\p*{t}$.

Let us say we would be interested in an intersection between the images of two functions $P(x)$ and $Q(y)$. More precisely we would want to
evaluate:
\[
\#\p*{(x, y) : P(x) = Q(y)}.
\]
To start with, for sake of intuition, let's suppose we are working over $\ZZ_3$ (\ie, both functions have a $3$-element domain), and that
the image of $P$ is $\p*{0,1,1}$ (with $P(0) = 0$) while the image of $Q$ is $\p*{1,1,1}$. Then 
\[ 
\p*{(x, y) : P(x) = Q(y)} = \p*{(1,0), (1,1), (1, 2), (2,0), (2,1), (2,2)} 
\]
and the size of the intersection of the images is $6$. When we look directly at the multisets $\p*{0,1,1}$ and $\p*{1,1,1}$ that the images
constitute, we would like the intersection of them, as we understand it, to be also of size $6$. This gives rise to the
 ``\emph{multiplicative}'' intersection concept, which in our example is
\[ 
\p*{(1,1), (1,1), (1, 1), (1,1), (1,1), (1,1)}.
 \]
That is, each element from the first image is paired up with each element of the second image that it is equal to. For another example, the
intersection of $\p*{1, 1, 2}$ and $\p*{1, 2, 3}$ is deemed to be
\[ 
\p*{(1, 1), (1, 1), (2, 2)}. 
\]
In general, if the first multiset has distinct elements $p_i$ with respective numbers of occurrences $a_i$, and the second multiset has
distinct elements $q_j$ with numbers of occurrences $b_j$, then the size of their intersection is deemed to be
\[ 
\sum_{i, j: p_i=q_j} a_i b_j .
\]
As synonyms of ``intersection'' we will also say ``common elements'' or ``overlap.''

\begin{Col} \label{slices}
Let $P(x) = a_x x^2 + b_x x +c_x$ and $Q(y, h) = a_y y^2 + b_y y + c_y +  2^{r-d}h$ where $d \leq r$. Then for any $ v, q \in [d, r]$, 
when we work over $\ZZ_{2^r}$ it holds that:
\[
2^{min(v,q)+d} \;|\; \#\left\{(x, y, h) : x \in \bigcup\limits_{j=0}^{2^v-1} \{l_x+
2^{r-v}j \}, y \in \bigcup\limits_{j=0}^{2^q-1} \p*{l_y+ 2^{r-q}j} , h \in \bigcup\limits_{j=0}^{2^d-1} \p*{j}, P(x)  = Q(y, h) \right\}
\]
for any $l_x, l_y$, where $|$ stands for divides. \hfill  $\Box$
\end{Col}

When proving the base induction step of our main theorem we will come against a specific multiset that a polynomial we will have may
potentially intersect with. The following corollary gives us the divisibility of the size of such an intersection.

\begin{Col} \label{osc}
Let us work over a ring $\ZZ_{2^r}$. Let $P(x) = a_x x^2 + b_x x +c$ with $x$ being constrained to domain $ x \in
\bigcup\limits_{j=0}^{2^q-1} \p*{l_x+  2^{r-q}j}$ for certain $l_x$ and $v \leq r$. Let $S$ be the following multiset
\[
\left(\bigcup_{i=0}^{2^e-1} \bigcup_{f_s=0}^{v-1} f_s \bigcup_{s=0}^{2^{v-f_s-1}-1}  \p*{2^{r-e}i + 2^{r-v+f_s}(2s+1) } \right)
\cup \left(\bigcup_{i=0}^{2^e-1} (v+1) \p*{2^{r-e}i } \right)
\]
where $e \leq min(q, v)$. The number of elements of the intersection (understood as the ``multiplicative'' intersection) of the multiset $S$
and the image of $P$ is divisible by
\[
\hspace{6.91cm} 2^{e + min(q, v, \ceil*{\frac{r}{2}})} . \hspace{6.91cm} \Box
\]
\end{Col}

\begin{reptheorem}{main}
\emph{\textbf{(Base case)}}

For any polynomial $Q$ of $3$ variables ($x, y, z$) over $\ZZ_{2^r}$ of degree up to $2$, any
integers $q, v\leq r$ and $k, w, g, u$ and any linear polynomial $T(x, y)$, it holds that:
\[ 
\sum\limits_{i=0}^{2^q-1} \sum\limits_{j=0}^{2^v-1}\#(k{+}w2^{r-q}i{+}g2^{r-v}j)_{Q, z=T(x, y)+u2^{r-v}j}   = \TT2^{r+ \ceil*{
\frac{\min(2v,r)}{2}} +q-1}
\]
for certain integer $\TT$.
\end{reptheorem}

\begin{proof}
Note that when proving $M$ part in the general induction step we use as induction hypothesis
polynomials over $n-1$ variables, whereas when proving $Q$ part we use polynomials on $n$ variables, just without a square term in $z$.
Therefore for the base step we just need to prove that
\[ 
\sum\limits_{i=0}^{2^q-1} \sum\limits_{j=0}^{2^v-1}\#(k{+}w2^{r-q}i{+}g2^{r-v}j)_{M, z=T(x, y)+u2^{r-v}j}   =
\TT2^{r+\min(v,\ceil*{\frac{r}{2}}) +q-1},
\]
where 
\[ 
M(x, y, z) = P(x, y) + L_z(x, y)z,
\]
$L$ is a linear form, and $P$ has degree up to $2$. The transition from $M$ to $Q$ for $n=3$ is already taken care of by the
general induction step. Here we begin:
\[ 
\sum\limits_{i=0}^{2^q-1} \sum\limits_{j=0}^{2^v-1}\#(k{+}w2^{r-q}i{+}g2^{r-v}j)_{M, z=T(x, y)+u2^{r-v}j}  
 = \sum\limits_{i=0}^{2^q-1} \sum\limits_{j=0}^{2^v-1}\#(k{+}w2^{r-q}i{+}g2^{r-v}j)_{U, z=u2^{r-v}j}
\]
where $U(x, y, z) =   P(x, y) + T(x, y)L_z(x, y) + L_z(x, y)z$
\[
= \sum\limits_{i=0}^{2^q-1} \sum\limits_{j=0}^{2^v-1}\#(k{+}2^{r-q+\wu}i)_{R, z=2^{r-v}j}
\]
where $R(x, y, z) = P(x, y) + T(x, y)L_z(x, y) + L_z'(x, y)z$, $L_z'(x,y) = uL_z(x, y)-g$ and $\wu$ is the order of $w$,
\[
= 2^{min(q, \wu)}\sum\limits_{i=0}^{2^{max(q-\wu, 0)}-1} \sum\limits_{j=0}^{2^v-1}\#(k{+}2^{r-q+\wu}i)_{R, z=2^{r-v}j}
\]
\[
= 2^{min(q, \wu)}\sum\limits_{i=0}^{2^e-1} \sum\limits_{j=0}^{2^v-1}\#(k{+}2^{r-e}i)_{R, z=2^{r-v}j}
\]
where $e = max(q-\wu, 0)$. Let us ignore the $2^{min(q, \wu)}$ coefficient, and focus on the sum:
\[
W = \sum\limits_{i=0}^{2^e-1}\sum\limits_{j=0}^{2^v-1}\#(k{+}2^{r-e}i)_{R, z=2^{r-v}j}
\]
for which we need to show that it is divisible by $2^{r+\min(v,\ceil*{\frac{r}{2}}) + e -1}$. We have
\[
W = \sum\limits_{i=0}^{2^e-1}\sum\limits_{j=0}^{2^v-1} \sum\limits_{l=0}^{2^r-1} \#(k{+}2^{r-e}i-l2^{r-v}j)_{G,
L_z'(x, y)=l}
\]
where $G(x,y) =  P(x, y) + T(x, y)L_z(x, y)$. Let $L'_z(x, y) = \alpha_z2^{\beta_z} y + \zeta_z2^{\eta_z} x + \delta_z$, where $\alpha_z$
and $\zeta_z$ are odd and we assume without the loss of generality that $\beta_z \leq \eta_z$. Then,
\[
W = \sum\limits_{i=0}^{2^e-1}\sum\limits_{j=0}^{2^v-1} \sum\limits_{l=0}^{2^{r-\beta_z}-1} \#(k{+}2^{r-e}i-(2^{\beta_z}l +
\delta_z)2^{r-v}j)_{G, L_z'(x, y)=2^{\beta_z}l + \delta_z}
\]
as always $2^{\beta_z} | L'_z(x, y) - \delta_z$,
\[
= \sum\limits_{i=0}^{2^e-1}\sum\limits_{j=0}^{2^v-1} \sum\limits_{l=0}^{2^{r-\beta_z}-1}
\sum\limits_{s=0}^{2^{\beta_z}-1}
\#(k{+}2^{r-e}i-(2^{\beta_z}l + \delta_z)2^{r-v}j)_{G,  y=\frac{l}{\alpha_z}+ \frac{\zeta_z}{\alpha_z}2^{\eta_z-\beta_z} x + 2^{r-\beta_z}s
}
\]
where we solved $L'_z(x, y) = 2^{\beta_z}l + \delta_z$ for $y$.
Taking $G(x,y) = E(x) + L_y(x)y + my^2$, for certain up to quadratic $E$ and linear $L_y$ we can write:
\[
W = \sum\limits_{i=0}^{2^e-1}\sum\limits_{j=0}^{2^v-1} \sum\limits_{l=0}^{2^{r-\beta_z}-1}
\sum\limits_{s=0}^{2^{\beta_z}-1}
\#(k{+}2^{r-e}i-(\alpha_z2^{\beta_z}l + \delta_z)2^{r-v}j)_{U,  y=l + 2^{r-\beta_z}s
}
\]
where $U(x, y) = E(x) + L_y(x)(\frac{\zeta_z}{\alpha_z}2^{\eta_z-\beta_z}x +y) + m(\frac{\zeta_z}{\alpha_z}2^{\eta_z-\beta_z} x)^2 +
2m\frac{\zeta_z}{\alpha_z}2^{\eta_z-\beta_z}xy+ my^2 = E'(x) + L'_y(x)y + my^2$ for certain up to quadratic $E'$ and linear $L'_y$.
\[
W = \sum\limits_{i=0}^{2^e-1}\sum\limits_{j=0}^{2^v-1}
\sum\limits_{s=0}^{2^r-1}
\#(k{+}2^{r-e}i-O(s)2^{r-v}j)_{U,  y=s}
\]
where $O(s) = \alpha_z2^{\beta_z}(s \bmod 2^{r-\beta_z}) + \delta_z  = \alpha_z2^{\beta_z}s + \delta_z$,
\[
= \sum\limits_{i=0}^{2^e-1}\sum\limits_{j=0}^{2^v-1}
\sum\limits_{s=0}^{2^r-1} \sum\limits_{h=0}^{2^r-1}
\#(k{+}2^{r-e}i-(\alpha_z2^{\beta_z}s + \delta_z)2^{r-v}j-hs-ms^2)_{E',  L'_y(x)=h}
\]
\[
= \sum\limits_{i=0}^{2^e-1}\sum\limits_{j=0}^{2^v-1}
\sum\limits_{s=0}^{2^r-1} \sum\limits_{h=0}^{2^r-1}
\#(k{+}2^{r-e}i-\delta_z2^{r-v}j  -(\alpha_z2^{r-v+\beta_z}j+h)s-ms^2)_{E',  L'_y(x)=h}
\]
\[
= \sum\limits_{i=0}^{2^e-1}\sum\limits_{j=0}^{2^v-1}
\sum\limits_{s=0}^{2^r-1} \sum\limits_{h=0}^{2^r-1}
\#(k{+}2^{r-e}i-2^{r-v+\gamma_z}j  -(\alpha'_z2^{r-v+\beta_z}j+h)s-ms^2)_{E',  L'_y(x)=h}
\]
where $\gamma_z$ is the order of $\delta_z$ and $\alpha'_z$ is $\alpha_z$ divided by the odd factor in $\delta_z$,
{\small
\[
W = \sum\limits_{i=0}^{2^e-1}\sum\limits_{j=0}^{2^v-1}
\sum\limits_{s=0}^{2^r-1} \sum_{f_h = 0}^{r-v+\beta_z-1} \sum_{h=0}^{2^{r-f_h-1}-1}
\#(k{+}2^{r-e}i-2^{r-v+\gamma_z}j  -(\alpha'_z2^{r-v+\beta_z}j+2^{f_h}(1+2h))s-ms^2)_{E',  L'_y(x)=2^{f_h}(1+2h)}
\]
\begin{equation} \label{eq:sums}
+\sum\limits_{i=0}^{2^e-1}\sum\limits_{j=0}^{2^v-1}
\sum\limits_{s=0}^{2^r-1} \sum\limits_{h=0}^{2^{max(v-\beta_z, 0)}-1}
\#(k{+}2^{r-e}i-2^{r-v+\gamma_z}j  -(\alpha'_z2^{r-v+\beta_z}j+2^{r-v+\beta_z}h)s-ms^2)_{E',  L'_y(x)=2^{r-v+\beta_z}h}
\end{equation}
}
Let us focus now on the first of the two of the above sums,
{\small
\begin{equation} \label{eq:sum}
S = \sum\limits_{i=0}^{2^e-1}\sum\limits_{j=0}^{2^v-1}
\sum\limits_{s=0}^{2^r-1} \sum_{f_h = 0}^{r-v+\beta_z-1} \sum_{h=0}^{2^{r-f_h-1}-1}
\#(k{+}2^{r-e}i-2^{r-v+\gamma_z}j  -(\alpha'_z2^{r-v+\beta_z}j+2^{f_h}(1+2h))s-ms^2)_{E',  L'_y(x)=2^{f_h}(1+2h)}
\end{equation} 
}
\[
= \sum\limits_{i=0}^{2^e-1}\sum\limits_{j=0}^{2^v-1}
\sum\limits_{s=0}^{2^r-1} \sum_{f_h = 0}^{r-v+\beta_z-1} \sum_{h=0}^{2^{r-f_h-1}-1}
\#(k{+}2^{r-e}i-2^{r-v+\gamma_z}j  -2^{f_h}s-ms^2)_{E',  L'_y(x)=2^{f_h}(1+2h)}
\]
because $\alpha'_z2^{r-v+\beta_z}j$ plays no role due to $f_h < r-v+\beta_z$. Now we use the fact that $2^{r-e}i$ and $2^{r-v+\gamma_z}j$
``overlap'' each other
\[ 
S = 2^{min(\gamma_z, v) + min(e, max(v-\gamma_z, 0))} \sum\limits_{i=0}^{2^d-1}
\sum\limits_{s=0}^{2^r-1} \sum_{f_h = 0}^{r-v+\beta_z-1} \sum_{h=0}^{2^{r-f_h-1}-1}
\#(k{+}2^{r-d}i  -2^{f_h}s-ms^2)_{E',  L'_y(x)=2^{f_h}(1+2h)}
\]
where $d = max(e, v-\gamma_z)$. Let $\mu$ be the order of $m$. First let us consider the part of $S$ where $f_h \leq min(r-d,
\mu)$ 
\[ 
S_1 = 2^{min(\gamma_z, v) + min(e, max(v-\gamma_z, 0))} \sum\limits_{i=0}^{2^d-1}
\sum\limits_{s=0}^{2^r-1} \sum_{f_h = 0}^{l} \sum_{h=0}^{2^{r-f_h-1}-1}
\#(k{+}2^{r-d}i  -2^{f_h}s-ms^2)_{E',  L'_y(x)=2^{f_h}(1+2h)}
\]
where $l = min(r-v+\beta_z-1, r-d, \mu)$,
\[ = 2^{min(\gamma_z, v) + min(e, max(v-\gamma_z, 0))} \sum\limits_{i=0}^{2^d-1}
\sum\limits_{s=0}^{2^r-1} \sum_{f_h = 0}^{l} \sum_{h=0}^{2^{r-f_h-1}-1}
\#(k{+}2^{r-d}i  -2^{f_h+g_f}s)_{E',  L'_y(x)=2^{f_h}(1+2h)}
\]
where $g_f = 0$ or $g_f=1$ depending on $f_h$ ($g_f=1$ only when $r-d \geq \mu = f_h$). We could write the above owing to Lemma \ref{types},
since the polynomial $-2^{f_h}s -ms^2$ on $s$ is of either type a) or b).
\[ S_1 = 2^{min(\gamma_z, v) + min(e, max(v-\gamma_z, 0))} \sum_{f_h = 0}^{l}  2^{f_h+g_f}  
\sum\limits_{s=0}^{2^{r-f_h-g_f}-1} \sum\limits_{i=0}^{2^d-1} \sum_{h=0}^{2^{r-f_h-1}-1} 
\#(k{+}2^{r-d}i  -2^{f_h+g_f}s)_{E',  L'_y(x)=2^{f_h}(1+2h)}
\]
{\small
\[ = 2^{min(\gamma_z, v) + min(e, max(v-\gamma_z, 0))} \sum_{f_h = 0}^{l}  2^{min(r, d+f_h+g_f)}  
\sum\limits_{s=0}^{2^{max(r-f_h-g, d)}-1} \sum_{h=0}^{2^{r-f_h-1}-1} 
\#(k{+}2^{min(f_h+g_f, r-d)}s)_{E',  L'_y(x)=2^{f_h}(1+2h)}
\]
}
\[ =2^{e + v} \sum_{f_h = 0}^{l}  2^{min(r-d, f_h+g_f)}  
\sum\limits_{s=0}^{2^{max(r-f_h-g_f, d)}-1} \sum_{h=0}^{2^{r-f_h-1}-1} 
\#(k{+}2^{min(f_h+g_f, r-d)}s)_{E',  L'_y(x)=2^{f_h}(1+2h)}
\]
as $min(\gamma_z, v) + min(e, max(v-\gamma_z, 0))+d = e+v$, recall that $d = max(e, v-\gamma_z)$.

Now we can use Corollary \ref{slices}. Solving $L'_y(x)=2^{f_h}(1+2h)$ for $x$ will give us a certain slice in which $x$ has to be, and the
size of that slice will be at least $2^{r-f_h-1}$. Therefore the expression 
\[
\sum_{s=0}^{2^{max(r-f_h-g_f, d)}-1}  \sum_{h=0}^{2^{r-f_h-1}-1}
\#(k-2^{min(f_h+g_f, r-d)}s)_{E',  L'_y(x)=2^{f_h}(1+2h)} 
\]
equals the size of the common part of image of $E'(x)$ with $x$ in the just mentioned slice,
and image of function $k{+}2^{min(f_h+g_f, r-d)}s$ (with $s$ in a slice of size $2^{max(r-f_h-g_f, d)}$). Due to Corollary
\ref{slices}, this expression is divisible by $2^{r-f_h-1}$ (size of smaller domain, $d=0$ in the corollary use). Therefore, and because
$r-d \geq l$, we have 
\[
S_1  = 2^{e + v} \sum_{f_h = 0}^{l} 2^{min(r-d, f_h+g_f)} \TT_{f_h} 2^{r-f_h-1}
=  2^{r + e + v -1} \sum_{f_h = 0}^{l} 2^{u_f} \TT_{f_h},
\]
where $u_f =0$ or $1$.  This gives possibly even more than the desired divisibility.

When moving to the part of $S$, written as the sum (\ref{eq:sum}), where $ r-v+\beta_z > f_h > min(r-d, \mu)$, let us consider two cases,
either $r-d \leq \mu$ or $r-d > \mu$.
Let us start with the first of them, and let us consider it now together with the second sum from $W$  written in the form
\eqref{eq:sums}
\[
K = 2^{min(\gamma_z, v) + min(e, max(v-\gamma_z, 0))} \sum\limits_{i=0}^{2^{d} -1}
\sum\limits_{s=0}^{2^r-1} \sum_{f_h = r-d+1}^{r-v+\beta_z-1} \sum_{h=0}^{2^{r-f_h-1}-1}
\#(k{+}2^{r-d}i  -2^{f_h}s -ms^2)_{E',  L'_y(x)=2^{f_h}(1+2h)}
\]
\[
+ \sum\limits_{i=0}^{2^e-1}\sum\limits_{j=0}^{2^v-1}
\sum\limits_{s=0}^{2^r-1} \sum\limits_{h=0}^{2^{max(v-\beta_z, 0)}-1}
\#(k{+}2^{r-e}i-2^{r-v+\gamma_z}j  -(\alpha'_z2^{r-v+\beta_z}j+2^{r-v+\beta_z}h)s-ms^2)_{E',  L'_y(x)=2^{r-v+\beta_z}h}
\]
\[
= 2^{r + min(\gamma_z, v) + min(e, max(v-\gamma_z, 0))} \sum\limits_{i=0}^{2^{d} -1}
\sum_{f_h = r-d+1}^{r-v+\beta_z-1} \sum_{h=0}^{2^{r-f_h-1}-1}
\#(k{+}2^{r-d}i)_{E',  L'_y(x)=2^{f_h}(1+2h)}
\]
\[
+ 2^{r + min(\gamma_z, v) + min(e, max(v-\gamma_z, 0))} \sum\limits_{i=0}^{2^{d} -1}
\sum\limits_{h=0}^{2^{max(v-\beta_z, 0)}-1}
\#(k{+}2^{r-d}i)_{E',  L'_y(x)=2^{r-v+\beta_z}h},
\]
noting that both of $2^{f_h}s$ and $-(\alpha'_z2^{r-v+\beta_z}j+2^{r-v+\beta_z}h)s-ms^2$ have smaller ``granularity'' than $2^{r-d}i$ (the
latter, because we assumed for this case that $\mu \geq r-d$ and indirectly that $\beta_z > \gamma_z$).
Continuing,
\[
K = 2^{r + min(\gamma_z, v) + min(e, max(v-\gamma_z, 0))} \sum\limits_{i=0}^{2^{d} -1}
\sum\limits_{h=0}^{2^{d-1}-1}
\#(k{+}2^{r-d}i)_{E',  L'_y(x)=2^{r-d+1}h}
\]
(in this case we have an assumption that $d \geq 1$, due to $f_h > r-d$)
\[
= 2^{r + min(\gamma_z, v) + min(e, max(v-\gamma_z, 0))} \TT 2^{d-1} = \TT 2^{r + e + v -1},
\]
which gives the desired divisibility. We once again used Corollary \ref{slices}, in the same way as for the previous case.

Let us now move to the case when $r-d > \mu$ :
\[
K = 2^{min(\gamma_z, v) + min(e, max(v-\gamma_z, 0))} \sum\limits_{i=0}^{2^{d} -1}
\sum\limits_{s=0}^{2^r-1} \sum_{f_h = \mu+1}^{r-v+\beta_z-1} \sum_{h=0}^{2^{r-f_h-1}-1}
\#(k{+}2^{r-d}i  -2^{f_h}s -ms^2)_{E',  L'_y(x)=2^{f_h}(1+2h)}
\]
\[
+ \sum\limits_{i=0}^{2^e-1}\sum\limits_{j=0}^{2^v-1}
\sum\limits_{s=0}^{2^r-1} \sum\limits_{h=0}^{2^{max(v-\beta_z, 0)}-1}
\#(k{+}2^{r-e}i-2^{r-v+\gamma_z}j  -(\alpha'_z2^{r-v+\beta_z}j+2^{r-v+\beta_z}h)s-ms^2)_{E',  L'_y(x)=2^{r-v+\beta_z}h}
\]
\[
= 2^{min(\gamma_z, v) + min(e, max(v-\gamma_z, 0))} \sum\limits_{i=0}^{2^{d} -1}
\sum\limits_{s=0}^{2^r-1} \sum_{f_h = \mu+1}^{r-v+\beta_z-1} \sum_{h=0}^{2^{r-f_h-1}-1}
\#(k{+}2^{r-d}i -ms^2)_{E',  L'_y(x)=2^{f_h}(1+2h)}
\]
\[
+  2^{min(\gamma_z, v) + min(e, max(v-\gamma_z, 0))} \sum\limits_{i=0}^{2^d-1}
\sum\limits_{s=0}^{2^r-1} \sum\limits_{h=0}^{2^{max(v-\beta_z, 0)}-1}
\#(k{+}2^{r-d}i -ms^2)_{E',  L'_y(x)=2^{r-v+\beta_z}h}
\]
(since both polynomials on $s$ are of type $c)$, per Lemma \ref{types})
\[
=  2^{min(\gamma_z, v) + min(e, max(v-\gamma_z, 0))}  \sum\limits_{i=0}^{2^d-1}
\sum\limits_{s=0}^{2^r-1} \sum\limits_{h=0}^{2^{r-\mu-1}-1}
\#(k{+}2^{r-d}i -2^\mu s^2)_{E',  L'_y(x)=2^{\mu+1}h}
\]
\[
=  2^{min(\gamma_z, v) + min(e, max(v-\gamma_z, 0)) + \mu}  \sum\limits_{i=0}^{2^d-1}
\sum\limits_{s=0}^{2^{r-\mu}-1} \sum\limits_{h=0}^{2^{r-\mu-1}-1}
\#(k{+}2^{r-d}i -2^\mu s^2)_{E',  L'_y(x)=2^{\mu+1}h}
\]

\[
=  2^{min(\gamma_z, v) + min(e, max(v-\gamma_z, 0))+\mu} \TT 2^{r-\mu-1+d},
\]
where we used Corollary \ref{slices} remembering that $r-d > \mu \Leftrightarrow r-\mu-1 \geq d$.  So finally
\[
K =  2^{r + e + v -1} \TT
\]
which in this case gives possibly even more than the required divisibility.

Now we are left with the situation when $min(r-d, \mu) \geq r-v+\beta_z$, and we already know that $S$, \ie the first sum from $W$ written
in the form \eqref{eq:sums}, has the desired divisibility. Let us look now at the second of the two sums of $W$:
\[
C = \sum\limits_{i=0}^{2^e-1}\sum\limits_{j=0}^{2^v-1}
\sum\limits_{s=0}^{2^r-1} \sum\limits_{h=0}^{2^{max(v-\beta_z, 0)}-1}
\#(k{+}2^{r-e}i-2^{r-v+\gamma_z}j  -(\alpha'_z2^{r-v+\beta_z}j+2^{r-v+\beta_z}h)s-ms^2)_{E',  L'_y(x)=2^{r-v+\beta_z}h}
\]

Because $r-v+\gamma_z \geq r-v+\beta_z$, we can let $\sigma_z = \gamma_z - \beta_z$, where $\sigma_z \geq 0$. Then we have:
\[
C =\sum\limits_{i=0}^{2^e-1}\sum\limits_{j=0}^{2^v-1}
\sum\limits_{s=0}^{2^r-1} \sum\limits_{h=0}^{2^{max(v-\beta_z, 0)}-1}
\#(k{+}2^{r-e}i- 2^{r-v+\beta_z} (\alpha''_z2^{\sigma_z}j  + (j+h)s)-ms^2)_{E',  L'_y(x)=2^{r-v+\beta_z}h},
\]
where $\alpha''_z = \alpha'^{-1}_z$.
Because
\[ \alpha''_z2^{\sigma_z}j  + (j+h)s = (j+h)(\alpha''_z2^{\sigma_z} + s) - \alpha''_z2^{\sigma_z}h, \]
we obtain
\[
C =\sum\limits_{i=0}^{2^e-1}\sum\limits_{j=0}^{2^v-1}
\sum\limits_{s=0}^{2^r-1} \sum\limits_{h=0}^{2^{max(v-\beta_z, 0)}-1}
\#(k{+}2^{r-e}i- 2^{r-v+\beta_z} (j+h)(\alpha''_z2^{\sigma_z} + s)-ms^2 - \alpha''_z2^{r-v+\gamma_z}h)_{E',  L'_y(x)=2^{r-v+\beta_z}h}.
\]
We are operating now under an assumption that $r-v+\beta_z \leq \mu$, therefore $C$ equals
{\small
\[ 2^{r-v} \sum\limits_{i=0}^{2^e-1}\sum\limits_{j=0}^{2^v-1}
\sum\limits_{s=0}^{2^v-1} \sum\limits_{h=0}^{2^{max(v-\beta_z, 0)}-1}
\#(k{+}2^{r-e}i- 2^{r-v+\beta_z} ((j+h)(\alpha''_z2^{\sigma_z} + s)- \md2^{\mu'} s^2) - \alpha''_z2^{r-v+\gamma_z}h)_{E',
L'_y(x)=2^{r-v+\beta_z}h},
\]
}
where $\mu' = \mu - (r-v+\beta_z)$, and $\md$ is the odd factor in $m$. Let us notice now that:
\[
\bigcup_{j, s \in Z_{2^v}} \p*{ (j+h)(\alpha''_z2^{\sigma_z} + s)- \md2^{\mu'} s^2 }
= \bigcup_{j, s \in Z_{2^v}} \p*{ \alpha''_z2^{\sigma_z}(j+h) + s( j+h - \md2^{\mu'}s) }  
\]
\[
= \bigcup_{s, j \in Z_{2^v}} \p*{ \alpha''_z2^{\sigma_z}(j+h + \md2^{\mu'}s) + s( j+h) }
\]
where we changed the starting point of the iteration on $j$ by $\md2^{\mu'}s$,
\[
= \bigcup_{s, j \in Z_{2^v}} \p*{ \alpha''_z2^{\sigma_z}(j+h) + s( j+h + \md2^{\mu'}\alpha''_z2^{\sigma_z}) }
= \bigcup_{s, j \in Z_{2^v}} \p*{ \alpha''_z2^{\sigma_z}(j+h - \md2^{\mu'}\alpha''_z2^{\sigma_z} ) + s( j+h) }
\]
where we changed the starting point of the iteration on $j$ by an additional $\md2^{\mu'}\alpha''_z2^{\sigma_z}$,
\[
= \bigcup_{s, j \in Z_{2^v}} \left\{ \alpha''_z2^{\sigma_z}(h - \md\alpha''_z2^{\mu' +\sigma_z}) + \alpha''_z2^{\sigma_z}j+ s( j+h)
\right\}
\]
\[
= \bigcup_{s, j \in Z_{2^v}} \p*{ \alpha''_z2^{\sigma_z}(h - \md\alpha''_z2^{\mu' +\sigma_z}) + \alpha''_z2^{\sigma_z}(j-h)+ sj }
\]
where we changed the starting point of the iteration on $j$ by an additional $h$,
\[
= \bigcup_{s, j \in Z_{2^v}} \left\{ -\md\alpha''^2_z2^{\mu'+2\sigma_z} + j(s + \alpha''_z2^{\sigma_z})
\right\}
= \bigcup_{s, j \in Z_{2^v}} \p*{ -\md\alpha''^2_z2^{\mu'+2\sigma_z} + js}.
\]
The last follows because $\bigcup_{s \in Z_{2^v}} \p*{s + \alpha''_z2^{\sigma_z}} = \bigcup_{s \in Z_{2^v}} \p*{s}$.
Going back to our original formula we can write that:
\[
C = 2^{r-v} \sum\limits_{i=0}^{2^e-1}\sum\limits_{j=0}^{2^v-1}
\sum\limits_{s=0}^{2^v-1} \sum\limits_{h=0}^{2^{max(v-\beta_z, 0)}-1}
\#(k{+}2^{r-e}i- 2^{r-v+\beta_z} (js -\md\alpha''^2_z2^{\mu'+2\sigma_z}) - \alpha''_z2^{r-v+\gamma_z}h)_{E',
L'_y(x)=2^{r-v+\beta_z}h}
\]
\[
= 2^{r-v} \sum\limits_{i=0}^{2^e-1}\sum\limits_{j=0}^{2^v-1}
\sum\limits_{s=0}^{2^v-1} \sum\limits_{h=0}^{2^{max(v-\beta_z, 0)}-1}
\#(k'{+}2^{r-e}i- 2^{r-v+\beta_z} (js - \alpha''_z2^{\sigma_z}h))_{E',
L'_y(x)=2^{r-v+\beta_z}h}
\]
(where $k' = k -2^{r-v+\beta_z}\md\alpha''^2_z2^{\mu'+2\sigma_z}$)
\[
=2^{r-v+2min(v, \beta_z)} \sum\limits_{i=0}^{2^e-1}\sum\limits_{s,j=0}^{2^{max(v', 0)}-1}
\sum\limits_{h=0}^{2^{max(v', 0)}-1}
\#(k'{+}2^{r-e}i- 2^{r-v'} (js- \alpha''_z2^{\sigma_z}h))_{E',
L'_y(x)=2^{r-v'}h},
\]
where $v' = v-\beta_z$. If $v' \leq 0 \Leftrightarrow v \leq \beta_z$, then $e=0$, since $min(r-d, \mu) \geq r-v+\beta_z$ in this case, and
$d = max(e, v-\gamma_z)$. In such a case $C$ becomes
\[
2^{r+v} \#k'_{E', L'_y(x)=0},
\]
which trivially divides by $2^{r+e+min\left(v, \ceil*{\frac{r}{2}}\right)-1}$. Let us continue our proof now with the assumption that $v' > 0$, for which our formula is
\[
C =2^{r-v+2 \beta_z} \sum\limits_{i=0}^{2^e-1}\sum\limits_{s,j=0}^{2^{v'}-1}
\sum\limits_{h=0}^{2^{v'}-1}
\#(k'{+}2^{r-e}i- 2^{r-v'} (js- \alpha''_z2^{\sigma_z}h))_{E',
L'_y(x)=2^{r-v'}h}.
\]

We use Lemma \ref{multi} on $js$, and $j$ and $s$ are now confined to the ring $\ZZ_{2^{v'}}$:
\[
C =2^{r-v+2\beta_z} \sum\limits_{i=0}^{2^e-1}\sum\limits_{s=0}^{2^{v'}-1}
\sum\limits_{h=0}^{2^{v'}-1} (O(s)+1)2^{v'-1}
\#(k'{+}2^{r-e}i- 2^{r-v'} (s - \alpha''_z2^{\sigma_z}h))_{E',
L'_y(x)=2^{r-v'}h},
\]
(where $O(s)$ is the order of $s$, unless $s=0$, when it is the order (\ie $v'$) increased by $1$)
\[
=2^{r+\beta_z-1} \sum\limits_{i=0}^{2^e-1}\sum\limits_{s=0}^{2^{v'}-1}
\sum\limits_{h=0}^{2^{v'}-1} (O(s)+1)
\#(k'{+}2^{r-e}i- 2^{r-v'} (s - \alpha''_z2^{\sigma_z}h))_{E',
L'_y(x)=2^{r-v'}h}
\]
\[
=2^{r+\beta_z-1} \sum\limits_{i=0}^{2^e-1}\sum\limits_{s=0}^{2^{v'}-1}
\sum\limits_{h=0}^{2^{v'}-1} (O(s)+1)
\#(k'{+}2^{r-e}i- 2^{r-v'} s )_{E',
L'_y(x)=2^{r-v'}h}.
\]
The last follows because for any $h$, 
\[
\bigcup_{s=0}^{2^{v'}-1} \p*{2^{r-v'}s} = \bigcup_{s=0}^{2^{v'}-1} \p*{2^{r-v'}(s - \alpha''_z2^{\sigma_z}h)}.
\]
To increase clarity, let us omit the $2^{r+\beta_z-1} $ coefficient.  We need to show  that
\[
2^{e + min(v', \ceil*{\frac{r}{2}})}
 \]
divides the remaining sum
\[ 
C_1 = \sum\limits_{i=0}^{2^e-1}\sum\limits_{s=0}^{2^{v'}-1}
\sum\limits_{h=0}^{2^{v'}-1} (O(s)+1)
\#(k'{+}2^{r-e}i- 2^{r-v'}s )_{E',
L'_y(x)=2^{r-v'}h},
\]
or even just
\[ 
C_2 = \sum\limits_{i=0}^{2^e-1}\sum\limits_{s=0}^{2^{v'}-1}
\sum\limits_{h=0}^{2^{v'}-1} O(s)
\#(k'{+}2^{r-e}i- 2^{r-v'} s )_{E',
L'_y(x)=2^{r-v'}h},
\]
since it is easy to show using our earlier techniques that
\[ 
\sum\limits_{i=0}^{2^e-1}\sum\limits_{s=0}^{2^{v'}-1}
\sum\limits_{h=0}^{2^{v'}-1}
\#(k'{+}2^{r-e}i- 2^{r-v'} s )_{E',
L'_y(x)=2^{r-v'}h}
\]
is divisible by $2^{e + v'}$ (via Corollary \ref{slices}).

Focusing on $C_2$, let us split it into the possible orders $f_s$ of $s$:
\[ 
C_2 = \sum\limits_{i=0}^{2^e-1}\sum\limits_{s=0}^{2^{v'}-1}
\sum\limits_{h=0}^{2^{v'}-1} O(s) \#(k'{+}2^{r-e}i- 2^{r-v'} s )_{E', L'_y(x)=2^{r-v'}h}
\]
\[
= \left(\sum\limits_{h=0}^{2^{v'}-1}\sum_{i=0}^{2^e-1} \sum_{f_s=0}^{v'-1} f_s \sum_{s=0}^{2^{v'-f_s-1}-1}  \#(k'{+}2^{r-e}i- 2^{r-v'}
(2s+1) )_{E', L'_y(x)=2^{r-v'}h} \right) 
\]
\[
+ \left(\sum\limits_{h=0}^{2^{v'}-1} \sum_{i=0}^{2^e-1} (v'+1)\#(k'{+}2^{r-e}i)_{E', L'_y(x)=2^{r-v'}h} \right).
\]
Solving $L'_y(x)=2^{r-v'}h$ for $x$ with $h \in \ZZ_{2^{v'}}$ confines $x$ to a certain slice of size at least
$2^{v'}$, which then becomes the domain of $E'$.
Therefore the expression
\[ 
\sum_{i=0}^{2^e-1} \sum_{f_s=0}^{v'-1} f_s \sum_{s=0}^{2^{v'-f_s-1}-1}  \#(k'{+}2^{r-e}i- 2^{r-v'}(2s+1) )_{E', L'_y(x)=2^{r-v'}h} 
\]
equals the size of the common part of the image of $E'(x)$ with $x$ in the just mentioned slice,
and the multiset
\[ 
\bigcup_{i=0}^{2^e-1} \bigcup_{f_s=0}^{v'-1} f_s \bigcup_{s=0}^{2^{v'-f_s-1}-1}  \p*{k'{+}2^{r-e}i- 2^{r-v'} (2s+1)} .
\]
This works analogously to the second component of the sum above, and allows us to directly use Corollary \ref{osc}, recalling that $e \leq v'$ as $r-e
\geq r-v+\beta_z$ in this case. The polynomial $E'$ corresponds to $P$, and the multiset we obtain here from both components of the sum
corresponds to the multiset $S$ (we just need to shift both of them by $k'$).
By use of the corollary we obtain that our formula is divisible by
\[ 
2^{e + min(v', \ceil*{\frac{r}{2}})}, 
\]
which at last concludes the proof of the base case and the entire theorem from the stated lemmas.
\end{proof}

\bigskip
We remark that the most difficult juncture of the above proof seems to be the treatment of multiplicities in the intersections.  Even in better-behaved cases of algebraic varieties defined by polynomials over fields, {\em intersection theory\/} is known as a relatively difficult subject.  It is possible that carrying over some of this theory to $\ZZ_{2^r}$ may improve the conceptual highness of abstraction in the proof, but we have not seen how to do this.  

The rest of this chapter---amounting to most of it---gives proofs of the lemmas and corollaries stated earlier in this section, as well as
some new ones that are needed for their proofs. This requires more situational analysis of intersections and multiplicities.
 
\begin{replemma}{multi}
\emph{\textbf{(restated)}}
For any $m \in \ZZ_{2^r}$, 
\begin{equation*}
\#\p*{x,y \in \ZZ_{2^r}: xy = m} =
\begin{cases}
(o(m)+1)2^{r-1} & \text{\;if $m \neq 0$}.\\
(r+2)2^{r-1} & \text{\;otherwise}.
\end{cases}
\end{equation*}

\end{replemma}
\begin{proof}
Let us start with the case $m \neq  0$. Let us notice that
\[
\forall_{t < r} \#\p*{m \in \ZZ_{2^r} : o(m) = t} = 2^{r-t-1}  
\text{\;\;\;\;\;\;\;\;and\;\;\;\;\;\;\;\;\;} xy = m \Rightarrow o(x) + o(y) =o(m).
\]
To find the number of pairs $x,y$ such that $o(x) + o(y) = o(m)$ , we can start by first taking pairs where $o(x) = 0$ and $o(y) = o(m)$, and go all they way until: $o(x) =
o(m)$ and $o(y) = 0$. It gives us:
\[
\forall_{t < r} \sum_{m \in \ZZ_{2^r}, o(m) = t} \#\p*{x,y \in \ZZ_{2^r}: xy = m}  = \sum_{i=0}^t 2^{r-1-i} 2^{r-t-1+i}
\]
Because of symmetry,
\[
\#\p*{x,y \in \ZZ_{2^r}: xy = m} = \frac{\sum_{m' \in \ZZ_{2^r}, o(m') = o(m)} \#\p*{x,y \in \ZZ_{2^r}: xy = m'}}{\#\p*{m' \in Z_{2^r} : o(m') = o(m)}}.
\]
Therefore, with $t=o(m)$:
\[
\#\p*{x,y \in Z_{2^r}: xy = m} = \frac{\sum_{i=0}^t 2^{2r-2-t}}{2^{r-t-1}} = (t+1)2^{r-1}.
\]

\bigskip
For $m = 0$, let us just subtract from all pairs, those pairs for cases when $m > 0$:
\[
\#\p*{x,y \in Z_{2^r}: xy = 0} = 2^{2r} - \sum_{t=0}^{r-1} \sum_{i=0}^t 2^{2r-2-t}, 
\]
which after some transformations (including use of the formula for the sum of arithmetic-geometric series) gives the desired result.
\end{proof}

\begin{lemma} \label{sqrr}
For any $t \in \ZZ_{2^r}$ there is a $k \in \ZZ_{2^r}$ such that
\[
t^2 = 2^{2o(t)} + 2^{2o(t)+3}k.
\]
Additionally following holds:
\[
\forall_{a, k} \exists_{t}^{=n}: t^2 = 2^{2a} + 2^{2a+3}k,
\]
where:
\[ 
n = \left\{
\begin{array}{ll}
2^{\min(a+2, r-a-1)}   & \text{if }  a < \frac{r}{2}\\
2^{\lfloor \frac{r}{2} \rfloor} & \text{otherwise}
\end{array}
\right.
\]
Moreover, when $a < \frac{r}{2}$ then the order of all such $t$-s equals $a$.
\end{lemma}
\begin{proof}

Let $q$ be the order of $t$. Then $t = 2^q(1+2m)$ for some $m$, and $t^2 = 2^{2q} + 2^{2q+2}m(m+1) = 2^{2q} + 2^{2q+3}k$ for certain $k$.  This proves the first part of the lemma.

Now let us show that
\[
\forall_{a< \frac{r}{2}, k} \exists_{t\in \ZZ_{2^r}}: t^2 \equiv 2^{2a} + 2^{2a+3}k\;(\text{mod } 2^r).
\]
We choose $t = 2^a(1+2m)$, and obtain:
\[ 
(2^a(1+2m))^2 \equiv 2^{2a}+2^{2a+3}k \;(\text{mod } 2^r),
\]
\[ 
1+4m+4m^2 \equiv 1+8k \;(\text{mod } 2^{r-2a}).
\]
If $r-2a < 3$ the above is true for any $m$ and $k$. Otherwise:
\[ 
m^2+m \equiv 2k \;(\text{mod } 2^{r-2a-2}).
\]
We show that the last statement from above is true (\ie for any $k$ there is an $m$ making it true)
through induction on $r$ and the use of Hensel lifting.
Let us take $P(m) = m^2+m-2k$, and let us note that $\forall_m P'(m) \not\equiv 0$ modulo
any non-zero power of $2$ ($P'$ being the derivative of $P$). This allows us to use Hensel's
lemma, which in this case says that if
\[ 
m^2+m-2k  \equiv 0 \;(\text{mod } 2^{r-1}) 
\]
has a solution, then also
\[ 
m^2+m-2k  \equiv 0 \;(\text{mod } 2^{r}) 
\]
does. Checking the base case of $r=1$ is trivial.
Therefore we have that $\forall_{a < \frac{r}{2}, k} \exists_{t}: t^2 = 2^{2a} + 2^{2a+3}k$, and that $o(t) = a$.
Let us take $a < \frac{r}{2} -1$, then:
\[ 
t^2 \equiv ( 2^{r-a-1} + t)^2 \equiv  (2\cdot2^{r-a-1} + t)^2 \equiv \ldots \equiv ((2^{a+1}-1)\cdot2^{r-a-1} + t)^2
\]
\[ 
\equiv  (2^{r-a-1} -t )^2  \equiv  (2\cdot2^{r-a-1} -t)^2 \equiv \ldots \equiv (2^{a+1}\cdot2^{r-a-1} -t)^2
\;(\text{mod } 2^{r}).
\]
The numbers in those squares are all different (to wit, $t$ and $2^{r-a-1} -t$ are different, all
because $a < \frac{r}{2} - 1$). There are exactly $2^{a+2}$ of those numbers. Those are also all
such numbers whose squares equal $2^{2a} + 2^{2a+3}k$. It is because for a given $a$ there are $2^{r-a-1}$ numbers with order $a$, and there
are also $2^{r-2a-3}$ possible values to which squares of them can evaluate. If we find for any such square
value $2^{a+2}$-many $t$'s that evaluate to it, then we obtain $2^{r-2a-3 + a+2} = 2^{r-a-1}$, which means
we have found all such $t$'s.

The remaining options for $a < \frac{r}{2}$ are when $a = \frac{r-1}{2}$ (\ie $r$ is odd) or $a = \frac{r}{2} -1$
(\ie $r$ is even).
For $a = \frac{r-1}{2}$ we have
\[ 
t^2 \equiv ( 2^{r-a} + t)^2 \equiv  (2\cdot2^{r-a} + t)^2 \equiv \ldots \equiv ((2^{a}-1)\cdot2^{r-a} + t)^2,
\]
and there are exactly $2^{a} = 2^{r-a-1}$ of those numbers.
Meanwhile for $a = \frac{r}{2}-1$ we obtain
\[ 
t^2 \equiv ( 2^{r-a-1} + t)^2 \equiv  (2\cdot2^{r-a-1} + t)^2 \equiv \ldots \equiv ((2^{a+1}-1)\cdot2^{r-a-1} + t)^2
\;(\text{mod } 2^{r}),
\]
and there are exactly $2^{a+1} = 2^{r-a-1}$ of those numbers. In the last two cases, it is straightforward to see,
that we have found all applicable numbers $t$, after all there are exactly $2^{r-a-1}$ numbers of
order $a$.

Combining the above results, we can write that when $a < \frac{r}{2}$, we get $2^{\min(a+2, r-a-1)}$ of
the numbers $t$ that we look for.

Now let us look at the case when $a \geq \frac{r}{2}$. Then any square of a number of order $a$ has to evaluate
to $0$. And there are $2^{r-\ceil{\frac{r}{2}}}  = 2^{\floor*{\frac{r}{2}}}$ numbers that have such orders. No number of any lower
order can have its square evaluate to $0$.
\end{proof}

\begin{replemma} {types}
\emph{\textbf{(restated)}}

Let us take a polynomial $P(x) = ax^2 + bx +c$ over $\ZZ_{2^r}$. Let $a = q2^w$ and $b = g2^h$ such that $q$ and $g$ are odd and $w, h$ are
orders of respectively $a$ and $b$.
Let $m = min(w, h)$. The image of $P(x)$ treated as a multiset equals:
\begin{enumerate}[label=\alph*)]
  \item If $w > h$ :
  \[ 2^m \bigcup_{i=0}^{2^{r-m}-1} \p*{ 2^m i + c} \]
  
  \item If $w = h$ :
  
\[ 
\begin{array}{ll}
2^{m+1} \bigcup\limits_{i=0}^{2^{r-m-1}-1} \p*{ 2^{m+1} i + c}   & \text{if }  m < r\\
2^r\p*{c} & \text{if } m = r
\end{array}
\]

\item If $w < h$ :\\
\[
\left(\bigcup_{f=0}^{\ceil*{\frac{r-m}{2}}-1} 2^{min(f+2, r-f-1)+min(m, max(0, r-2f-3))} \bigcup_{i=0}^{max(0, 2^{r-2f-3-m}-1)}
\p*{2^{2f+3+m}i + q2^{2f+m} + t } \right) 
\]
\[
\cup\; 2^{\floor*{\frac{r+m}{2}}}\p*{t}
\]
where $ t= c-\frac{b^2}{2^{m+2}q}$.

\end{enumerate}
\end{replemma}

\begin{proof}
We prove each of the cases on its own:

\begin{enumerate}[label=\alph*)]
  
  \item $w > h$:\\
  First we show that:
  \[ \forall_{i\in \ZZ_{2^r}} \exists_{x\in \ZZ_{2^r}} : q2^{w-m}x^2 + gx \equiv i\;(\text{mod } 2^r) \]
  We do this by induction on $r$ with the use of Hensel's lemma. The base step for $r=1$ is easy to check. For the general step, let us take $Q(x) =
  q2^{w-m}x^2 + gx -i$ and notice that regardless of $x$, $Q'(x)$ is always odd.   This means, via Hensel's lemma, that if $t$ is a
  solution of $Q(x)$  over $2^r$, then there is also a unique solution $S(t)$ that solves $Q(x)$ over $2^{r+1}$, and $S$ is a
  bijection.
  Due to the above we obtain that
  \[ \bigcup_{x=0}^{2^{r}-1} \p*{ q2^{w-m}x^2 + gx} = \bigcup_{i=0}^{2^{r}-1} \p*{  i}, \]
  which after multiplying both sides by $2^m$ and adding $c$ gives us the expected result.
  
  \item $w = h$:\\
  Analogous to the above, by the use of Hensel's lemma we obtain that:
  \[ \forall_{i\in \ZZ_{2^r}} \exists_{x\in \ZZ_{2^r}}^{=2} : qx^2 + gx \equiv 2i\;(\text{mod } 2^r), \]
  which allows us to write:
  \[ \bigcup_{x=0}^{2^{r}-1} \p*{ qx^2 + gx} = 2\bigcup_{i=0}^{2^{r-1}-1} \p*{  2i}, \]
  which after multiplying both sides by $2^m$ and adding $c$ gives us the desired result.
  
  \item $w < h$:\\
  Let us start by removing $c$ from $P$ and dividing it by $2^m$ (we will introduce these factors back later):
  \[
  qx^2 + g2^{h-m}x = q\left(x^2 + \frac{g}{q}2^{h-m}x\right) = q\left(\left(x + \frac{g}{q}2^{h-m-1}\right)^2 -
  \left(\frac{g}{q}2^{h-m-1} \right)^2 \right),
  \]
  \[
  \bigcup_{x=0}^{2^{r}-1} \p*{ \left(x + \frac{g}{q}2^{h-m-1}\right)^2 }
  = \bigcup_{x=0}^{2^{r}-1} \p*{ x^2},
  \]
    and by use of Lemma \ref{sqrr},
  \[
  =\left(\bigcup_{f=0}^{\ceil*{\frac{r}{2}}-1} 2^{min(f+2, r-f-1)} \bigcup_{i=0}^{max(0, 2^{r-2f-3}-1)} 
  \p*{ 2^{2f+3} i + 2^{2f} } \right) \cup 2^{\floor*{\frac{r} {2}}}  \p*{ 0 }.
  \]
  After shifting the elements by $\left(\frac{g}{q}2^{h-m-1} \right)^2 = \left(\frac{b}{2^{m+1}q}\right)^2$ and multiplying them
  by $q$, this gives:
  \[
  \left(\bigcup_{f=0}^{\ceil*{\frac{r}{2}}-1} 2^{min(f+2, r-f-1)} \bigcup_{i=0}^{max(0, 2^{r-2f-3}-1)} 
  \p*{ q\left(2^{2f+3} i + 2^{2f} -\left(\frac{b}{2^{m+1}q}\right)^2 \right) } \right) 
  \]
  \[ 
  \cup 2^{\floor*{\frac{r} {2}}} \p*{ -q\left(\frac{b}{2^{m+1}q} \right)^2 },
  \]
  which after multiplying the elements by $2^m$ and then shifting by $c$ is:
  \[ 
\left(\bigcup_{f=0}^{\ceil*{\frac{r}{2}}-1} 2^{min(f+2, r-f-1)} \bigcup_{i=0}^{max(0, 2^{r-2f-3}-1)} \left\{2^{2f+3+m}i +
q2^{2f+m} +t \right \} \right) \cup\; 2^{\floor*{\frac{r} {2}}} \left \{
t \right\}
\]
\[
= \left(\bigcup_{f=0}^{\ceil*{\frac{r}{2}}-1} 2^{min(f+2, r-f-1)+min(m, max(0, r-2f-3))} \bigcup_{i=0}^{max(0, 2^{r-2f-3-m}-1)}
\p*{2^{2f+3+m}i + q2^{2f+m} + t } \right) \]\[\cup\; 2^{\floor*{\frac{r}{2}}}\p*{t}
\]
\[
= \left(\bigcup_{f=0}^{\ceil*{\frac{r-m}{2}}-1} 2^{min(f+2, r-f-1)+min(m, max(0, r-2f-3))} \bigcup_{i=0}^{max(0, 2^{r-2f-3-m}-1)}
\p*{2^{2f+3+m}i + q2^{2f+m} + t } \right)\]\[ \cup\; 2^{\floor*{\frac{r+m}{2}}}\p*{t}.
\]
We could remove $q$ that was multiplied by $2^{2f+3+m}i$ owing to Lemma \ref{lem1}.
\end{enumerate}
\end{proof}

\begin{Col} \label{csqr}
Let us take a polynomial $P(x) = ax^2 + bx +c$, over $\ZZ_{2^r}$ with the domain restricted to $\bigcup\limits_{j=0}^{2^v-1} \p*{l+ 
2^{r-v}j}$ where $v \leq r$, $l\in \ZZ_{2^r}$.
Let $k = r-v$, $a' = a2^{2k}$, $b' = (2al + b)2^k$, $c' = al^2 + bl + c$. Let $a' = q'2^{w'}$ and $b' = g'2^{h'}$ such that $q'$ and
$g'$ are odd and $w', h'$ are orders of respectively $a'$ and $b'$.
Let $m' = min(w', h', r)$. The co-domain of $P(x)$ treated as a multiset equals:

\begin{enumerate}[label=\alph*)]
  \item If $w' > h'$ :
  \[ 2^{m'-k} \bigcup_{i=0}^{2^{r-m'}-1} \p*{ 2^{m'} i + c'} \]
  
  \item If $w' = h'$ :
  
\[ 
\begin{array}{ll}
2^{m'+1-k} \bigcup\limits_{i=0}^{2^{r-m'-1}-1} \p*{ 2^{m'+1} i + c'}   & \text{if }  m' < r\\
2^{r-k}\p*{c'} & \text{if } m' = r
\end{array}
\]

\item If $w' < h'$ :
\[
\left(\bigcup_{f=0}^{\ceil*{\frac{r-m'}{2}}-1} 2^{min(f+2, r-f-1)+min(m', max(0, r-2f-3))-k} \bigcup_{i=0}^{\max(0, 2^{r-2f-3-m'}-1)}
\p*{2^{2f+3+m'}i + q'2^{2f+m'} + t' } \right)\]\[ \cup\; 2^{\floor*{\frac{r+m'} {2}}-k} \left \{
t' \right\}
\]
where $ t'= c'-\frac{b'^2}{2^{m'+2}q'}$
\end{enumerate}
\end{Col}

\begin{proof}
Any $x$ in the domain can be written as $l+  2^k j$ for certain $j$, and therefore:
\[ 
ax^2 + bx +c = a(l+  2^{r-v}j)^2 + b(l+  2^{r-v}j) + c = a2^{2k}j^2 + (a2^{k+1}l + b2^k)j + al^2 + bl + c,
\]
from which the mapping to variables with primes automatically follows. Finally, we need to divide the number of occurrences of each element of the
multiset by $2^k$, as we have just $2^v = 2^{r-k}$-many $j$'s. 
\end{proof}

In the following lemma we take a slice for a single $f$ from category c), and count how many elements it has.

\begin{lemma} \label{1sum}
Let us define a multiset $S$ over $\ZZ_{2^r}$ by
\[ 
S = 2^{min(f+2, r-f-1)+min(m', max(0, r-2f-3))-k} \bigcup_{i=0}^{\max(0, 2^{r-2f-3-m'}-1)}
\p*{2^{2f+3+m'}i + q'2^{2f+m'} + t' },
\]
where all constants are integers between $0$ and $r$ inclusive, $f \leq \ceil*{\frac{r-m'}{2}}-1 $, and $ r > m' \geq k$. The number of
elements of the multiset $S$ equals
\[ 
2^{r-f-k-1}.
\]
\end{lemma}

\begin{proof}
The number of elements of $S$ is
\[ 
\#S = 2^{max(0, r-2f-3-m') + min(f+2, r-f-1)+min(m', max(0, r-2f-3))-k}.
\]
Let us go through following cases:
\begin{itemize}
  \item [$\bullet$] $0 > r-2f-3$:\\
  \[\#S = 2^{ 0 + r-f-1 + 0 -k} = 2^{r-f-k-1}  \]
  \item [$\bullet$] $r-2f-3 \geq 0 > r-2f-3-m'$:\\
  \[\#S = 2^{0 +  f+2+min(m', r-2f-3)-k}  = 2^{ f+2+r-2f-3-k} = 2^{r-f-k-1}     \]
  \item [$\bullet$] $r-2f-3-m' \geq 0$:\\
  \[\#S = 2^{r-2f-3-m' + f+2 +m'-k}  = 2^{r-f-k-1}.     \]
\end{itemize}
\end{proof}

In the next lemma we count the number of elements of slices for a single $f$, all $f'>f$, and also the slice $\p*{t'}$.

\begin{lemma} \label{multisum}
Let us define a multiset $S$ over $\ZZ_{2^r}$ by
\[
S = \left(\bigcup_{f=f_s}^{\ceil*{\frac{r-m'}{2}}-1} 2^{min(f+2, r-f-1)+min(m', max(0, r-2f-3))-k} \bigcup_{i=0}^{\max(0,
2^{r-2f-3-m'}-1)} \p*{2^{2f+3+m'}i + q'2^{2f+m'} + t' } \right)\]\[ \cup\; 2^{\floor*{\frac{r+m'} {2}}-k} \left \{
t' \right\}
\]
where all constants are integers between $0$ and $r$ inclusive, $f_s \leq \ceil*{\frac{r-m'}{2}} $ and $m' \geq k$. The number
of elements of the multiset $S$ equals
\[
2^{r-f_s-k}.
\]
\end{lemma} 

\begin{proof}
When $f_s \geq \ceil*{\frac{r-m'}{2}}-1$ this result, with the help of Lemma \ref{1sum}, is trivial to check. Let us assume now that $f_s
\leq \ceil*{\frac{r-m'}{2}}-2$. We consider two cases. The first case is when $m' = 0$ and $r$ is odd, or $m' \leq 1$ and $r$ is even.
Then the number of elements of $S$ is:
\[
\#S = 2^{\floor*{\frac{r} {2}}-k} + \sum_{f=f_s}^{\ceil*{\frac{r}{2}}-1} 2^{min(f+2, r-f-1) +min(m', max(0, r-2f-3))-k +
max(0,r-2f-3-m')}
\]
\[
=2^{\floor*{\frac{r} {2}}-k} + 2^{ r-\left(\ceil*{\frac{r}{2}}-1\right)-1-k }
+ \sum_{f=f_s}^{\ceil*{\frac{r}{2}}-2} 2^{f+2
+m'-k +r-2f-3-m'}
\]
\[
=2^{\floor*{\frac{r} {2}}-k} + 2^{\floor*{\frac{r} {2}}-k}
+ \sum_{f=f_s}^{\ceil*{\frac{r}{2}}-2} 2^{r-f-1-k}.
\]
We use now the formula for the sum of a geometric series:
\[
\#S= 2^{\floor*{\frac{r} {2}}-k+1} + 2^{r-f_s-k-1} \frac{1-2^{(-1)(\ceil*{\frac{r}{2}}-1-f_s)} } {1 - 2^{-1}}
= 2^{\floor*{\frac{r} {2}}+1} + 2^{r-f_s-k} \left(1-2^{f_s+1- \ceil*{\frac{r}{2}}}\right)
\]
\[
=2^{\floor*{\frac{r} {2}}-k+1} + 2^{r-f_s-k} - 2^{\floor*{\frac{r} {2}}-k+1}
= 2^{r-f_s-k}.
\]
Let us now consider the case where $m' > 0$ if $r$ odd, or $m' > 1$ when $r$ even.
The same sum becomes:
\[
\#S = 2^{\floor*{\frac{r+m'} {2}}-k} + \sum_{f=f_s}^{\ceil*{\frac{r-m'}{2}}-1} 2^{min(f+2, r-f-1) +min(m', max(0, r-2f-3))-k +
max(0,r-2f-3-m')}
\]
\[
= 2^{\floor*{\frac{r+m'} {2}}-k} + \sum_{f=f_s}^{\ceil*{\frac{r-m'}{2}}-1} 2^{ f+2 +min(m', r-2f-3)-k + max(0,r-2f-3-m')}
\]
\[
= 2^{\floor*{\frac{r+m'} {2}}-k} + 2^{ \ceil*{\frac{r-m'}{2}}+1 + r - 2 \ceil*{\frac{r-m'}{2}} - 1 -k}
+ \sum_{f=f_s}^{\ceil*{\frac{r-m'}{2}}-2} 2^{ f+2 +m'-k + r-2f-3-m'}
\]
\[
= 2^{\floor*{\frac{r+m'} {2}}-k+1}  + \sum_{f=f_s}^{\ceil*{\frac{r-m'}{2}}-2} 2^{r-f-k-1}.
\]
We use once again the formula for the sum of a geometric series, which produces
\[
\#S = 2^{\floor*{\frac{r+m'} {2}}-k+1}  + 2^{r-f_s-k} - 2^{\floor*{\frac{r+m'} {2}}-k+1}
=2^{r-f_s-k}.
\]
\end{proof}

In the proof of the next lemma we will use variables analogous to those we presented in Corollary \ref{csqr}. Before proceeding further, let
us introduce one new notation that we will employ frequently:
\[
m^* = m' -k,  
\]
where we should note that $m^* \geq 0$ since $m' \geq k$.

Before starting to prove Lemma \ref{lin}, let us note first that it is not superseded by the earlier mentioned result of
Marshall and Ramage \cite{mar75}. Furthermore, we do not see a way to employ their proof technique to obtain this
lemma---even when $v=r$ and $d=0$. It is because, at the very beginning, they constraint solutions to equivalence classes that are over ring
$\ZZ_{2^{r-1}}$.
For $n=2$ this right away limits the divisibility they may obtain to $2^{r-1}$, which is less than Lemma \ref{lin} produces.
\begin{lemma} \label{lin}
Let $P(x) = a_x x^2 + b_x x +c$ and $Q(y, h) = a_y y^2 + b_y y +  2^{r-d}h$. Then for any $ d \leq v \leq r$  when
we work over $\ZZ_{2^r}$, it holds that:
\[
2^{v+d} \;|\; \#\left\{(x, y, h) : x \in \bigcup\limits_{j=0}^{2^v-1} \{l_x+
2^{r-v}j \}, y \in \bigcup\limits_{j=0}^{2^v-1} \p*{l_y+ 2^{r-v}j} , h \in \bigcup\limits_{j=0}^{2^d-1} \p*{j}, P(x)  = Q(y, h) \right\}
\]
for any $l_x, l_y$.
\end{lemma}
\begin{proof}

Depending on its constants, $P$ may fall into category a), b) or c) as per Corollary \ref{csqr}. We will also consider $Q$ to be in one
of those categories depending to which of them the $a_y y^2 + b_y y$ part of $Q$ belongs. We will go through all possible pairings of those
categories for $P$ and $Q$ and proof the result for each of them. Let us notice though that if $P$ is in category b), we could treat it
just as being in category a) but with taking its $m'$ (or $m'_x$ as we will call it) to be bigger by $1$, but not bigger than $r$.
The same goes for $Q$. Therefore category b) can be easily ``reduced'' to category a) and in the rest of the proof it is sufficient if
we only consider $P$ and $Q$ to be in categories either a) or c).

\begin{enumerate}
  \item \label{itm:aa} Both $P$ and $Q$ are in category a).\\
   The image of $P$ is
  \[ 2^{m'_x-k} \bigcup_{i=0}^{2^{r-m'_x}-1} \p*{ 2^{m'_x} i + c' }  \] 
  whereas the image of $Q$ is
  \[ 2^{m'_y-k} \bigcup_{i=0}^{2^{r-m'_y}-1} \bigcup_{h=0}^{2^{d}-1} \p*{ 2^{m'_y} i + 2^{r-d}h }
   =  2^{m'_y+min(r-m'_y, d)-k} \bigcup_{i=0}^{2^{r - min(m'_y, r-d)}-1}  \p*{ 2^{min(m'_y, r-d)}i },   \]
 where variables are defined in analogy to Corollary \ref{csqr}. If those two images have no common element we automatically obtain
  the required divisibility. Otherwise, their distinct common elements are all elements of the more sparse image, that is:
  
  \[  \bigcup\limits_{i=0}^{2^{r-max(m'_x, min(m'_y, r-d))}-1} \p*{ 2^{max(m'_x, min(m'_y, r-d))} i + g'},\]
  where $g'= 0$ or $c'$ depending on value of $max(m'_x, min(m'_y, r-d))$. There is $2^{r-max(m'_x, min(m'_y, r-d))}$ of those elements, and
  each of them has  
  \[2^{m'_x-k}2^{m'_y+min(r-m'_y, d)-k}   \]
  occurrences. This gives the size of the whole overlap to be
  \[ 2^{r-max(m'_x, min(m'_y, r-d))} 2^{m'_x-k}2^{m'_y+min(r-m'_y, d)-k}  =
  2^{min(r-m'_x, r- min(m'_y, r-d)) + min(r-m'_y, d) + m^*_x + m^*_y}   \]
  \[
  =2^{min(r-m'_x, max(r-m'_y, d) ) + min(r-m'_y, d) + m^*_x + m^*_y} 
  = 2^{min(v-m^*_x+min(v-m^*_y, d)  , v-m^*_y+ d) + m^*_x + m^*_y} 
  \]
  \[
  = 2^{min(v+min(v, d+m^*_y), v+ d + m^*_x)} 
  = 2^{min(2v, v+ d + m^*_x)},
  \]
  which produces the desired divisibility.

  \item \label{itm:ca} $P$ is in category c) and $Q$ is in category a).\\
  The image of $P$ is
    \[
  \left(\bigcup_{f=0}^{\ceil*{\frac{r-m'_x}{2}}-1} 2^{min(f+2, r-f-1)+min(m'_x, max(0, r-2f-3))-k} \bigcup_{i=0}^{max(0,
  2^{r-2f-3-m'_x}-1)} \p*{2^{2f+3+m'_x}i + q'2^{2f+m'_x} } \right)\]\[ \cup\; 2^{\floor*{\frac{r+m'_x}{2}}-k}\p*{t'},
  \]
  and the image of $Q$ is, once again,
  \[ 2^{min(r, m'_y+d)-k} \bigcup_{i=0}^{2^{r - min(m'_y, r-d)}-1}  \p*{ 2^{min(m'_y, r-d)}i }.
  \]
  The image of $P$ consists of $\ceil*{\frac{r-m'_x}{2}}$ linear slices (one per value of $f$) and then the slice for
  $\p*{t'}$, which we can take to have a period of $2^r$. If the images of $P$ and $Q$ have no common element, we automatically get the
  result. When the contrary is true, let us first assume that there is a common element between the two images at a slice
    \[\p*{2^{2f_0+3+m'_x}i + q'2^{2f_0+m'_x} + t' },\]
    for certain $f_0$. Let us consider the following cases:
    \begin{enumerate}[label*=\arabic*., ref=\arabic*]
      \item $min(m'_y, r-d) \geq 2f_0+1+m'_x$ \\
      
      Under this condition only elements in that particular slice for $f_0$ may
      be common for the two images. Let $min(m'_y, r-d) = 2f_0+1+m'_x+h$ for certain $h\geq0$, which also automatically means that 
       $min(r, m'_y +d)-k = 2f_0+1+m'_x +h +d -k$. When $h \leq 2$ the whole slice is common.
      Then, starting at $h=2$, whenever $h$ increases by $1$ the part of the slice for $f_0$ that is common is halved. By Lemma
      \ref{1sum}, the number of elements in the whole $f_0$ slice equals $2^{r-f_0-k-1}$, and the size of whole overlap is the number of common
      elements in the slice for $f_0$ multiplied by the number of occurrences of each distinct element in the image of $Q$. Therefore the
      size of the intersection is
      \[
      2^{r-f_0-k-1}2^{min(-(h-2), 0)}2^{min(m'_y, m'_y+d)-k} = 
      2^{r-f_0-k-1 + min(2-h, 0)+ 2f_0+1+m'_x +h +d -k}\]
      \[ = 2^{r-k+ f_0+m^*_x+ d + min(2, h)} =2^{v+ d + f_0+m^*_x+ min(2, h)},
      \]
      which has the desired divisibility.

      \item $min(m'_y, r-d) \leq 2f_0+m'_x$ \\
      
        If $min(m'_y, r-d) \geq m'_x$ then we can write that $min(m'_y, r-d) = 2f_s+m'_x+h$ for certain $f_s
  		\leq f_0$ and $0 \leq h \leq 1$. Otherwise, we have that $min(m'_y, r-d) = k + h$ for certain $h < m'_x - k$, and we set $f_s=0$.
  		Now the common elements for both images that are in the image of $P$ are:
  		  {\small
		  \[
		  \left(\bigcup_{f=f_s}^{\ceil*{\frac{r-m'_x}{2}}-1} 2^{min(f+2, r-f-1)+min(m'_x, max(0, r-2f-3))-k} \bigcup_{i=0}^{max(0, 2^{r-2f-3-m'_x}-1)}
		  \p*{2^{2f+3+m'_x}i + q'2^{2f+m'_x} + t' } \right)\]\[ \cup\; 2^{\floor*{\frac{r+m'_x} {2}}-k} \left \{
		  t' \right\},
		  \]
		  }that is, all slices for $f = f_s$ and higher orders. Each element in this part of the $P$ image will be multiplied $2^{min(r, m'_y
		  +d)-k}$ times, as there are that many elements in $Q$ image equal to it. We know from Lemma \ref{multisum} that there is $2^{r-f_s-k}$ elements in the image of $P$
		  that are common. Let us do the multiplication considering the cases on $min(m'_y, r-d)$ we described earlier. First when
		  $min(m'_y, r-d) = 2f_s+m'_x+h$, $0 \leq h \leq 1$
		  \[
		  2^{r-f_s-k}2^{min(r, m'_y+d)-k} = 2^{r-f_s-k +2f_s+m'_x+h+d-k} = 2^{v+d +f_s + h + m^*_x },
		  \]
		  and now when $min(m'_y, r-d) = k + h$, $0 \leq h < m'_y -k$, $f_s = 0$
		  \[
		  2^{r-f_s-k}2^{min(r, m'_y+d)-k} = 2^{r-k +k + h +d - k } = 2^{v+d+h}.
		  \]
		  This gives us the desired divisibility in both cases.\\
       \end{enumerate}
  
  	Finally we need to consider the scenario when the common element between both images is $\p*{t'}$, and there are no other
  	different common elements. This means that 
  	\[min(m'_y, r-d) > 2\left(\ceil*{\frac{r-m'_x}{2}}-1\right) + m'_x\]
  	therefore 
  	\[ min(m'_y+d, r) =  r-1+d+h\]
  	for certain $h=0$ or $1$ depending on parity of $r-m'_x$ . In this case the number of common elements is:
  	\[
  	2^{\floor*{\frac{r+m'_x}{2}}-k} 2^{min(r, m'_y+d)-k} 
  	=  2^{\floor*{\frac{r+m'_x}{2}}-k +  r-1+d+h-k}
  	=  2^{ v + d+ h-1 + \floor*{\frac{v+m^*_x}{2}}}.
  	\]
  	If $h=1$ (\ie $r-m'_x$ is odd) this has the required divisibility. Otherwise, we know that $v+m^*_x$ has to be even, since  
  	$r-m'_x$ is. If $v \geq 1$ this gives the desired divisibility, and when $v=0$ the whole lemma becomes trivial.
  	
\item \label{itm:ac} $P$ is in category a) and $Q$ is in category c).\\ 
The image of $P$ is
\[ 
2^{m'_x-k} \bigcup_{i=0}^{2^{r-m'_x}-1} \p*{ 2^{m'_x} i + c'}.
\] 
When it comes to the image of $Q$, for $d=0$ it would be just an image of $a_y y^2 + b_y y$, which consists of slices as we know
them from Corollary \ref{csqr}:
\[
\left(\bigcup_{f=0}^{\ceil*{\frac{r-m'_y}{2}}-1} 2^{min(f+2, r-f-1)+min(m'_y, max(0, r-2f-3))-k} \bigcup_{i=0}^{max(0,
2^{r-2f-3-m'_y}-1)} \p*{2^{2f+3+m'_y}i + q'2^{2f+m'_y} + t' } \right)
\]
\[
\cup\; 2^{\floor*{\frac{r+m'_y}{2}}-k}\p*{t'}.
\]
When $d > 0$, each of those slices is ``affected'' by the slice $2^{r-d}h$, which has granularity $2^{r-d}$. If an affected slice already has
at least that granularity, then each of its elements has just $2^d$ more occurrences. Otherwise, we obtain a slice with $2^{r-d}$
period, where each element has number of occurrences equal to number of all elements in the affected slice. In effect the image of $Q$
is:
{\small
\[
\left(\bigcup_{f=0}^{\ceil*{\frac{r-d-3-m'_y}{2}}-1} 2^{min(f+2, r-f-1)+min(m'_y, max(0, r-2f-3))-k+d} \bigcup_{i=0}^{max(0,
2^{r-2f-3-m'_y}-1)} \p*{2^{2f+3+m'_y}i + q'2^{2f+m'_y} + t' } \right)
\]
\[
\cup\; \left(\bigcup_{f=\ceil*{\frac{r-d-3-m'_y}{2}}}^{\ceil*{\frac{r-m'_y}{2}}-1} 2^{r-f-k-1}
\bigcup_{h=0}^{2^d-1} \p*{2^{r-d}h + q'2^{2f+m'_y} + t' } \right)
\cup\; \bigcup_{h=0}^{2^d-1} 2^{\floor*{\frac{r+m'_y}{2}}-k}\p*{2^{r-d}h + t'}
\]
\[
= \left(\bigcup_{f=0}^{\ceil*{\frac{r-d-3-m'_y}{2}}-1} 2^{min(f+2, r-f-1)+min(m'_y, max(0, r-2f-3))-k+d} \bigcup_{i=0}^{max(0,
2^{r-2f-3-m'_y}-1)} \p*{2^{2f+3+m'_y}i + q'2^{2f+m'_y} + t' } \right)
\]
\[
\cup\; \left(\bigcup_{f=\ceil*{\frac{r-d-3-m'_y}{2}}}^{\ceil*{\frac{r-d-m'_y}{2}}-1} 2^{r-f-k-1}
\bigcup_{h=0}^{2^d-1} \p*{2^{r-d}h + q'2^{2f+m'_y} + t' } \right)
\cup\; \bigcup_{h=0}^{2^d-1} 2^{\floor*{\frac{r+min(d+m'_y, r)}{2}}-k}\p*{2^{r-d}h + t'}.
\]
}Basing on Lemmas \ref{1sum} and \ref{multisum}, we can say that a slice for any given $f$ in $Q$'s image has $2^{r+d-f-k-1}$ elements, and
the number of elements in union of slices for a set $f=f_s$, all higher $f$'s, and the $\p*{2^{r-d}h + t'}$ slice, equals $2^{r+d-f-k}$.

  If the images of $P$ and $Q$  have no common elements, we automatically get the result.
  Now let us assume that they have a common element at some slice
  \[\p*{2^{2f_0+3+m'_y}i + q'2^{2f_0+m'_y} + t' }\]
  $\bb*{\text{\ie} f_0 \leq \ceil*{\frac{r-d-3-m'_y}{2}}-1}$.
  If $m'_x \geq 2f_0+1+m'_y$, then only elements in that particular slice may be common for the images. Let $m'_x = 2f_0+1+m'_y +h$ for
  certain non-negative $h$. When $h \leq 2$ then the whole slice is common. Then, starting at $h=2$, whenever $h$ increases by $1$, the
  part of the slice that is common is halved. Due to Lemma \ref{1sum}, the number of elements in whole slice equals $2^{r-f_0-k-1}$, and
  the number of common elements is:
  \[
  2^{m'_x-k}2^{r+d-f_0-k-1}2^{min(-(h-2), 0)} = 2^{2f_0+1+m'_y +h-k + r+d-f_0-k-1 + min(2-h, 0)}\]\[ = 2^{f_0+m^*_y+ r+d-k + min(2, h)} =
  2^{f_0+m^*_y+ min(2, h) + v +d},
  \]
  which gives the desired divisibility.\\
  
  Now let us assume that the images have a common element at some slice
  \[\p*{2^{r-d}h_0 + q'2^{2f_0+m'_y} + t' }\]
  $\bb*{\text{\ie} \ceil*{\frac{r-d-m'_y}{2}}-1 \geq f_0 \geq \ceil*{\frac{r-d-3-m'_y}{2}}}$.
  If $m'_x \geq 2f_0+1+m'_y$, then only elements in that particular slice may be common for the images. Let $m'_x = 2f_0+1+m'_y +h$ for
  certain non-negative $h$. When $m'_x \leq r-d \Leftarrow 2f_0+1+m'_y +h \leq r-d \Leftrightarrow h \leq r-d-1-2f_0-m'_y$ then the whole
  slice is common. Then, starting at $h=r-d-1-2f_0-m'_y$, whenever $h$ increases by $1$ then the part of the slice that is common is halved.
  We know that the number of elements in whole slice equals $2^{r+d-f_0-k-1}$, and the number of common elements is:
  \[
  2^{m'_x-k}2^{r+d-f_0-k-1}2^{min(-(h-(r-d-1-2f_0-m'_y)),0)} = 2^{2f_0+1+m'_y +h -k+ r+d-f_0-k-1 + min(r-d-1-2f_0-m'_y -h, 0)}\]
  \[ = 2^{f_0 + m^*_y + r + d -k + min(r-d-1-2f_0-m'_y, h) } = 2^{f_0+m^*_y + min(r-d-1-2f_0-m'_y, h)  + v +d},
  \]
  which again gives the desired divisibility.

  Let us consider now a case when $m'_x \leq 2f_0+m'_y$, \ie where more than one slice form $Q$ is common with $P$. If $m'_x \geq m'_y$ then we
  can write that $m'_x = 2f_s+m'_y+h$ for certain $f_s \leq f_0$ and $h \leq 1$. Otherwise we have that $m'_x = k + h$ for certain $h <
  m'_y -k$, and we set $f_s=0$.
  
  Now the common elements for both images that are in the image of $Q$ are:
  {\small
  \[
  \left(\bigcup_{f=f_s}^{\ceil*{\frac{r-d-3-m'_y}{2}}-1} 2^{min(f+2, r-f-1)+min(m'_y, max(0, r-2f-3))-k+d} \bigcup_{i=0}^{max(0,
  2^{r-2f-3-m'_y}-1)} \p*{2^{2f+3+m'_y}i + q'2^{2f+m'_y} + t' } \right)
  \]
  \[
  \cup\; \left(\bigcup_{f=max(\ceil*{\frac{r-d-3-m'_y}{2}}, f_s)}^{\ceil*{\frac{r-d-m'_y}{2}}-1} 2^{r-f-k-1}
  \bigcup_{h=0}^{2^d-1} \p*{2^{r-d}h + q'2^{2f+m'_y} + t' } \right)
  \cup\; \bigcup_{h=0}^{2^d-1} 2^{\floor*{\frac{r+min(d+m'_y, r)}{2}}-k}\p*{2^{r-d}h + t'},
  \]
  }that is, all slices for $f = f_s$ and higher orders. Each element in this part of the $Q$ image will be multiplied $2^{m'_x-k}$ times,
  as there are that many elements in $P$ image equal to it. We know basing on Lemma \ref{multisum} that there are $2^{r+d-f_s-k}$
  elements in $Q$ image that are common. Let us do the multiplication considering the cases on $m'_x$ we described earlier. First when $m'_x =
  2f_s+m'_y+h$, $h \leq 1$
  \[
  2^{r+d-f_s-k}2^{m'_x-k} = 2^{r+d+f_s-k + h + m^*_y} = 2^{f_s + h + m^*_y + d + v}
  \]
  and now when $m'_x = k + h$, $0 \leq h < m'_y -k$, $f_s = 0$
  \[
  2^{r+d-f_s-k}2^{m'_x-k} = 2^{r-k +k + h -k +d } = 2^{h+d+v}.
  \]
  This gives the desired divisibility in both cases.\\

  Finally we need to consider the case in which the only common slice between the images is $\p*{2^{r-d}h+t'}$.
  This means that 
  \[ m'_x > 2 \left( \ceil*{\frac{r-d-m'_y}{2}}-1 \right) + m'_y \] 
  therefore
  \[ m'_x = r-d-1+h \]
  for certain $h\geq0$ or $h\geq1$ depending on the parity of $r-d-m'_y$. In that case the number of common elements is:
  \[
   C = 2^{min(d, r-m'_x)} 2^{m'_x-k}2^{\floor*{\frac{r+min(d+m'_y, r)}{2}}-k} 
  \]
  Let us consider two sub-cases, first when $v>d$ giving:
  \[
  C = 2^{r-d-1+h-k+\floor*{\frac{r+min(d+m'_y, r)}{2}}-k}
  = 2^{min(d, d+1-h)+ v-d-1+h+\floor*{\frac{v+min(d+m^*_y, v)}{2}}}
  \]
  \[
  = 2^{ v+d + min(h, 1)-1 +\floor*{\frac{v-d+min(m^*_y, v-d)}{2}}}
  \]
  If $h \geq 1$ this gives the required divisibility. When $h=0$ it means that $r-d-m'_y$ is even and therefore also that $v-d+min(m^*_y,
  v-d)$ is even. Because $v>d$ it means that $v-d+min(m^*_y, v-d) \geq 2$, which also results in required divisibility.
  Let us consider now the case for $v=d$, starting over from the initial formula for the size of the intersection:
  \[ 
    C = 2^{min(v, r-m'_x)} 2^{m'_x-k}2^{\floor*{\frac{r+min(v+m'_y, r)}{2}}-k}
    = 2^{v + \floor*{\frac{v+min(v+m^*_y, v)}{2}}}
    = 2^{v + \floor*{\frac{2v}{2}}} = 2^{2v},
  \]
  which also has the desired divisibility.
  
  \item Both $P$ and $Q$ are in category c).\\ 
  The image of $P$ is
  \[
  \left(\bigcup_{f=0}^{\ceil*{\frac{r-m'_x}{2}}-1} 2^{min(f+2, r-f-1)+min(m'_x, max(0, r-2f-3))-k} \bigcup_{i=0}^{max(0,
  2^{r-2f-3-m'_x}-1)} \p*{2^{2f+3+m'_x}i + 2^{2f+m'_x} } \right)\]\[ \cup\; 2^{\floor*{\frac{r+m'_x}{2}}-k}\p*{0}
  \]
  and the image of $Q$, once again, is:
  {\small
  \[
  \left(\bigcup_{f=0}^{\ceil*{\frac{r-d-3-m'_y}{2}}-1} 2^{min(f+2, r-f-1)+min(m'_y, max(0, r-2f-3))-k+d} \bigcup_{i=0}^{max(0,
  2^{r-2f-3-m'_y}-1)} \p*{2^{2f+3+m'_y}i + q'_y2^{2f+m'_y} + t'_y } \right)
  \]
  \[
  \cup\; \left(\bigcup_{f=\ceil*{\frac{r-d-3-m'_y}{2}}}^{\ceil*{\frac{r-d-m'_y}{2}}-1} 2^{r-f-k-1}
  \bigcup_{h=0}^{2^d-1} \p*{2^{r-d}h + q'_y2^{2f+m'_y} + t'_y } \right)
  \cup\; \bigcup_{h=0}^{2^d-1} 2^{\floor*{\frac{r+min(d+m'_y, r)}{2}}-k}\p*{2^{r-d}h + t'_y}
  \]
  }We could take $q'_x=1$ and $t'_x=0$ (which allowed us to omit them in the formula), since when faced with $P(x) = Q(y, h)$ we can first
  divide both sides by $q'_x$ and then subtract $\frac{t'_x}{q'_x}$, which appropriately also adjusts $q'_y$ and $t'_y$ (both $q'_x$ and
  $q'_y$ are guaranteed to be odd).
  If there are no common elements between the two images then we are automatically done. Let us now assume that there is an overlap, \ie that
  there is at least one common element between the images.
  Our approach is to go through all possible classes of overlaps between pairs of slices, and then for any such overlap we will show either that it has the required divisibility, or that there have to be
  some more overlaps that together with this one have that divisibility. We will also never count the same overlaps more than once.

\begin{enumerate}[label*=\alph*., ref=\arabic{enumi}.\alph{enumii}]
  \item \label{itm:3a} First let us assume there is a common element between certain slice for $f_x$  and  a slice for $f_y \leq
  \ceil*{\frac{r-d-3-m'_y}{2}}-1$.
  Let us go through cases:
  \begin{enumerate}[label*=\arabic*., ref=\arabic{enumi}.\alph{enumii}.\arabic*]
    \item $2f_x + m'_x \geq 2f_y + m'_y$:\\
    Then the number of common elements between those two slices is the number of elements in the slice $f_x$ multiplied by the number of
    occurrences of each distinct element in the slice $f_y$, that is:
    \[ 2^{r - f_x - k -1 }2^{min(f_y+2, r-f_y-1)+min(m'_y, max(0, r-2f_y-3))-k+d} \]\[= 2^{r - f_x - k -1 + f_y+2 +m'_y  -k+d}  = 2^{v + d
    + m^*_y -f_x+f_y+1} \]
    since $f_y \leq \ceil*{\frac{r-d-3-m'_y}{2}}-1$.
    If $f_x \leq m^*_y + f_y +1$ then the common part of those two slices already has the desired divisibility. Therefore let us assume
    now, that there is a pair of slices that overlap, for which $f_x \geq  m^*_y +f_y + 2$. The slice for $f_y$ has a ``period'' of
    $2^{2f_y+3 +m'_y}$. If that slice overlaps with the slice
    
    \[ \p*{2^{2f_x + m'_x}} = \p*{2^{2(m^*_y +f_y + 2+h)+ m'_x}} = \p*{2^{2f_y+3 +m'_y + m^*_x + m^*_y+1 + h}} \]
    
    (for certain even $h \geq 0$), then it also overlaps with all other slices attainable from it by shifting
    by a multiple of the $2^{2f_y+3 +m'_y}$ period. Therefore the slice for $f_y$ overlaps with slices
    
    \[\p*{0} \text{\;\;\;and\;\;\;}\p*{2^{2f_y+3 +m'_y+ m^*_x + m^*_y+1 +h}} \]
    
    for all even $h \geq  -m^*_x - m^*_y -1$. This means that the slice for $f_y$ overlaps with slices $\p*{0}$ and
    $\p*{2^{2f_x+m'_x}}$ for any $f_x$ such that
    
    \[2f_x + m'_x \geq 2f_y+3 +m'_y \Leftrightarrow f_x \geq \ceil*{\frac{2f_y+3 +m'_y -
    m'_x}{2}}.\]
    
    Let us count the number of common elements between the slice for $f_y$ and the union of the just-mentioned slices. Using
    Lemma \ref{multisum} we know the size of that union is $2^{r-\ceil*{\frac{2f_y+3 +m'_y - m'_x}{2}}-k}$, whereas the number of
    occurrences of each distinct element in the slice for $f_y$ is $f_y+2 +m'_y  -k+d$, which gives:
    \[ 2^{r-\ceil*{\frac{2f_y+3 +m'_y - m'_x}{2}}-k + f_y+2 +m'_y  -k+d} = 2^{v-\ceil*{\frac{2f_y+3 +m'_y - m'_x}{2}} + f_y+2 +m^*_y+d}
    = 2^{v+d-\ceil*{\frac{ -m^*_y - m^*_x-1}{2}}},
    \]
    that produces the desired divisibility.
    
    \item \label{itm:3a2}$2f_x + m'_x < 2f_y + m'_y$:\\
    Now the number of common elements between the slices is the number of elements in the slice $f_y$ multiplied by the number of
    occurrences of each distinct element in the slice $f_x$, that is:
    \[ D = 2^{r +d - f_y - k -1 }2^{min(f_x+2, r-f_x-1)+min(m'_x, max(0, r-2f_x-3))-k}. \]
    Now we need to consider the following two subcases:
    \begin{enumerate}[label*=\arabic*., ref=\arabic{enumi}.\alph{enumii}.\arabic{enumiii}.\arabic*]
      \item \label{itm:3a21} $r-2f_x-3-m'_x \geq 0$:\\
      \[D = 2^{r+d - f_y - k -1 + f_x+2 +m'_x  -k+d}  = 2^{v + d + m^*_x -f_y+f_x+1}. \]
      This case is analogous to the case we already considered---just $d$ shows up in a different place.
      If $f_y \leq m^*_x + f_x +1$, then the common part of those two slices already has the desired divisibility. Therefore let us assume
      now that there is a pair of slices that overlap, for which $f_y \geq  m^*_x +f_x + 2$. The slice for $f_x$ has a ``period'' of
      $2^{2f_x+3 +m'_x}$. If that slice overlaps with the slice
      
      \[ \p*{q'_y2^{2f_y + m'_y}  +t'_y} = \p*{q'_y2^{2(m^*_x +f_x + 2+h)+ m'_y}+t'_y} = \{q'_y2^{2f_x+3 +m'_x + m^*_y + m^*_x+1 +
      h}+t'_y\}\]
      
      (for certain even $h \geq 0$), then it also overlaps with all other slices attainable from it by shifting
      by a multiple of the $2^{2f_x+3 +m'_x}$ period. Therefore the slice for $f_x$ overlaps with slices
      
      \[\p*{t'_y} \text{\;\;\;and\;\;\;}\p*{q'_y2^{2f_x+3 +m'_x+ m^*_y + m^*_x+1 +h}+t'_y} \]
      
      for all even $h \geq  -m^*_x - m^*_y -1$. This means that the slice for $f_x$ overlaps with slices $\p*{t'_y}$ and
      $\p*{q'_y2^{2f_y+m'_y}+t'_y}$, for any $f_y$ such that
      
      \[2f_y + m'_y \geq 2f_x+3 +m'_x \Leftrightarrow f_y \geq \ceil*{\frac{2f_x+3 +m'_x -
      m'_y}{2}}.\]
      
      Let us count the number of common elements between the slice for $f_x$ and the union of the just-mentioned slices. Basing on
      Lemma \ref{multisum} we know the size of that union is $2^{r-\ceil*{\frac{2f_x+3 +m'_x - m'_y}{2}}-k+d}$, whereas the number of
      occurrences of each distinct element in the slice for $f_x$ is $f_x+2 +m'_x  -k$ which gives:
      \[ 2^{r-\ceil*{\frac{2f_x+3 +m'_x - m'_y}{2}}-k +d + f_x+2 +m'_x  -k} = 2^{v-\ceil*{\frac{2f_x+3 +m'_x - m'_y}{2}} + f_x+2 +m^*_x+d}
      = 2^{v+d-\ceil*{\frac{ -m^*_x - m^*_y-1}{2}}},
      \]
      which also produces the desired divisibility.
      
      \item $ 0 > r-2f_x-3-m'_x$ \\
      \[D = 2^{r - f_y - k -1 + r - f_s - k -1 + d} = 2^{2v -f_y - f_x -2 + d}. \]
      In this case $f_x = \ceil*{\frac{r-m'_x}{2}}-1 =
      \ceil*{\frac{v-m^*_x}{2}}-1$,  $f_y \leq \ceil*{\frac{r-3-m'_y-d}{2}}-1 = \ceil*{\frac{v-m^*_y-3 -d}{2}}-1$. Let us
      substitute them into the equation (we take worst case for $f_y$), obtaining:
      \[ D = 2^{2v -\ceil*{\frac{v-m^*_x}{2}} -\ceil*{\frac{v-m^*_y-3-d}{2}} + d } = 2^{v + d + h}  \]
      for certain non-negative $h$.
    \end{enumerate}
  \end{enumerate}
  
  \item \label{itm:3b} Now let us assume there is a common element between certain slices for $f_x$  and  $
  \ceil*{\frac{r-d-m'_y}{2}}-1   \geq   f_y \geq  \ceil*{\frac{r-d-3-m'_y}{2}}$. Let us once again
  consider two cases:
  \begin{enumerate}[label*=\arabic*.]
    \item $2f_x + 3 + m'_x \geq r-d$:\\
    In this case the number of common elements between the two slices is the number of elements in the slice for $f_x$ multiplied by the
    number of occurrences of each distinct element in the slice $f_y$, that is:
    \[ 2^{r - f_x - k -1 } 2^{r - f_y - k -1} = 2^{2v - f_x -f_y -2}. \]
    There are one or two possible values of $f_y$, depending whether $r-d-m'_y$ is even or odd.
    \begin{enumerate}[label*=\arabic*., ref=\arabic{enumi}.\alph{enumii}.\arabic{enumiii}.\arabic*]
      \item\label{itm:3b11} First let us take $f_y = \ceil*{\frac{r-d-m'_y}{2}} - 1$, in which case above expression equals
      \[ 2^{2v - f_x - \ceil*{\frac{r-d-m'_y}{2}} - 1} .\]
      If $f_x \leq v -d - \ceil*{\frac{r-d-m'_y}{2}} - 1$ we automatically get the required divisibility.
      Therefore let us assume,
      now, that there is a pair of slices that overlap, for which $f_x \geq  v -d - \ceil*{\frac{r-d-m'_y}{2}}$. The slice for $f_y$ has a
      ``period'' of $2^{r-d}$. If that slice overlaps with slice
      
      \[ \p*{2^{2f_x + m'_x}} = \p*{2^{2\left(v -d - \ceil*{\frac{r-d-m'_y}{2}}\right)+ m'_x + g}} = \{2^{r-d + m^*_y + m^*_x +
      s}\}\]
      
      (for certain  $g\geq 0$, $s \geq -1$), then it also overlaps with all other slices attainable from it by shifting
      by a multiple of the $2^{r-d}$ period.
      
      If $m^*_y=m^*_x=0$ and $s=-1$ (\ie $r-d-m'_y$ is odd) it means that the only overlapping slice
      is $\p*{2^{r-d+2}i + 2^{r-d-1}}$, as
      $\p*{2^{r-d}h+  q'_y2^{r-d-1} + t'_y} $ doesn't intersect with any other slice for any $h$. In this case $2f_x+m'_x = r-d-1
      \Leftrightarrow f_x = \ceil*{\frac{r-d-m'_x-1}{2}}$, but the ceiling means nothing as $r-d-m'_x-1$ is even, since $r-d-m'_y$ is odd and
      $m'_y = m'_x$ due to $m^*_y=m^*_x=0$. The number of common elements between those two slices is therefore:
      
      \[
      2^{r - \ceil*{\frac{r-d-m'_x-1}{2}}-k -1 + r - \ceil*{\frac{r-d-m'_y}{2}} - k }
      = 2^{2v - \ceil*{\frac{v-d-1}{2}} - \ceil*{\frac{v-d}{2}} - 1 } = 2^{v+d-1}
      \]
      
      Yet in this case we claim that also the slice $\p*{2^{r-d}h + t'_y}$ overlaps with all the slices for $f_x+1$ and higher.
      First, the
      slice for $f_y$ is $\p*{2^{r-d}h+  q'_y2^{r-d-1} + t'_y}$ in this case, and the slice for $f_x$ is $\{2^{r-d+2}i  +
      2^{r-d-1}\}$, and they overlap - which means that  $2^{r-d} \mid t'_y$ ($q'_y$ is guaranteed to be odd).
      Any slice for any higher $f'_x = f_x+c+1$, $c\geq0$ is $\p*{ 2^{r-d+4+2c}i + 2^{r-d+1+2c}}$,
      and therefore it has to overlap with the slice $\p*{2^{r-d}h + t'_y}$. Let us count the size of this overlap:
      
      \[2^{r-\left(\ceil*{\frac{r-d-m'_x-1}{2}} + 1\right)-k}2^{\floor*{\frac{r+min(d+m'_y, r)}{2}}-k}  = 2^{v+d-1}  \]
      Together with the $2^{v+d-1}$ we obtained earlier, this gives the required divisibility.
      
      Now let us suppose that $m^*_y \geq 1$ or $m^*_x \geq 1$ or $s\geq0$. Then the slice for $f_y$ overlaps with some slice $\{2^{r-d
      + z_0}\}$ for certain $z_0 \geq 0$, which means that it also overlaps with slice $\p*{0}$, and we are in the case \ref{3e}.
      
      \item Let us consider now the second possible value of $f_y$, that is $f_y = \ceil*{\frac{r-d-3-m'_y}{2}}$, and in this case
      $r-d-m'_y$ is odd (otherwise this value of $f_y$ would equal the value we already considered). Now the size of the common part of the two slices that
      overlap is:
      \[ 2^{2v - f_x -f_y -2} = 2^{2v - f_x - \ceil*{\frac{r-d-3-m'_y}{2}} -2}.\]
      If $f_x \leq v-d - \ceil*{\frac{r-d-3-m'_y}{2}} -2$ this gives the required divisibility. Let us assume now it does not, and we have
      $f_x \geq v-d - \ceil*{\frac{r-d-3-m'_y}{2}} -1$. In this case the slice for $f_y$ overlaps with slice:
      \[
      \p*{2^{2f_x + m'_x} } = \p*{2^{2\left(v -d - \ceil*{\frac{r-d-3-m'_y}{2}} -1\right)+ m'_x + g }} = \p*{2^{2v-2d -r+d+2k+1 + m^*_y +
      m^*_x + g}}
      \]
      \[
      = \p*{2^{r-d+1+ m^*_y + m^*_x + g}}
      \]
      for certain non-negative even $g$. Now, also the slice for $f_y$ overlaps with all the slices above for all possible $g$,
      therefore the common part is:
      \[2^{r-\left(v-d - \ceil*{\frac{r-d-3-m'_y}{2}} -1\right)-k}2^{r-\ceil*{\frac{r-d-3-m'_y}{2}}-k-1}  = 2^{v+d},  \]
      which gives the required divisibility.
    \end{enumerate}
    
    \item $2f_x + 3 + m'_x < r-d$:\\
    In this case the number of common elements between the two slices is the number of elements in the slice for $f_y$ multiplied by the
    number of occurrences of each distinct element in the slice $f_x$, that is:
    \[ D = 2^{r +d - f_y - k -1 }2^{min(f_x+2, r-f_x-1)+min(m'_x, max(0, r-2f_x-3))-k} \]
    Now we need to consider the following two cases:
    \begin{enumerate}[label*=\arabic*., ref=\arabic{enumi}.\alph{enumii}.\arabic{enumiii}.\arabic*]
      \item \label{itm:3b21}$r-2f_x-3-m'_x \geq 0$:\\
      \[ D= 2^{r+d - f_y - k -1 + f_x+2 +m'_x  -k+d}  = 2^{v + d + m^*_x -f_y+f_x+1} .\]
      This case behaves exactly as the already considered case \ref{itm:3a21}.
      
      \item \label{itm:3b22} $ 0 > r-2f_x-3-m'_x$ \\
      \[ D = 2^{r +d - f_y - k -1 + r - f_x - k -1 } = 2^{2v -f_y - f_x -2 + d}. \]
      In this case $f_x = \ceil*{\frac{r-m'_x}{2}}-1 =
      \ceil*{\frac{v-m^*_x}{2}}-1$,  $f_y \leq \ceil*{\frac{r-m'_y-d}{2}}-1 = \ceil*{\frac{v-m^*_y-d}{2}}-1$. Let us
      substitute them into the equation (we take worst case for $f_y$)
      \[ D= 2^{2v -\ceil*{\frac{v-m^*_x}{2}} -\ceil*{\frac{v-m^*_y-d}{2}} + d }.  \]
      If either  $m^*_x, m^*_y$ or $d$ is greater than $0$ or $v$ is even, then this gives the required divisibility. Otherwise, we obtain 
      \[D =2^{2v -2\ceil*{\frac{v}{2}}}, \] 
      and $r-k$ is odd (since $v$ is odd). Therefore
      \[2^{2f_x+m'_x} = 2^{2f_y+m'_y} = 2^{2\ceil*{\frac{r-k}{2}}-2 +k} = 2^{r-1}. \]
      This means that the two overlapping slices we are looking at are both $\p*{2^{r-1}}$. It also means that $t'_y =0$, and that
      also slices containing $\p*{0}$ overlap.
      The number of common elements for the $\p*{0}$ slices is
      \[ 2^{2\left(\floor*{\frac{r+k}{2}}-k\right)} = 2^{2\floor*{\frac{v}{2}}}, \]
      and let us note that
      \[ 2^{2v -2\ceil*{\frac{v}{2}}} +  2^{2\floor*{\frac{v}{2}}} = 2^{2\floor*{\frac{v}{2}} + 1} = 2^{v}. \]
      This gives the desired divisibility.
    \end{enumerate}
  \end{enumerate}
  
  \item Now let us consider the cases where the slices that overlap are for certain $f_x$ and for $\p*{2^{r-d}h + t'_y}$. Once again, we
  need to look at two subcases:
  \begin{enumerate}[label*=\arabic*., ref=\arabic{enumi}.\alph{enumii}.\arabic*]
    \item \label{itm:3c1} $2f_x + 3 + m'_x \geq r-d$:\\
    The number of common elements is the number of elements in slice for $f_x$ multiplied by the number of occurrences of each element in the slice
    $\p*{2^{r-d}h + t'_y}$.
    \[
    2^{r-f_x-k-1}2^{\floor*{\frac{r+min(d+m'_y, r)}{2}}-k}  = 2^{v - f_x -1 + \floor*{\frac{r+min(d+m'_y, r) -2k}{2}}}
    = 2^{v - f_x -1 + \floor*{\frac{v+min(d+m^*_y, v) }{2}}}.
    \]
    If $f_x \leq \floor*{\frac{v+min(d+m^*_y, v)}{2}} -d -1$, then we obtain the required divisibility. Let us take now that $f_x
    \geq \floor*{\frac{v+min(d+m^*_y, v)}{2}} -d$. The slice $\p*{2^{r-d}h + t'_y}$ has a common element with the slice
    \[ \p*{2^{2f_x + m'_x}} = \p*{2^{2\left(\floor*{\frac{v+min(d+m^*_y, v)}{2}} -d\right)+ m'_x + g}}
    = \p*{2^{r-d+min(m^*_y, v-d) + m^*_x + s}}\]
    (for certain  $g\geq 0$, $s \geq -1$), and so it likewise overlaps with all other slices attainable from it by shifting
    by a multiple of the $2^{r-d}$ period. If $m^*_x=0$, $s=-1$ (\ie $v+min(d+m^*_y, v)$ is odd) and $min(m^*_y, v-d) =0$, it means that
    the only overlapping slice is
    \[\p*{2^{r-d+2}i + 2^{r-d-1}},\]
    since
    \[\p*{2^{r-d}h + t'_y} \]
    doesn't intersect with any other slice for any $h$. In this case $2f_x+m'_x = r-d-1
    \Leftrightarrow f_x = \ceil*{\frac{r-d-m'_x-1}{2}}$. 
    The number of common elements between those two slices is:
    \[
    2^{r - \ceil*{\frac{r-d-m'_x-1}{2}}-k -1 + \floor*{\frac{r+min(d+m'_y, r)}{2}}-k}
    = 2^{v - \ceil*{\frac{v-d-1}{2}} + \floor*{\frac{r+min(d+m'_y, r)-2k}{2}} - 1 }
    \]
    \[
    = 2^{v - \ceil*{\frac{v-d-1}{2}} + \floor*{\frac{v+d}{2}} - 1 }= 2^{v+d-1}.
    \]
    Because $v+min(d+m^*_y, v)$ is odd, it means that $min(d+m^*_y, v)= d+m^*_y$, $min(m^*_y, v-d) = m^*_y $, and therefore $m^*_y = 0$.
    Also, because $v+d+m^*_y$ is odd, we get that $v-d+m^*_y = r-d+m^*_y$ is odd, which means that $m'_y$ has different divisibility by
    $2$ than $r-d$. Additionally $v-d-m^*_y = r-d-m'_y$ is odd too. Therefore, if we take $f'_y =
    \ceil*{\frac{r-d-m'_y}{2}} - 1$, the slice for it will be
    \[\p*{2^{r-d}h + q'_y2^{r-d-1} + t'_y} \]
    (we know that $d+m'_y \leq r$ as $min(d+m^*_y, v)= d+m^*_y$).
    
    First, the slice for (our original) $f_y$ is $\p*{2^{r-d}h + t'_y}$ in this case, and the slice for $f_x$ is $\{2^{r-d+2}i  +
    2^{r-d-1}\}$, and they overlap---which means $2^{r-d} \nmid t'_y$  and $2^{r-d-1} \mid t'_y$. Slice for any higher $f'_x = f_x+c+1$,
    $c\geq0$ equals $\p*{ 2^{r-d+4+2c}i + 2^{r-d+1+2c}}$, and therefore it has to overlap with slice $\p*{2^{r-d}h + q'_y2^{r-d-1} + t'_y}$.
    Let us count the size of this overlap:
    
    \[2^{r-\left(\ceil*{\frac{r-d-m'_x-1}{2}} + 1\right)-k}2^{r - \ceil*{\frac{r-d-m'_y}{2}}-k}  = 2^{v+d-1}  \]
    Together with the $2^{v+d-1}$ we obtained earlier, this gives the required divisibility. Let us note that this case is different from case
    \ref{itm:3b11} with $s=-1$ and $m^*_x=m^*_y=0$, as here the slice $\p*{2^{r-d+2}i + 2^{r-d-1}}$ overlaps with $\p*{2^{r-d}h + t'_y}$, whereas
    there the slice $\p*{2^{r-d+2}i + 2^{r-d-1}}$ overlapped with $\p*{2^{r-d}h + q'_y2^{r-d-1} + t'_y}$.
    
    Now let us suppose that $min(m^*_y, v-d) \geq 1$ or $m^*_x \geq 1$ or $s\geq0$. Then the slice $\p*{2^{r-d}h + t'_y}$ overlaps with some slice
    $\p*{2^{r-d + z_0}}$ for certain $z_0 \geq 0$, which means that it also overlaps with slice $\p*{0}$, and we are in the case \ref{3f}.
    \item $2f_x + 3 + m'_x < r-d:$ \\
    The number of common elements is the number of elements in slice $\p*{2^{r-d}h + t'_y}$ multiplied by number of occurrences of each
    element in the slice for $f_x$:
    \[
    D = 2^{\floor*{\frac{r+min(d+m'_y, r)}{2}} - k+d} 2^{min(f_x+2, r-f_x-1)+min(m'_x, max(0, r-2f_x-3))-k}.
    \]
    Let us consider the following two cases:
    \begin{enumerate}[label*=\arabic*., ref=\arabic{enumi}.\alph{enumii}.\arabic{enumiii}.\arabic*]
      \item \label{itm:3c21} $r-2f_x-3-m'_x \geq 0$:\\
      \[
      D= 2^{\floor*{\frac{r+min(d+m'_y, r)}{2}} - k+d + f_x + 2 + m'_x  -k}
      = 2^{\floor*{\frac{v+min(d+m^*_y, v)}{2}} +d + f_x + 2 + m^*_x }.
      \]
      If $min(d+m^*_y, v) = v$ we automatically have the required divisibility. Let us assume now the other case, when the above becomes
      \[
     	D = 2^{\floor*{\frac{v+d+m^*_y}{2}} +d + f_x + 2 + m^*_x }.
      \]
      If
      \[
      \floor*{\frac{v+d+m^*_y}{2}}+ d+ f_x+2 +m^*_x \geq v+d \Leftrightarrow f_x \geq v -\floor*{{\frac{v+d+m^*_y}{2}}} - 2 -
      m^*_x
      \]
      \[
      \Leftrightarrow f_x \geq \ceil*{{\frac{v-d-m^*_y}{2}}} - 2 - m^*_x \Leftrightarrow f_x \geq \ceil*{{\frac{r-d-m'_y}{2}}} - 2 - m^*_x,
      \]
      then we automatically obtain the required divisibility. Let us assume now that
      \[f_x \leq \ceil*{{\frac{r-d-m'_y}{2}}} - 3 - m^*_x .\]
      In this case, there is a certain $f_y$ with which slice the slice for our $f_x$ overlaps. Because the slice for $f_x$ overlaps with
      $\p*{2^{r-d}h+ t'_y}$, and has a period of $2^{2f_x+m'_x+3}$, then it also contains any element of the form
      
      \[\p*{2^{2f_x+m'_x+3}i + t'_y} = \p*{2^{2\ceil*{\frac{r-d-m'_y}{2}}- 3 - 2m^*_x-h+m'_x}i + t'_y}\]\[ =
      \p*{2^{2\ceil*{\frac{r-d-m'_y}{2}}- 3 -m^*_x-h+k}i + t'_y}\]

      for certain $h \geq -1$ and any $i$. The maximum possible $f_y$ equals $\ceil*{{\frac{r-d-m'_y}{2}}} - 1$, and its slice contains
      \[
      \p*{ 2^{2\ceil*{\frac{r-d-m'_y}{2}} -2 + m'_y} + t'_y  } =\left\{ 2^{2\ceil*{\frac{r-d-m'_y}{2}} -2 + m^*_y + k} + t'_y
      \right\}.
      \] 
      Therefore the slices for our $f_x$ and for the maximum possible $f_y$ overlap, which means we are in the already-considered
      case \ref{itm:3b21}---we just approached it here from the ``other end''.
      
      \item \label{itm:3c22} $ 0 > r-2f_x-3-m'_x$ \\
      \[ D = 2^{\floor*{\frac{r+min(d+m'_y, r)}{2}} - k+d + r - f_x - k -1}
      = 2^{v + \floor*{\frac{v+min(d+m^*_y, v)}{2}} +d - f_x -1 }.
      \]
      In this case $f_x = \ceil*{\frac{r-m'_x}{2}}-1 = \ceil*{\frac{v-m^*_x}{2}}-1$, let us substitute:
      \[ D = 2^{v +d + \floor*{\frac{v+min(d+m^*_y, v)}{2}}  - \ceil*{\frac{v-m^*_x}{2}}   } \]
      This gives the desired divisibility unless $m^*_y = m^*_x = d = 0$ and $v$ is odd, in which case we have
      \[ D = 2^{v -1}.\]
      Let us take $f_y$ to be maximal, \ie $f_y = \ceil*{\frac{r-m'_y}{2}}-1$. When $v$ is odd and $m^*_y = m^*_x = 0$, then also $r-k$ is odd, and
      additionally
      
      \[2^{2f_y+m'_y} = 2^{2f_x+m'_x} = 2^{2\ceil*{\frac{r-k}{2}}-2 +k} = 2^{r-1}.\]
      
      Because $\p*{2^{r}h+ t'_y}$  overlaps with the slice for $f_x$, it means that $t'_y = 2^{r-1}$. Therefore also the slice
      $\p*{0}$ overlaps with the slice for the maximal $f_y$, \ie $\p*{2^{r}h + 2^{2f_y+m'_y}+ t'_y}$. Owing to symmetry (which is here due to
      $d=0$), the size of that overlap also equals $2^{v -1}$, which together with the earlier overlap gives the desired $2^{v}$.
    \end{enumerate}
  \end{enumerate}
  \item Let us consider now the case where the slices that overlap are $\p*{0}$ and a certain slice for $f_y$ where $f_y \leq
  \ceil*{\frac{r-d-3-m'_y}{2}}-1$. This case is very similar to \ref{itm:3c21}. The number of common elements is the number of elements
  in the slice $\p*{0}$ multiplied by the number of occurrences of each element in slice for $f_y$, that is
  \[
  2^{\floor*{\frac{r+m'_x}{2}}-k} 2^{min(f_y+2, r-f_y-1)+min(m'_y, max(0, r-2f_y-3))-k+d}
  = 2^{\floor*{\frac{r+m'_x}{2}}-k + f_y + 2 +m^*_y +d},\]
  since $f_y \leq \ceil*{\frac{r-d-3-m'_y}{2}}-1$.
  
  If
  \[
  f_y \geq \ceil*{{\frac{r-m'_x}{2}}} - 2 - m^*_y,
  \]
  then we automatically obtain the required divisibility. Let us assume now that
  \[f_y \leq \ceil*{{\frac{r-m'_x}{2}}} - 3 - m^*_y .\]
  In this case, there is a certain $f_x$ with which slice the slice for our $f_y$ overlaps. Because the slice for $f_y$ overlaps with
  $\p*{0}$, and has the period of $2^{2f_y+m'_y+3}$, then it also contains any element of the form
  
  \[\p*{2^{2f_y+m'_y+3}i} = \p*{2^{2\ceil*{\frac{r-m'_x}{2}}- 3 - 2m^*_y-h+m'_y}i } =
  \p*{2^{2\ceil*{\frac{r-m'_x}{2}}- 3 -h -m^*_y+k}i }\]

  for certain $h \geq -1$ and any $i$. The maximum possible $f_x$ equals $\ceil*{{\frac{r-m'_x}{2}}} - 1$, and its slice contains
  \[\p*{ 2^{2\ceil*{\frac{r-m'_x}{2}} -2 + m'_x} } =\p*{ 2^{2\ceil*{\frac{r-m'_x}{2}} -2 + m^*_x + k} }.   \]
  Therefore the slice for our $f_y$ and for the maximal $f_x$ overlap, which means we are in the already-considered case
  \ref{itm:3a}---we just approached it here from the ``other end''.
  
  \item \label{3e} Let us assume now that there is a common element between the slice $\p*{0}$  and a certain slice for $
  \ceil*{\frac{r-d-m'_y}{2}}-1 \geq f_y \geq \ceil*{\frac{r-d-3-m'_y}{2}}$.
  
  The slice for $f_y$ is
  \[
  \p*{2^{r-d}h + q'_y2^{2f_y + m'_y} + t'_y} = \p*{2^{r-d}h + q'_y2^{2\left(\ceil*{\frac{r-d-m'_y}{2}}-g\right) + m'_y} + t'_y  }
  \]
  \[
  = \p*{2^{r-d}h + q'_y2^{r-d-s} + t'_y},
  \]
  where $g=1$ or $g=2$, and $s=1$ or $s=3$ (when $r-d-m'_y$ odd) or $s=2$ (when $r-d-m'_y$ even).
  Because $\p*{0}$ overlaps with $\p*{2^{r-d}h + q'_y2^{r-d-s} + t'_y}$ it means that $2^{r-d} | q'_y2^{r-d-s} + t'_y $, and that the second
  slice is just $\p*{2^{r-d}h}$.
  
  Any slice
  \[ \p*{2^{2f_x+3+m'_x}i + 2^{2f_x+m'_x}  }  \]
  overlaps with the slice
  \[\p*{2^{r-d}h}\]
  whenever
  \[ 2f_x+m'_x \geq r-d \Leftrightarrow f_x \geq \ceil*{\frac{r-d-m'_x}{2}} .\]
  
  Let us count number of common elements of those overlaps:
  \[2^{r-\ceil*{\frac{r-d-m'_x}{2}} -k}2^{r - \left(\ceil*{\frac{r-d-m'_y}{2}}-g\right)-k-1}
  =2^{2v-\ceil*{\frac{v-d-m^*_x}{2}} - \ceil*{\frac{v-d-m^*_y}{2}}-1+g}. \]
  It gives the desired divisibility unless $g=1$, $m^*_y = m^*_x =0$ and $v-d$ is odd.
  Yet in that case let us notice that the slice for $f'_x = \ceil*{\frac{r-d-m'_x}{2}} - 1 = \ceil*{\frac{v-d}{2}} - 1$ and
  the slice $\p*{ 2^{r-d}h + t'_y}$ are respectively:
  \[ \p*{ 2^{2f'_x+3+m'_x}i + 2^{r-d-1}}   \text{\;\;\;and\;\;\;}   \p*{ 2^{r-d}h + 2^{r-d-1}}   \]
  ($q'_y$ is always odd, in this case $s=1$ and therefore $t'_y =2^{r-d-1} $). This means that those two slices overlap, which indicates
  this is the same case as \ref{itm:3c1} for $m^*_x=m^*_y=0$ and $v + min(d + m^*_y,v)$ being odd.
  
  \item \label{3f} Finally let us consider the case where the slice $\p*{0}$ overlaps with  the slice $\p*{2^{r-d}h + t'_y}$.\\
  This case is very similar to the previous one, yet different. Because those two slices overlap it means that $2^{r-d} | t'_y$, and in fact
  the second slice is just $\p*{2^{r-d}h}$.
  
  Any slice
  \[ \p*{2^{2f_x+3+m'_x}i + 2^{2f_x+m'_x}  }  \]
  overlaps with the slice
  \[\p*{2^{r-d}h}\]
  whenever
  \[ 2f_x+m'_x \geq r-d \Leftrightarrow f_x \geq \ceil*{\frac{r-d-m'_x}{2}} .\]
  Let us count number of common elements of those overlaps:
  \[2^{r-\ceil*{\frac{r-d-m'_x}{2}} -k}2^{\floor*{\frac{r+min(d+m'_y, r)}{2}}-k}  =2^{\floor*{\frac{r+min(d+m'_x, r)}{2}} -k
  + \floor*{\frac{r+min(d+m'_y, r)}{2}}-k} \]
  \[
  = 2^{\floor*{\frac{v+min(d+m^*_x, v)}{2}} + \floor*{\frac{v+min(d+m^*_y, v)}{2}}}
  \]
  If $min(d+m^*_x, v) = v$, it is easy to check this gives the required divisibility. Let us assume the opposite
  case, obtaining
  \[
  = 2^{\floor*{\frac{v+d+m^*_x}{2}} + \floor*{\frac{v+d+m^*_y}{2}}}.
  \]
  This gives the required divisibility, unless $m^*_y = m^*_x=0$ and $v+d$ is odd.
  Yet in that case let us notice that the slices for $f'_x = \ceil*{\frac{r-d-m'_x}{2}} - 1 = \ceil*{\frac{v-d}{2}} - 1$ and
  $f'_y = \ceil*{\frac{r-d-m'_y}{2}} - 1 = \ceil*{\frac{v-d}{2}}-1$ are respectively:
  \[ \p*{ 2^{2f'_x+3+m'_x}i + 2^{r-d-1}}   \text{\;\;\;and\;\;\;}   \p*{ 2^{r-d}h + 2^{r-d-1}}   \]
  ($q'_y$ is always odd), which means those two slices overlap. And this means that this is the same case as \ref{itm:3b11} for
  $m^*_x=m^*_y=0$ and $r-d-m'_y$ odd, and it was already considered.
\end{enumerate}

To sum up, we showed that whenever there is any overlap, then either it is divisible by $2^{v+d}$, or there are some additional overlaps
which sum together to that divisibility. We also never counted the same common elements of images $P$ and $Q$ more than once, which can be
seen via comparison of cases that refer to each other.
\end{enumerate}
\end{proof}

\begin{repCol} {slices}
\emph{\textbf{(restated)}}

Let $P(x) = a_x x^2 + b_x x +c_x$ and $Q(y, h) = a_y y^2 + b_y y + c_y +  2^{r-d}h$ where $d \leq r$. Then for any $ v, q \in [d, r]$, 
when we work over $\ZZ_{2^r}$ it holds that:
\[
2^{min(v,q)+d} \;|\; \#\left\{(x, y, h) : x \in \bigcup\limits_{j=0}^{2^v-1} \{l_x+
2^{r-v}j \}, y \in \bigcup\limits_{j=0}^{2^q-1} \p*{l_y+ 2^{r-q}j} , h \in \bigcup\limits_{j=0}^{2^d-1} \p*{j}, P(x)  = Q(y, h) \right\}.
\]
for any $l_x, l_y$.
\end{repCol}
\begin{proof}
Let let us take $M = max(v,q)$ and $m = min(v,q)$. Then we can split the larger of the two domains into $2^{M-m}$ slices each of some
form 
\[
\bigcup\limits_{j=0}^{2^m-1} \p*{l +  2^{r-m}j}
\]
for appropriately shifted $l$'s. Next we can use Lemma \ref{lin} $2^{M-m}$ times, for each of those smaller slices of the larger domain
separately. Then we notice that each time we obtain that the size of the intersection is divisible by $2^{m+d}$, and the divisibility of a
sum is not smaller than the smallest of the divisibilities of its components. Finally, to obtain any $c_x$ and $c_y$ we wish, we can just choose
$c = c_x-c_y$ in Lemma \ref{lin}, and then add $c_y$ to both sides of the equation.
\end{proof}

When proving the base case of our main theorem we have come against a specific multiset that a polynomial we will have may
potentially intersect with. The following lemma contributed to the divisibility of the size of such an  intersection.

\begin{lemma} \label{os}
Let us work over  $\ZZ_{2^r}$. Let $P(x) = a_x x^2 + b_x x +c$ with $x$ being constrained to domain $ x \in \bigcup\limits_{j=0}^{2^v-1}
\p*{l_x+  2^{r-v}j}$ for certain $l_x$ and $v \leq r$. Let $S$ be the following multiset:
\[
S = \left(\bigcup_{i=0}^{2^e-1} \bigcup_{f_s=0}^{v-1} f_s \bigcup_{s=0}^{2^{v-f_s-1}-1}  \p*{2^{r-e}i + 2^{r-v+f_s}(2s+1) } \right)
\cup \left(\bigcup_{i=0}^{2^e-1} (v+1) \p*{2^{r-e}i } \right),
\]
where $e \leq v$. The number of elements of the intersection (understood as a ``multiplicative'' intersection as described above) of the
multiset $S$ and the image of $P$ is divisible by \[2^{e + min(v, \ceil*{\frac{r}{2}})} .\]

\end{lemma}
\begin{proof}
The multiset $S$ can be also written as
\[
S = \left(\bigcup_{f_s=0}^{v-e-1} f_s2^{e} \bigcup_{s=0}^{2^{v-f_s-1}-1}  \p*{ 2^{r-v+f_s}(2s+1) } \right)
\cup  \left(\bigcup_{i=0}^{2^e-1} \bigcup_{f_s=v-e}^{v-1} f_s2^{v-f_s-1}  \p*{2^{r-e}i } \right)
\cup \left( \bigcup_{i=0}^{2^e-1} (v+1) \p*{2^{r-e}i } \right) 
\]
\[
=\left(\bigcup_{f_s=0}^{v-e-1} f_s2^{e} \bigcup_{s=0}^{2^{v-f_s-1}-1}  \p*{ 2^{r-v+f_s}(2s+1) } \right)
\cup \left( 2^{e}(v-e+1) \bigcup_{i=0}^{2^e-1}\p*{2^{r-e}i } \right),
\]
where we used the formula for the sum of an arithmetic-geometric series on $\bigcup_{f_s=v-e}^{v-1} f_s2^{v-f_s-1}$. From now on we will
refer to set $S$ as presented in the last of the written above forms.

$P$ may belong to one of the categories a), b) or c) as per Corollary \ref{csqr}. Let us go through cases on those categories:

\begin{enumerate}[label*=\arabic*.]
  \item $P$ is in category a).\\
  The image of $P$ is 
  	\[ 2^{m'_x-k} \bigcup_{i=0}^{2^{r-m'_x}-1} \p*{ 2^{m'_x} i + c'} \]
  	If it has no common elements with $S$ we are automatically done.  When the contrary is true, let us first assume that a common element
  	is in a slice 
  	\[\p*{ 2^{r-v+f_s}(2s+1) }  \]
  	for certain $f_s$. Let us consider two cases
  	\begin{enumerate}[label*=\alph*.]
  	  \item $m'_x > r-v+f_s$ \\
  	  In this case only the elements belonging to the slice for $f_s$ may be in the intersection. The size of the common part equals the
  	  number of all elements in the $P$'s image multiplied by the number of occurrences of each distinct element in the slice for $f_s$, giving
  	  \[2^{v}f_s2^{e} = f_s2^{v+e},\]
  	  which has the required divisibility.
  	  \item $m'_x \leq r-v+f_s$ \\ 
  	  	In this case the of $P$ intersects with all elements of all slices for any $f'_s$ such that $m'_x \leq r-v+f'_s$, and also of 
  	  	the slice $\p*{2^{r-e}i }$. The size of that overlap is the number of all those elements multiplied by the
  	  	number of occurrences of each distinct element from the image of $P$. The number of elements in the slices for all $f_s \geq m^*_x$,
  	  	and the last slice $\p*{2^{r-e}i}$, is 
  	  	\[
  	  	2^e((m^*_x+1)2^{v-m^*_x} - (v-e+1) 2^{e} + (v-e+1) 2^{e}) = (m^*_x+1) 2^{e+v-m^*_x} 
  	  	\] 
  	  	(via use of the formula for an arithmetic-geometric series sum).  Therefore the size of the intersection is
  	  	\[
  	  	(m^*_x+1) 2^{e+v-m^*_x} 2^{m'_x-k} = (m^*_x+1) 2^{v+e},
  	  	\]
  	  	which has the required divisibility.  	  	
  	\end{enumerate}
    Finally we need to consider the scenario in which the common elements are in the slice $\p*{2^{r-e}i}$, and there are no common elements
    in any other slice of $S$.
    This means that 
    \[
    m'_x = r-e+h
    \]
  	for certain $h \geq 0$. In this case the size of the intersection equals to number of all elements of the image of $P$ multiplied by the
  	number of occurrences of each distinct element of the slice $\p*{2^{r-e}i}$, giving
  	\[ 
  	2^v (v-e+1)2^{e} = (v-e+1)2^{v+e},
  	\]
    which has the required divisibility.
  	
 \item $P$ is in category b).\\
  	This case can be reduced to the previous one by taking $m'_x$ to be bigger by $1$, but not bigger than $r$.

 \item $P$ is in category c).\\
Because $P$ is of type c) we know that $m'_x \geq 2k + h_x$ for certain $h_x \geq 0$ (compare with $a'$ from Corollary \ref{csqr}), and
also $m'_x < r$.
Therefore
\[2k +h_x < r \Leftrightarrow 2(r-v) + h_x < r \Leftrightarrow v > \ceil*{\frac{r+h_x}{2}},\]
which means that it will be sufficient to show the divisibility by $2^{e + \ceil*{\frac{r}{2}}}$.
In this case the image of $P$ is
{\small
\[
\left(\bigcup_{f_x=0}^{\ceil*{\frac{r-m'_x}{2}}-1} 2^{min(f_x+2, r-f_x-1)+min(m'_x, max(0, r-2f_x-3))-k} \bigcup_{i=0}^{max(0,
2^{r-2f_x-3-m'_x}-1)} \p*{2^{2f_x+3+m'_x}i + q'2^{2f_x+m'_x} + t'} \right)\]
\[ \cup\; 2^{\floor*{\frac{r+m'_x}{2}}-k}\p*{t'},
\]
}whereas $S$ is, once again,
\[\left(\bigcup_{f_s=0}^{v-e-1} f_s2^{e} \bigcup_{s=0}^{2^{v-f_s-1}-1}  \p*{ 2^{r-v+f_s}(2s+1) } \right)
\cup \left( 2^{e}(v-e+1) \bigcup_{i=0}^{2^e-1}\p*{2^{r-e}i } \right).\]

If there is no overlap between $S$ and the image of $P$, then we automatically obtain the desired result. Let us assume now that there is
an overlap. Our strategy will be to go through all possible cases of such overlaps and then show that either the overlap has required
divisibility on its own, or that there have to be some additional overlaps that together with the initial one have the desired
divisibility. We will also never use any overlap more than once to complement an overlap not having sufficient divisibility on its
own. Let us proceed to the cases:
\begin{enumerate}[label*=\alph*., ref=\arabic{enumi}.\alph*]
  \item \label{itm:t0} There is an overlap between the slices $\p*{t'}$ and $\p*{2^{r-e}i}$.\\
  In this case the slice $\p*{2^{r-e}i}$ overlaps with all slices for $f_x$ such that
  $2f_x+m'_x \geq r-e \Leftrightarrow f_x \geq \ceil*{\frac{r-e-m'_x}{2}}$. The size of that overlap is the number of elements in
  slices for all such $f_x$ and $\p*{t'}$, multiplied by the number of occurrences of each distinct element in the slice $\p*{2^{r-e}i}$.
  Using Lemma \ref{multisum} this gives
  \[2^{r-\ceil*{\frac{r-e-m'_x}{2}}-k} 2^{e}(v-e+1) =
  (v-e+1) 2^{e+\floor*{\frac{r+e+m'_x}{2}}-k  } = (v-e+1) 2^{e+\floor*{\frac{r+e+h_x}{2}}}.\]
  
  The above has the required divisibility unless $h_x=e=0$, $r$ is odd, and $v$ is even. Yet, it is not a sole overlap when that
  happens.
  The pairs of slices for any $f_x$ such that $2f_x+m'_x <o(t')$ and $f_s$ such that $r-v+f_s = 2f_x+m'_x$ also overlap,
  since $2^{r-v+f_s+1} | t'$ due to $f_s \leq v-e-1$ and $2^{r-e}|t'$. Because we require
  
  \[  2f_x+m'_x <o(t') \Leftrightarrow 2f_x< o(t')-m'_x\]
  
  and
  \[  f_s = 2f_x+m^*_x < v-e \Leftrightarrow 2f_x < v-e, \]
  let us denote $min(o(t')-m'_x, v-e-m^*_x)$ as $M$ so that $\ceil*{\frac{M}{2}}-1$ will be the limit of our sum.
  Let us count now the size of those pairwise
  overlaps. We will compute it for general $h_x, e$
  etc.\ first, since that computation will be also useful for us later:
  
  \[\sum_{f_x=0}^{\ceil*{\frac{M}{2}}-1}  2^{r - f_x -k -1} f_{s} 2^{e}
  = 2^{e} \sum_{f_x=0}^{\ceil*{\frac{M }{2}}-1}  (2f_x+m^*_x) 2^{v - f_x -1}\]
  \begin{align*}
= 2^{e}
\left(\frac{\left(2\ceil*{\frac{M }{2}}-2+m^*_x\right)2^{v - \ceil*{\frac{M}{2}}
-1}}{-\frac{1}{2}} - \frac{m^*_x 2^{v-1}}{-\frac{1}{2}}
- \frac{2\left(2^{v - \ceil*{\frac{M }{2}}
-1} - 2^{v-2} \right)}{\frac{1}{4}}  \right),
\end{align*}
(where we used the formula for sum of the arithmetic-geometric series)
\begin{align*}
= 2^{e+1}
\left( m^*_x2^{v-1} + 2^{v}
- \left(2\ceil*{\frac{M}{2}}-2+m^*_x\right)2^{v - \ceil*{\frac{M}{2}}-1}
- 2^{v -\ceil*{\frac{M}{2}}+1} \right)
\end{align*}
\begin{align*}
= 2^{e+1}
\left( \left(m^*_x + 2 \right)2^{v-1}
- \left(min(o(t')-k, v-e)+g -2\right)2^{v - \ceil*{\frac{M}{2}}-1}
- 2^{v - \ceil*{\frac{M}{2}}+1}
\right)
\end{align*}
\[ = 2^{e}
\left( \left(m^*_x + 2 \right)2^{v}
- \left(min(o(t')-k, v-e)+g+2\right)2^{v - \ceil*{\frac{M}{2}}} \right) \]
for certain $g=0$ or $1$  depending on the parity of $M$. The first element of this difference already has the required divisibility, so let
us focus on the second one:
\[ (min(o(t')-k, v-e)+g+2) 2^{e+v - \ceil*{\frac{min(o(t')-k, v-e-m^*_x)}{2}}} \]\[
=  (min(o(t')-k, v-e)+g+2) 2^{e - \ceil*{\frac{-2v + min(o(t')-k, 2v-r-e-h_x)}{2}}} \]
\[  =  (min(o(t')-k, v-e)+g+2) 2^{e + \floor*{\frac{ max(r + v -o(t'), r+e+h_x)}{2}}}.
\]
Let us use now our assumptions (\ie $h_x=e=0$, $r$ is odd, and $v$ is even, $t'=0$ due to $e=0$), which continue the calculation as:
\[  =  (v+g+2) 2^{\floor*{\frac{r}{2}}} =  (v+3) 2^{\floor*{\frac{r}{2}}},
\]
since $g=1$ owing to $M = min(o(t')-m'_x, v-e-m^*_x) = min(2v-r, 2v-r)$ being odd. Let us add it now to the earlier-found overlap (for the
$\p*{2^{r-e}i}$ slice):
\[   (v+3) 2^{\floor*{\frac{r}{2}}} + (v+1) 2^{\floor*{\frac{r}{2}}} = (v+2) 2^{\floor*{\frac{r}{2}}+1}.
\]
This gives the desired divisibility.

\item \label{itm:ts} There is an overlap between the slice $\p*{t'}$,\ and a slice for certain $f_s \leq v-e-1$.\\
The slice for $f_s$ is
\[\p*{ 2^{r-v+f_s+1}s+ 2^{r-v+f_s} },\]
and for it to overlap with $\p*{t'}$ it is necessary that
\[t' = 2^{r-v+f_s+1}j + 2^{r-v+f_s} \]
for certain $j$. Let us also notice that in this case the slice for $f_s$ overlaps with any slice  $f'_x$ such that $2f'_x+m'_x
\geq r-v+f_s+1$, as any such slice is:
\[\p*{2^{2f_x+3+m'_x}i + q'2^{2f_x+m'_x} + t' }
= \p*{2^{r-v+f_s+1+d +3}i + q'2^{r-v+f_s+1+d} + 2^{r-v+f_s+1}j + 2^{r-v+f_s} }.  \]
\[=\p*{2^{r-v+f_s+1}(2^{d +3}i + q'2^{d} + j)    + 2^{r-v+f_s}}\]
for certain $d \geq 0$. The size of that whole overlap is the number of elements of all slices for
$f'_x \geq \ceil*{\frac{r-v+f_s+1-m'_x }{2}}$ multiplied
by the number of occurrences of each distinct element in the slice $f_s$, that is:
\[2^{r-\ceil*{\frac{r-v+f_s+1-m'_x }{2}}-k}f_s2^e = f_s2^{e + v -\ceil*{\frac{f_s+1-m^*_x}{2}}}
\]\[= f_s2^{e -\ceil*{\frac{-2v+f_s+1-r+v-h_x}{2}}}
= f_s2^{e +\floor*{\frac{r+v-f_s-1+h_x}{2}}}   \]
This gives the required divisibility unless $f_s=v-1$, $h_x =0$ and $r$ is odd. In that unfortunate case, we get
that $t' = 2^{r-1}$ and $e=0$, and there is also an additional overlap ``group''.

Let us take the maximal $f'_x$, \ie $2f'_x + m'_x = r-1$ ($r$ is odd, $h_x =0  \Rightarrow m'_x$ is even),
and look at its slice
\[ \p*{2^{2f'_x+3+m'_x}i + q'2^{2f'_x+m'_x} + t' }
= \p*{ q'2^{r-1} + 2^{r-1} }  = \p*{ 0 } \]
Therefore it overlaps with the slice $\p*{2^{r-e}i} = \p*{ 0}$ (since $e=0$). We have that $f'_x = \frac{2v-r-1}{2}$
and the size of that overlap is
\[ 2^{r - \frac{2v-r-1}{2} -k -1} 2^{e}(v-e+1) =(v+1) 2^{\frac{r-1}{2} }  .  \]
Let us add it to the earlier-found overlap:
\[ (v+1) 2^{\frac{r-1}{2} } + f_s2^{e +\floor*{\frac{r+v-f_s-1+h_x}{2}}}
=  (v+1) 2^{\frac{r-1}{2} } + (v-1)2^{\floor*{\frac{r}{2}}} = v 2^{\ceil*{\frac{r}{2}}}.\]
This gives the required divisibility.
\item \label{itm:x0} There is an overlap between a slice for certain $f_x \leq \ceil*{\frac{r-m'_x}{2}}-1$ and the slice $\p*{ 2^{r-e}i}$.\\
The slice for $f_x$ is
\[\p*{ 2^{2f_x+3+m'_x}i+ q'2^{2f_x+m'_x} + t'} \]
and for it to overlap with the slice $\p*{2^{r-e}i}$ it is necessary that
\[t' = 2^{min(2f_x+3+m'_x, r-e)}j - q'2^{2f_x+m'_x} \]
for certain $j$. Let us consider three sub-cases:
\begin{enumerate}[label*=\arabic*.]
  \item $r-e \leq 2f_x+m'_x$ \\
  In this case we have
  \[t' = 2^{r-e}j'. \]
  for certain $j'$.
  It is easy to notice that the slice $\p*{2^{r-e}i}$ overlaps with slice for any $f'_x$ such that $2f'_x + m'_x \geq r-e$ and with
  the slice $\p*{t'}$. Let us count the size of this overlap as a whole, which equals to number of all elements in this overlap from
  the image of P multiplied by number of occurrences of any distinct element in the slice $\p*{2^{r-e}i}$.
  \[2^{r-\ceil*{\frac{r-e-m'_x}{2}}-k} 2^{e}(v-e+1) =(v-e+1) 2^{e + \floor*{\frac{r+e+h_x}{2}}} \]
  which gives the required divisibility unless $e=h_x=0$, $v$ is even, and $r$ is odd. Yet, because in such a case we have that
  $\p*{2^{r-e}i}$ intersects with $\p*{t'}$, it is a situation that we already considered in the case \ref{itm:t0}.
  \item $ 2f_x+3+m'_x> r-e > 2f_x+m'_x$ \\
  Now we have
  \[t' = 2^{r-e}j - q'2^{2f_x+m'_x}  \text{\;\;\;\;\;\;\;\;\;\; and \;\;\;\;\;\;\;\;\;\;\;\;\;}  f_x =
  \ceil*{\frac{r-e-m'_x}{2}}-1.\]
  
  Let us count the size of the overlap between the slice for $f_x$ and  $\p*{2^{r-e}i}$:
  \[2^{r-f_x-k-1}2^{e}(v-e+1) = (v-e+1)2^{e + r - \ceil*{\frac{r-e-m'_x}{2}}+1-k-1 }
  = (v-e+1)2^{e + \floor*{\frac{r+e+h_x}{2}}}.
  \]
  This gives the required divisibility unless $e=h_x=0$, $v$ is even, and $r$ is odd, which we assume now. Therefore, we get $t' =
  2^{r-1}$.
  Let us notice that the slice for $f'_s$
  such that $r-v+f'_s = 2f_x+m'_x =r-1$ is:
  \[\p*{2^{r-1} }  \]
  which overlaps with the slice $\p*{t'} = \p*{ 2^{r-1}}$. This means that this is the case \ref{itm:ts}, which we alraedy considered
  (we just arrived at it from another end).
  \item $ r-e \geq 2f_x+3+m'_x$ \\
  Now we have
  \[t' = 2^{2f_x+3+m'_x}j - q'2^{2f_x+m'_x}.\]

  In this case the slice for $f_x$ is
  \[ \p*{2^{2f_x+3+m'_x}i + q'2^{2f_x+m'_x} + t'}   =\p*{2^{2f_x+3+m'_x}i} \]
  and it is easy to notice that it overlaps with all slices for $f_s$ such that $r-v+f_s \geq 2f_x+3+m'_x$ and with the slice
  $\p*{2^{r-e}i}$. The size of this
  overlap is the number of elements in the slices for those $f_s$ multiplied by the number of occurrences of each distinct element
  in the slice for $f_x$.  The number of elements in the slices for all $f_s \geq 2f_x+3+m^*_x$, and the last slice
  $\p*{2^{r-e}i}$, is \[2^e((2f_x+4+m^*_x)2^{v-2f_x-3-m^*_x} - (v-e+1) 2^{e} + (v-e+1) 2^{e}) = (2f_x+4+m^*_x)2^{e+v-2f_x-3-m^*_x}
  \] (via use of formula for arithmetic-geometric sequence sum), while the number of repetitions of any element in the slice for
  $f_x$ is \[ 2^{min(f_x+2, r-f_x-1)+min(m'_x, max(0, r-2f_x-3))-k} = 2^{f_x+2+m'_x-k},  \]
  since $r \geq 2f_x+3+m'_x$. The size of the overlap is:
  \[(2f_x+4+m^*_x)2^{e+v-2f_x-3-m^*_x}2^{f_x+2+m'_x-k} = (2f_x+4+m^*_x) 2^{e+v-f_x-1} \]
  Because $f_x = \floor*{\frac{r-e-m'_x-3}{2}}- d$ for certain non-negative $d$, we have
  \[= (v-e+1-2d) 2^{e+v-\floor*{\frac{r-e-m'_x-3}{2}}-1+d}
  = (v-e+1-2d) 2^{e+\ceil*{\frac{2v -r+e+m'_x+1}{2}}+d}
  \]
  \[ =   (v-e+1-2d) 2^{e+\ceil*{\frac{r+e+h_x+1}{2}}+d},\]
  which gives the required divisibility.
\end{enumerate}
\item \label{itm:xs} There is an overlap between a slice for certain $f_x \leq \ceil*{\frac{r-m'_x}{2}}-1$, and a slice for certain
$f_s \leq v-e-1$.\\
Let us consider following sub-cases:
\begin{enumerate}[label*=\arabic*.]
  \item $2f_x+m'_x \geq  r-v+f_s+1$ \\
  In this case $2^{2f_x+m'_x} = 2^{r-v+f_s+1+d}$ for certain $d \geq 0$. The slice for $f_x$ is
  \[\p*{2^{2f_x+3+m'_x}i + q'2^{2f_x+m'_x} + t' } = \p*{2^{r-v+f_s+1+d +3}i + q'2^{r-v+f_s+1+d} + t' }.  \]
  For it to overlap with the slice
  \[\p*{ 2^{r-v+f_s+1}s+ 2^{r-v+f_s} },\]
  it is necessary that
  \[t' = 2^{r-v+f_s+1}j + 2^{r-v+f_s} \]
  for certain $j$. Let us also notice that in this case the slice for $f_s$ overlaps with any slice  $f'_x$ such that $2f'_x+m'_x
  \geq r-v+f_s+1$ (if $f_x$ was the smallest of them, then we would have $d=0$ or $d=1$). That also includes an overlap with the slice
  $\p*{t'}$.
  The size of that whole overlap is the number of elements of all slices for
  $f'_x \geq \ceil*{\frac{r-v+f_s+1-m'_x }{2}}$ multiplied
  by the number of occurrences of each distinct element in the slice $f_s$, that is:
  \[2^{r-\ceil*{\frac{r-v+f_s+1-m'_x }{2}}-k}f_s2^e = f_s2^{e + v -\ceil*{\frac{f_s+1-m^*_x}{2}}}\]\[
  = f_s2^{e + \floor*{\frac{2v-f_s-1+m^*_x}{2}}} = f_s2^{e + \floor*{\frac{r+v-f_s-1+h_x}{2}}}.
  \]
  The above gives the desired divisibility unless $f_s=v-1$, $h_x=0$, $v$ is even, and $r$ is odd, which we now assume. This gives us also
  $e=0$.
  Because in this case the slice for $f_s$ also overlaps with the slice $\p*{t'} = \p*{2^{r-1}}$, it means we are in the case \ref{itm:ts},
  which we already considered.
  \item \label{itm:xs1}$2f_x+m'_x = r-v+f_s$ \\
  The slice for $f_x$ is
  \[\p*{2^{2f_x+3+m'_x}i + q'2^{2f_x+m'_x} + t'}  = \p*{2^{r-v+f_s+3}i + q'2^{r-v+f_s} + t'}  \]
  and the slice for $f_s$ is
  \[\p*{ 2^{r-v+f_s+1}s+ 2^{r-v+f_s} }. \]
  For those two slices to overlap it is necessary that
  \[t' = 2^{r-v+f_s+1}j \]
  for certain $j$ such that $o(j) \leq v-f_s-1 $ . For any $f'_x$ such that $2f'_x + m'_x < o(t')$, and $f'_s$ such that $r-v+f'_s =
  2f'_x + m'_x$ and $f'_s \leq v-e-1 \Leftrightarrow 2f'_x + m'_x < r-e$, the slices for $f'_x$ and $f'_s$ overlap. Let
  us denote $min(o(t')-m'_x, r-e-m'_x) =min(o(j)+f_s+1-m^*_x, v-e-m^*_x)$ by $M$, and let us count the size of those overlaps:
  \[\sum_{f'_x=0}^{\ceil*{\frac{M}{2}}-1}  2^{r - f'_x -k -1} f'_{s,f'_x} 2^{e}
  = 2^{e} \sum_{f'_x=0}^{\ceil*{\frac{M }{2}}-1}  (2f'_x+m^*_x) 2^{v - f'_x -1}\]
  \begin{align*}
= 2^{e}
\left(\frac{\left(2\ceil*{\frac{M }{2}}-2+m^*_x\right)2^{v - \ceil*{\frac{M}{2}}
-1}}{-\frac{1}{2}} - \frac{m^*_x 2^{v-1}}{-\frac{1}{2}}
- \frac{2\left(2^{v - \ceil*{\frac{M }{2}}
-1} - 2^{v-2} \right)}{\frac{1}{4}}  \right)
\end{align*}
(where we used the formula for sum of the arithmetic-geometric series)
\begin{align*}
= 2^{e+1}
\left( m^*_x2^{v-1} + 2^{v}
- \left(2\ceil*{\frac{M}{2}}-2+m^*_x\right)2^{v - \ceil*{\frac{M}{2}}-1}
- 2^{v -\ceil*{\frac{M}{2}}+1} \right)
\end{align*}
\begin{align*}
= 2^{e+1}
\left( \left(m^*_x + 2 \right)2^{v-1}
- \left(min(o(j)+f_s+1, v-e)+g -2\right)2^{v - \ceil*{\frac{M}{2}}-1}
- 2^{v - \ceil*{\frac{M}{2}}+1}
\right)
\end{align*}
\[ = 2^{e}
\left( \left(m^*_x + 2 \right)2^{v}
- \left(min(o(j)+f_s+1, v-e+g+2\right)2^{v - \ceil*{\frac{M}{2}}} \right) \]
for certain $g=0$ or $1$. The first element of this difference already has the required divisibility, so let
us focus on the second one: 
\[ (min(o(j)+f_s+1, v-e)+g+2) 2^{e+v - \ceil*{\frac{min(o(j)+f_s+1-m^*_x, v-e-m^*_x)}{2}}} \]\[
=  (min(o(j)+f_s+1, v-e)+g+2) 2^{e - \ceil*{\frac{-2v + min(o(j)+f_s+1-r+v-h_x, 2v-r-e-h_x)}{2}}} \]
\[  =  (min(o(j)+f_s+1, v-e)+g+2) 2^{e + \floor*{\frac{ max(r + v -o(j)-f_s-1+h_x, r+e+h_x)}{2}}}
\]
The above has the required divisibility unless $o(j)+f_s+1 = v$ (\ie $o(t') = r$), $e=h_x=0$, $v$ is even, and $r$ is odd, which we
now assume.
Yet due to that we get $t' = 0$, and therefore the slices $\p*{t'}$ and $\p*{2^{r-e}i}$ overlap. This means we land in the already-considered case \ref{itm:t0}.
\item $ 2f_x+3+m'_x > r-v+f_s > 2f_x+m'_x$\\
In this case we have
\[f_x = \ceil*{\frac{r-v+f_s-m'_x}{2}}-1. \]
Let us count size of the overlap between those two slices:
\[2^{r-f_x-k-1}f_s2^{e} =  f_s 2^{e+r-\ceil*{\frac{r-v+f_s-m'_x}{2}}+1-k-1} \]
\[=  f_s 2^{e+\floor*{\frac{r+v-f_s+h_x}{2}}}  \]
which gives the required divisibility, as $f_s \leq v-1$.
\item  $r-v+f_s \geq  2f_x+3+m'_x$ \\
For the slice
\[\p*{2^{2f_x+3+m'_x}i + q'2^{2f_x+m'_x} + t'}  \]
to overlap with the slice
\[\p*{ 2^{r-v+f_s+1}s+ 2^{r-v+f_s} } = \p*{ 2^{2f_x+4+m'_x+d}s+ 2^{2f_x+3+m'_x+d} } \]
(for certain $d\geq0$) it is necessary that
\[t' = 2^{2f_x+3+m'_x}j - q'2^{2f_x+m'_x}\]
in which case the slice for $f_x$ is just
\[\p*{2^{2f_x+3+m'_x}i} \]
and it is easy to notice that it overlaps with all slices for $f'_s$ such that $r-v+f'_s \geq 2f_x+3+m'_x$ and also with the slice
$\p*{2^{r-e}i}$. Fortunately, we already considered such a scenario in the case \ref{itm:x0}.
\end{enumerate}
\end{enumerate}
\end{enumerate}
\end{proof}

\begin{repCol} {osc} \emph{\textbf{(restated)}}

Let us work over a ring $\ZZ_{2^r}$. Let $P(x) = a_x x^2 + b_x x +c$ with $x$ being constrained to the domain $ x \in
\bigcup\limits_{j=0}^{2^q-1} \p*{l_x+  2^{r-q}j}$ for certain $l_x$ and $v \leq r$. Let $S$ be the following multiset:
\[
S = \left(\bigcup_{i=0}^{2^e-1} \bigcup_{f_s=0}^{v-1} f_s \bigcup_{s=0}^{2^{v-f_s-1}-1}  \p*{2^{r-e}i + 2^{r-v+f_s}(2s+1) } \right)
\cup \left(\bigcup_{i=0}^{2^e-1} (v+1) \p*{2^{r-e}i } \right),
\]
where $e \leq min(q, v)$. The number of elements of the intersection (understood as a ``multiplicative'' intersection) of the multiset $S$
and the image of $P$ is divisible by \[2^{e + min(q, v, \ceil*{\frac{r}{2}})} .
\]
\end{repCol}

\begin{proof}
   The proof is analogous to that of Corollary \ref{slices}, but using Lemma \ref{os} as the base.
\end{proof}

\section{Future work}
\pagestyle{plain}

As exemplified in Theorems \ref{generalb} and \ref{main}, the solution spaces have a specific structure, built 
around cosets of ideals. This describes a certain kind of symmetry of the solutions around the unit circle, which, as mentioned, is significant for properties of the $Z$-function and may have applications in computational complexity. We are certain that a lot of this structure is still left to be discovered,
especially extending Theorem \ref{main} to polynomials of any degree and rings over general composite numbers is desirable. Theorem
\ref{generalb} shows that for $\ZZ_m$ where $m = {p_1^{r_1}p_2^{r_2}\ldots p_{k}^{r_{k}}}$, there are symmetries in multiple
``dimensions'' (which result from the use of Chinese remaindering), one per
each $p_i$. We speculate about possible symmetries beyond those given by the decomposition into local rings.

A different angle on the problem, mentioned earlier in the introduction, is provided by restricting the arguments to subsets of the domain,
which elements are pairwise incongruent modulo a set prime ideal. This allows restriction of variables to, for example, $\p*{0,1}^n$.
This is studied very recently by Clark, Forrow and Schmitt \cite{cla14}, with the focus on lower-bounding the size of the gap between $0$,
and the second smallest solution number. Thanks to values we present in Table
\ref{tab:d2 front gap} we can experimentally see that results from \cite{cla14} are not optimal for composite $m$.
For example, for a single polynomial over $\ZZ_{2^q}$ having degree up to $2$ and with variables not restricted to any subset of the ring, they say that the
first gap is at least $2^{q(n-2)}$, where $n$ is the number of variables of the polynomial. For $q=4$ and $n=2$ it gives $1$ as minimum size of the gap,
whereas the gap is already $8$. For $n=3$, their bound is $16$, with the gap being $64$. Therefore, there is a room for improvement and future research on
establishing tighter bounds for this gap. Yet, encouraged by our experiments, we believe that even greater results may be obtained by focusing on the size of
the last gap instead (\ie the gap between the two largest possible numbers of solutions). If one would look at rings of size $r=2^q$ in Table \ref{tab:d2 back
gap}, the guess could be, that the size of this gap is $2^{(q-1)n}$ when $n\leq2$, and $2^{q(n-2)+2(q-1)}$ otherwise. Similar results look
plausible for $r=3^q$, yet for higher primes a large gap already shows for $n=1$, due to the degrees of the polynomials in the mentioned
table being just up to $2$. For rings that are not prime powers, the sizes of this gap seem (for small values of $n$) quite unintuitive,
which suggests that the formula governing that size is complicated. For larger values of $n$, and for all rings, the gap sizes seem to build on
previous values, via multiplication by the ring size. Arguably, the general formula for the last gap size may be simplest when the degree of
the polynomials is completely unbounded.

In Chapter \ref{ch:exp} we presented many concrete metrics of solution spaces for small rings and small polynomials (due to computational
limitations). From them we were able to notice certain properties, that may also be true in general cases. Some of them may be relatively
easy to prove, or are just quite interesting. We list three most intriguing ones, as following hypotheses:

\begin{enumerate}
  \item Polynomials of $n$ variables, of degree up to $2$, over finite fields of prime size $p$, can only have one of
  $2\floor*{\frac{n}{2}}+3$ numbers of solutions if $p=2$, and one of $4(\floor*{\frac{n}{2}}+1)$ when $p \geq 3$.
  \item Numbers of possible numbers of solutions of polynomials over rings $\ZZ_m$, of $n$ variables and degree up to $d$, can be bounded by
  a polynomial in $r, n$ and $d$. This is especially plausible when $m$ is prime or a prime power and $d$ is small.
  \item If $m = p_1^{r_1}, p_2^{r_2}, \ldots, p_k^{r_k}$, then the number of possible solution numbers of polynomials of $n$ variables, of
  degree up to $d$ over $\ZZ_m$, is a function of the numbers of possible solution numbers of polynomials of the same number of variables and
  degree, over each of $\ZZ_{p_i^{r_i}}$, $i < k$. For an example, see  $r = 2, 3, 6$ and $n=2, 3$ in Table \ref{tab:d2 slots used}.
\end{enumerate}

There are some more, potentially general, properties that we described in Chapter \ref{ch:exp}, and certainly many more than we did not
notice, or that require more experimental results to become noticeable. Yet, it certainly seems that this area is very rich in interesting,
open mathematical problems, which additionally, due to increasing understanding of polynomials behaviour over rings, have a potential to be
useful to complexity theory.

Despite the pathology of zero-divisors, we believe that the solution sets of polynomials modulo
composites should have a natural, attractive, and unifying theory.  Such work would seem relevant to
the prospects for progress in complexity lower bounds. We hope that this research promotes interest and strategies in expanding this
theory.

\ifdefined\phantomsection
\phantomsection  
\else
\fi
\addcontentsline{toc}{chapter}{Acknowledgments}
\section*{Acknowledgments}
We would like to thank Pete L.\ Clark for bringing  \cite{Dkat09} to our attention, which led us to \cite{mar75}, and for further comments on this draft.  We also thank him and Hui June Zhu and Todd Cochrane for consultation on the state of mathematical knowledge in the domain of this paper.


\bibliographystyle{plain} 

\ifdefined\phantomsection
\phantomsection  
\else
\fi
\addcontentsline{toc}{chapter}{References}
\bibliography{SuRe2014}

\end{document}